\documentclass{article}
\pdfoutput=1
\usepackage{fullpage}
\usepackage[all=normal,bibliography=tight]{savetrees}
\usepackage{graphicx}
\usepackage{xspace}
\usepackage{url}
\usepackage{amssymb}
\usepackage{latexsym}
\usepackage{amsthm}
\usepackage{amsmath}
\usepackage{amstext}
\usepackage{times}
\usepackage{amstext}
\usepackage{calc}
\usepackage{amsopn}
\usepackage{eucal}
\usepackage{latexsym}
\usepackage{subfig}
\usepackage{tabularx}
\usepackage{todonotes}
\usepackage{multirow}
\usepackage{enumerate}
\usepackage{skull}
\usepackage{tikz}
\usepackage{marvosym}

\newcommand{\cP}{\ensuremath{\mathcal{P}}}

\newcommand{\eps}{\varepsilon}

\newcommand{\ett}{\mathfrak{t}}
\newcommand{\numcc}{s}
\newcommand{\numbcc}{s}

\newcommand{\weight}{\omega}

\newcommand{\cutsize}{k}
\newcommand{\bterms}{T_b}
\newcommand{\instance}{\ensuremath{\mathcal{I}}}
\newcommand{\sol}{\mathtt{sol}}
\newcommand{\bigcc}{\mathtt{big}}

\newcommand{\Psifamily}{\mathbf{PSI}}

\newcommand{\newinst}[1]{\widehat{#1}}
\newcommand{\domain}{\mathrm{dom}}

\newcommand{\randfamily}{\ensuremath{\mathcal{F}}}

\newcommand{\defproblemu}[3]{
  \vspace{1mm}
\noindent\fbox{
  \begin{minipage}{0.94\textwidth}
  #1 \\
  {\bf{Input:}} #2  \\
  {\bf{Question:}} #3
  \end{minipage}
  }
  \vspace{1mm}
}
\newcommand{\defparproblemu}[4]{
  \vspace{1mm}
\noindent\fbox{
  \begin{minipage}{0.94\textwidth}
  \begin{tabular*}{\textwidth}{@{\extracolsep{\fill}}lr} #1 & {\bf{Parameter:}} #3 \\ \end{tabular*}
  {\bf{Input:}} #2  \\
  {\bf{Question:}} #4
  \end{minipage}
  }
  \vspace{1mm}
}
\newcommand{\defproblemoutput}[3]{
  \vspace{1mm}
\noindent\fbox{
  \begin{minipage}{0.94\textwidth}
  #1 \\
  {\bf{Input:}} #2  \\
  {\bf{Output:}} #3
  \end{minipage}
  }
  \vspace{1mm}
}

\newcommand{\edges}{\delta}

\newcommand{\rel}{{\mathcal{R}}}

\newcommand{\fvs}{{\sc Feedback Vertex Set}\xspace}

\newcommand{\oct}{{\sc Odd Cycle Transversal}\xspace}

\newcommand{\bsteinercut}{{\sc{Border Steiner Cut}}\xspace}

\newcommand{\bnmwcu}{{\sc{Border N-MWCU}}\xspace}
\newcommand{\multicut}{{\sc{Multicut}}\xspace}

\newcommand{\mwc}{{\sc{Multiway Cut}}\xspace}

\newcommand{\nmwcu}{{\sc{Node Multiway Cut-Uncut}}\xspace}
\newcommand{\nmwcushort}{{\sc{N-MWCU}}\xspace}

\newcommand{\kwaycut}{$k$-{\sc{Way Cut}}\xspace}
\newcommand{\steinercut}{{\sc{Steiner Cut}}\xspace}
\newcommand{\ulcgen}{{\sc{Unique Label Cover}}\xspace}
\newcommand{\eulc}{{\sc{Edge Unique Label Cover}}\xspace}
\newcommand{\beulc}{{\sc{Border E-ULC}}\xspace}
\newcommand{\ulc}{{\sc{Node Unique Label Cover}}\xspace}
\newcommand{\bulc}{{\sc{Border N-ULC}}\xspace}
\newcommand{\edgeulc}{{\sc{Edge Unique Label Cover}}\xspace}
\newcommand{\mclique}{{\sc{Multicolored Clique}}\xspace}
\newcommand{\wulc}{{\sc{W-E-ULC}}\xspace}
\newcommand{\bwulc}{{\sc{Border W-E-ULC}}\xspace}

\newcommand{\ugc}{{\sc{Unique Games Conjecture}}\xspace}

\newtheorem{theorem}{Theorem}[section]
\newtheorem{lemma}[theorem]{Lemma}

\theoremstyle{definition}
\newtheorem{definition}[theorem]{Definition}

\newtheorem{step}{Step}[section]
\newtheorem{algrule}{Rule}[section]

\newcommand{\TheTitle}{Designing FPT algorithms for cut problems using randomized contractions}

\title{{\TheTitle}\thanks{A preliminary version of this work was presented at FOCS 2012~\cite{ChitnisCHPP12}.}}
\date{}
  \author{
  Rajesh Chitnis
  \thanks{
	Department of Computer Science, University of Maryland at College Park, USA
	(\texttt{rchitnis@cs.umd.edu}).
  Supported in part by NSF CAREER award 1053605, ONR YIP award N000141110662, DARPA/AFRL award FA8650-11-1-7162 and a University of Maryland Research and
Scholarship Award (RASA).
  }
  \and
  Marek Cygan %\thanks{Partially supported by NCN grant N206567140, Foundation for Polish Science.}
  \thanks{
    Institute of Informatics, University of Warsaw, Poland
      (\texttt{cygan@mimuw.edu.pl}).
  Supported in part by NSF CAREER award 1053605, ONR YIP award N000141110662, DARPA/AFRL award FA8650-11-1-7162, University of Maryland Research and
Scholarship Award (RASA)
  NCN grant N206567140, Foundation for Polish Science.
  }
  \and
  MohammadTaghi Hajiaghayi
  \thanks{
	Department of Computer Science, University of Maryland at College Park, USA
	(\texttt{hajiagha@cs.umd.edu}).
  Supported in part by NSF CAREER award 1053605, ONR YIP award N000141110662, DARPA/AFRL award FA8650-11-1-7162 and a University of Maryland Research and
Scholarship Award (RASA).
  }
  \and
  Marcin Pilipczuk %\thanks{Partially supported by NCN grant N206491038 and Foundation for Polish Science.}
  \thanks{
    Institute of Informatics, University of Warsaw, Poland
      (\texttt{malcin@mimuw.edu.pl}).
    Partially supported by NCN grant N206567140 and Foundation for Polish Science.
  }
  \and
  Micha\l{} Pilipczuk
  \thanks{
    Institute of Informatics, University of Warsaw, Poland, (\texttt{michal.pilipczuk@mimuw.edu.pl}).
The research of this author, leading to these results, was done when the author was affiliated with the University of Bergen, Norway, and
 has received funding from the European Research Council under the European Union's Seventh Framework Programme (FP/2007-2013) / ERC Grant Agreement n. 267959.
  }
  }

\begin{document}
\maketitle

\begin{abstract}
We introduce a new technique for designing fixed-parameter algorithms for cut problems, called {\emph{randomized
contractions}}. We apply our framework to obtain the first FPT algorithm for the \ulcgen problem and new FPT algorithms with
exponential speed up for the \steinercut and \nmwcu problems. More precisely, we show the following:
\begin{itemize}
\item We prove that the parameterized version of the \ulcgen problem, which is the base of the \ugc, can be solved in
    $2^{O(k^2\log |\Sigma|)}n^4\log n$ deterministic time (even in the stronger, vertex-deletion variant) where $k$ is the
    number of unsatisfied edges and $|\Sigma|$ is the size of the alphabet. As a consequence, we show that one can in
    polynomial time solve instances of {\sc Unique Games} where the number of edges allowed not to be satisfied
    is upper bounded by $O(\sqrt{\log n})$ to optimality, which improves over the trivial $O(1)$ upper bound.
   %TODO HACK linebreak
\item We prove that the \steinercut problem can be solved in $2^{O(k^2\log k)}n^4\log n$ deterministic time and \linebreak
    $\tilde{O}(2^{O(k^2\log k)}n^2)$ randomized time where $k$ is the size of the cutset. This result improves the double
    exponential running time of the recent work of Kawarabayashi and Thorup (FOCS'11).
\item We show how to combine considering `cut' and `uncut' constraints at the same time. More precisely, we define a
    robust problem \nmwcu that can serve as an abstraction of introducing uncut constraints, and show that it admits an
    algorithm running in $2^{O(k^2\log k)}n^4\log n$ deterministic time where $k$ is the size of the cutset. To the best
    of our knowledge, the only known way of tackling uncut constraints was via the approach of Marx, O'Sullivan and Razgon
    (STACS'10, ACM Trans. Alg. 2013), which yields algorithms with double exponential running time.
    %\tudu{mention weights}
\end{itemize}
An interesting aspect of our algorithms is that they can handle positive real weights.
\end{abstract}

\section{Introduction}

\noindent{\emph{Graph cut problems}} is a class of problems where, given a graph, one is asked to find a {\emph{cutset}} of minimum size whose removal makes the graph satisfy a global separation property. The motivation of studying graph cut problems stems from the fundamental minimum cut problem, where the goal is to separate two terminals from each other by removing the least possible number of vertices or edges, depending on the variant. Even though the minimum cut problem can be solved in polynomial time, many of its natural generalizations become NP-hard. Moreover, many problems, whose classical definitions do not resemble cut formulations, after choosing an appropriate combinatorial viewpoint show deep links with finding minimum separators; the most important examples are \fvs and \oct. 

Therefore, circumventing NP-hardness of fundamental graph cut problems, like \mwc (given a graph with a set of terminals, separate the terminals from each other using minimum size cutset) or \multicut (given a graph with a set of terminal pairs, separate terminals in the pairs using minimum size cutset), became an important algorithmic challenge. It is then no surprise that graph cut problems were studied intensively from the point of view of approximation; cf.~\cite{approx-8,approx-9,approx-5,approx-11,approx-12,approx-6,approx-7,approx-1,approx-2,approx-3,approx-10}.

In this paper we address a different paradigm of tackling NP-hard problems, that is, {\emph{fixed-parameter tractability}} (FPT). Recall that in the parameterized complexity setting an instance of the problem comes with an additional integer $k$, called the {\emph{parameter}}, which intuitively measures the hardness of the instance. The goal is to devise an algorithm solving the problem with running time of form $f(k)n^c$, where $f$ is some computable function and $c$ is a fixed constant. In other words, for every fixed parameter the algorithm has to work in polynomial time, where the degree of the polynomial is independent of the parameter. Algorithms with such a running time guarantee are called {\emph{fixed-parameter}} algorithms, and if a problem admits one, then we say that it is {\emph{fixed-parameter tractable}}. For a more detailed introduction to fixed-parameter tractability we address an interested reader to the recent monographs~\cite{our-book,DF13}.

Graph separation problems in the context of parameterized complexity were probably first considered in the seminal work of Marx \cite{marx:multiway-journal}. Marx established fixed-parameterized complexity of \mwc parameterized by the size of the cutset and \multicut parameterized by the size of the cutset plus the number of terminal pairs. The main tool introduced by Marx is the notion of an \emph{important separator}, which later turned out to be the core ingredient of parameterized algorithms
for, e.g., \textsc{Directed Feedback Vertex Set}~\cite{directed-fvs} or \textsc{Almost 2-SAT}~\cite{a2sat}. In the last decade, the graph separation problems become one of the most intensively studied subareas of parameterized complexity, leading to the development
of various interesting techniques, such as shadow removal~\cite{marx:multicut} and its generalizations to directed graphs~\cite{dir-mwc}, treewidth reduction~\cite{small-treewidth-marx}, and branching guided by an LP or ($k$-)submodular CSP relaxation~\cite{our-lp,magnus}.

We introduce a new technique, called {\emph{randomized contractions}}, of constructing fixed-parameter algorithms for graph cut problems. 
In this introduction, we first give an overview of this technique and our results, and then provide a discussion and comparison with other known techniques.

\subsection{Our techniques}

On high level, the technique of \emph{randomized contractions} is based on a WIN/WIN approach, introduced by Kawarabayashi and Thorup~\cite{kt11}, and also used by Grohe and Marx in their algorithm to test the topological minor relation~\cite{gm-topminor}.
The WIN/WIN approach can be described as follows: either we find a well-balanced separation of small order, whose one side can be simplified by a recursive call, or the graph admits a highly-connected structure, which can be used to identify the solution.
The main novelty of this paper is the way these steps are executed: we show that a well-balanced separation can be easily and efficiently found using the color coding technique introduced by Alon et al.~\cite{color-coding},
and the color coding technique also greatly helps in exhibiting the solution in the presence of the high-connected structure.

Recall that the main idea of the color coding technique, originally introduced to solve some special cases of the \textsc{Subgraph Isomorphism} problem,
is to color the graph at random and ensure that with high probability the solution gets sufficiently highlighted to be recognizable quickly. It has now become a classical tool in the parameterized complexity toolbox.
At heart of our results lies an observation that it can be also used to highlight either a well-balanced separation or a structure of the solution in a highly-connected graph.
Our usage of the color coding technique, especially in the search for a well-balanced separation, resembles the algorithm of Karger~\cite{karger-min-cut} that finds a minimum cut in a graph in near-linear time by contracting random edges; this inspiration gave the name to our technique.

Although the intuition behind color coding is of probabilistic nature, the algorithms obtained using this approach can be derandomized using the technique of {\emph{splitters}} of Naor et al. \cite{naor-schulman-srinivasan-derandom}. In fact, we find it more convenient to present our algorithms already in the derandomized version, so in spite of the name of the technique there will be no randomization at all; instead we use the following abstraction:

\begin{lemma}\label{lem:random}
Given a set $U$ of size $n$, and integers $0 \leq a,b \leq n$, one
can in time $2^{O(\min(a,b) \log (a+b))} n \log n$
construct a family $\randfamily$ of at most $2^{O(\min(a,b) \log (a+b))} \log n$
subsets of $U$, such that the following holds:
for any sets $A,B \subseteq U$, $A \cap B = \emptyset$, $|A|\leq a$, $|B|\leq b$,
there exists a set $S \in \randfamily$ with $A \subseteq S$ and $B \cap S = \emptyset$.
\end{lemma}

Our approach is most natural for edge-deletion problems; however, we can also extend it to node-deletion variants.
For the node deletion problems, however, the situation is more
complicated and we need to define two kinds of separations. % - good node separations and flower separations.
Only when the graph does not have both kinds of separations, we get enough structure to
solve the problem with other methods.
Moreover, one needs to be much more careful in this final case, as we obtain much weaker structural properties of the graph.

\subsection{Our results}

We use the technique of randomized contractions to provide the first fixed-parameter algorithm solving an important problem in parameterized complexity,
and moreover we show how our approach can be applied to reduce the time complexity of the best known algorithms from double exponential to single exponential for some problems already known to be FPT.

\subsubsection{Unique Label Cover}
In the \ulcgen problem we are given an undirected graph $G$, where each edge $uv=e\in E(G)$ is
associated with a permutation $\psi_{e,u}$ of a constant size alphabet $\Sigma$. The goal is to construct a labeling $\Psi: V(G)\to \Sigma$ 
maximizing the number of satisfied edge constraints, that is, edges for which $(\Psi(u),\Psi(v)) \in \psi_{uv,u}$ holds.
At the first glance \ulcgen does not seem related to the previously mentioned cut problems, 
however it is not hard to show that the node deletion version of \ulcgen is
a generalization of {\sc Group Feedback Vertex Set} problem~\cite{guillemot-2011},
and hence of {\sc Odd Cycle Transversal}, {\sc Feedback Vertex Set},
as well as \mwc.

The optimization version of \ulcgen is the subject of the 
very extensively studied \ugc, proposed by Khot~\cite{ugc} in 2002, which is used as a hardness assumption for
showing several tight inapproximability results.
The \ugc states that for every sufficiently small $\varepsilon, \delta > 0$,
there exists an alphabet size $|\Sigma|(\varepsilon,\delta)$, such that given an instance $(G,\Sigma,(\psi_{e,v})_{e \in E(G),v \in e})$ it is NP-hard to distinguish between the cases $|OPT|\le \delta|E(G)| $ and $|OPT| \ge (1-\eps)|E(G)|$. In 2010 Arora et al.~\cite{arora-ug} presented a breakthrough subexponential time algorithm, which in $2^{O(|\Sigma|n^\varepsilon)}$ 
running time satisfies $(1-\eps)|E(G)|$ edge constraints, assuming the given instance satisfies $|OPT| \geq (1-\eps^c) |E(G)|$.
We refer the reader to a recent survey of Khot~\cite{khot-survey} for more detailed discussion on the \ugc.
%For a description of other results related to \ugc we refer the reader to a recent survey of Khot~\cite{khot-survey}.

Since all the edge constraints are permutations, fixing a label for one vertex
gives only one possibility for each of its neighbors, assuming we want to satisfy all the edges.
For this reason we can verify in polynomial time, whether $OPT = |E(G)|$.
In this paper we show that we can efficiently solve the \ulcgen problem,
assuming almost all the edges are to be satisfied.
In particular, we design a fixed parameter algorithm for \ulc, which is a generalization of \eulc.

\defproblemu{\ulc}{An undirected graph $G$, a finite alphabet $\Sigma$ of size $s$,
 an integer $\cutsize$, and for each edge $e \in E(G)$ and each of its endpoints $v$ a permutation $\psi_{e,v}$ of $\Sigma$, such that
if $e = uv$ then $\psi_{e,u} = \psi_{e,v}^{-1}$.}{Does there exist
a set $X \subseteq V(G)$ of size at most $\cutsize$ and a function
$\Psi: V(G) \setminus X \to \Sigma$ such that
for any $uv \in E(G \setminus X)$ we have $(\Psi(u),\Psi(v)) \in \psi_{uv,u}$?}

%A brief sketch of the proof of the following theorem is given in Section~\ref{sec:ulc};
%the details are in Appendix~\ref{sec:full-ulc}.

\begin{theorem}
\label{thm:ulc-main}
There is an $O(2^{O(\cutsize^2 \log s)} n^4 \log n)$ time algorithm
solving \ulc.
% and consequently \eulc.
\end{theorem}

To justify our parameterization, we would like to note that there is a long line of polynomial time approximation algorithms 
designed for instances of \ulcgen, with currently best by Charikar et al.~\cite{charikar-ug}, working under the assumption $|OPT| \ge (1-\varepsilon)|E(G)|$, 
and where the alphabet is of constant size.
Therefore, it is reasonable to assume that only a small number of constraints is not going to be satisfied.
Our results imply that one can in polynomial time verify whether it is possible 
to satisfy $|E(G)|-O(\sqrt{\log n})$ constraints; consequently, we extend the range of instances that can be solved optimally in polynomial time.

Finally, we show that the dependence on the alphabet size in Theorem~\ref{thm:ulc-main} is probably necessary, since the problem parameterized by the cutsize only is $W[1]$-hard. Hence, the existence of an algorithm parameterized by the cutsize only would cause $FPT=W[1]$, which is considered implausible. For a more detailed introduction to the hierarchy of parameterized problems and consequences of its collapse, we refer to the books of Downey and Fellows \cite{downey-fellows:book} or of Flum and Grohe \cite{grohe:book}. We consider this result an interesting counterposition of the parameterized status of {\sc Group Feedback Vertex Set}~\cite{our-gfvs},
which is FPT even when the group size is not a parameter.

\begin{theorem}\label{thm:ulc-hardness}
The \eulc problem, and consequently \ulc, is $W[1]$-hard when parameterized by $\cutsize$ only.
\end{theorem}

\medskip

\subsubsection{Steiner Cut}
Next, we address a robust generalization of both \kwaycut and \mwc problems, namely the \steinercut problem.

\defproblemu{\steinercut}{A graph $G$, a set of terminals $T \subseteq V(G)$,
and integers $\numcc$ and $\cutsize$.}{Does there exist a set $X$ of at most $\cutsize$
edges of $G$, such that in $G \setminus X$ at least $\numcc$ connected components
contain at least one terminal?}

Using our technique we present an FPT algorithm working in $O(2^{O(\cutsize^2 \log \cutsize)} n^4 \log n)$,
where the polynomial factor can be improved to $\tilde{O}(n^2)$ at the cost of our algorithm being randomized.
These results improve the double exponential time complexity of the recent algorithm of Kawarabayashi and Thorup~\cite{kt11}\footnote{In~\cite{kt11} the authors
solve the \kwaycut problem, however a straightforward generalization of their algorithm solves the \steinercut problem as well.}.
%A brief sketch of the proof of the following theorem is given in
%Section~\ref{sec:steiner-cut} and the details are in
%Appendix~\ref{sec:full-steiner}.

\begin{theorem}
\label{thm:steiner-main}
There is a deterministic $O(2^{O(\cutsize^2 \log \cutsize)} n^4 \log n)$ and randomized $\tilde{O}(2^{O(\cutsize^2 \log s)} n^2)$ running time algorithm solving \steinercut.
\end{theorem}

\subsubsection{Connectivity constraints} We define the following problem as an abstraction of introducing ``cut'' and ``uncut'' constraints at the same time.

\defproblemu{\nmwcu (\nmwcushort)}{A graph $G$ together with a set of terminals $T\subseteq V(G)$, an equivalence relation $\rel$ on the set $T$, and an integer $\cutsize$.}{Does there exist a set $X \subseteq V(G) \setminus T$ of at most $\cutsize$ nonterminals such that for any $u,v \in T$, the vertices $u$ and $v$ belong to the same connected component of $G \setminus X$ if and only if $(u,v) \in \rel$?}

Fixed-parameter tractability of this problem can be derived from the framework of Marx, Razgon, and O'Sullivan~\cite{small-treewidth-marx}, complemented with a reduction of the number of equivalence classes of $\rel$ in flavour of the reduction for \mwc of Razgon~\cite{razgon:mwc-k2-terms}. However, the dependence on $k$ of the running time is double exponential. Using our framework we show the following.

\begin{theorem}
\label{thm:cutuncut-main}
There is an $O(2^{O(\cutsize^2 \log \cutsize)} n^4 \log n)$ time algorithm solving \nmwcu.
\end{theorem}

\subsubsection{Weights}

As mentioned in the abstract, our approach generalizes well to the weighted setting, which is not the
case for many other techniques in parameterized complexity such as important separators.
As the level of technical details in all our algorithms is high,
we prove Theorems~\ref{thm:ulc-main},
\ref{thm:steiner-main} and \ref{thm:cutuncut-main}
in the unweighted case (i.e., as they are stated in the introduction).
Then, in Section~\ref{sec:weights}, we discuss extensions to the weighted setting.

\subsubsection{Subsequent usages and extensions}

We would like to mention here a few applications and extensions of our techniques, developed after the extended abstract of our work has been published~\cite{ChitnisCHPP12}.

First, the technique turned out to be useful in a number of other problems.
Lokshtanov during Dagstuhl Seminar 14071 (February 2014) noticed that our technique immediately gives fixed-parameter tractability the \textsc{Vector Connectivity} problem, parameterized by the cutset size only, solving an open problem posed by Milani\v{c};
we refer to the recent work of Kratsch and Sorge~\cite{stefan-vector-conn} for problem definition and a discussion of recent developments. 
Bringman et al.~\cite{ejvl-steinermulticut}, in their study of different parameterizations of \textsc{Steiner Multicut}, noticed that one can use randomized contractions to obtain an FPT algorithm for one of their most natural parameterizations.

Finally, a subset of the current authors together with Lokshtanov and Saurabh~\cite{bisection} developed a way to replace the recursive scheme in our approach with a static tree decomposition,
where every adhesion has bounded size and every bag has properties similar to those dubbed ``highly-connected'' in the description above.
This improvement has led to an FPT algorithm for \textsc{Minimum Bisection}.

\subsection{Discussion of related work}

\noindent\paragraph{Important separators} Perhaps the most fruitful consequence of the early work of Marx~\cite{marx:multiway-journal}
was the introduction of the concept of an {\emph{important separator}}.
Important separators proved to be a robust tool that enable us to capture the bounded-in-parameter character of the family of reasonable cutsets. They also can be naturally extended to the directed setting.
This basic technique has found numerous applications \cite{chen-nmc,directed-fvs,dir-mwc,sfvs,Guillemot11a,clustering-daniels,marx:multiway-journal,a2sat}.

The important separators technique is based on greedy arguments, which unfortunately makes this approach work only in restricted
settings.
Consider, for instance, the ``uncut'' constraint present in the \nmwcu problem, i.e.,
we look for a cutset that separates some pairs of terminals, but is required {\emph{not}} to separate some other pairs.
Any greedy choice of the farthest possible cutset, which is precisely the idea behind the notion of an important separator, can spoil the delicate requirements of the existence of some paths.

Furthermore, the proof of the core property of important separators --- the bound on their number expresses in the parameter only --- relies on
amortization by the increase of the cost of the separation, which makes the argument work only in the unweighted setting (or with small integer weights).
It is unclear whether this notion can lead to parameterized algorithms in the setting with arbitrary (positive) real weights.

\paragraph{Shadow removal}
The fixed-parameter tractability of \multicut parameterized by the cutsize only, after resisting attacks as a long-standing open problem,
was finally resolved in 2011 by Marx and Razgon~\cite{marx:multicut} and, independently, by Bousquet et al. \cite{thomasse:multicut}.
The most important contribution of the work of Marx and Razgon~\cite{marx:multicut} was the introduction of the shadow removal technique,
and intricate blend of the important separators with the color coding technique. 
In some problems (e.g., \multicut) one can argue that a greedy step, in the sense of important separators, is possible, but
one cannot apply it directly, as one does not know one side of the separation. The color coding technique is used to highlight
possible application places.

A subset of the current authors together with Marx~\cite{dir-mwc} showed that, after a delicate transfer of the shadow removal technique to directed graphs,
it almost immediately yields an FPT algorithm for \mwc in directed graphs.
Further usages include~\cite{dsfvs,multicut-in-dags,clustering-daniels}.

On high level, one could say that the shadow removal technique extends the applicability of important separators,
and is used to obtain additional properties of the cutset we are looking for. 
In some sense it is perpendicular to randomized contractions:
On one hand, its applicability is limited due to the need of some greedy reasoning to apply important separators. 
On the other hand, shadow removal seems crucial for most of its applications, especially in directed graphs; in particular, we are unable to handle these applications using randomized contractions.

\paragraph{Treewidth reduction}
The treewidth reduction technique, developed by Marx, O'Sullivan, and Razgon~\cite{small-treewidth-marx}, 
is probably the closest, in terms of the scope of applicability, to randomized contractions.
It essentially states that in an undirected graph $G$ with two terminals $s$ and $t$, all inclusion-wise minimal cuts
between $s$ and $t$ of size at most $k$ live in a part of $G$ of treewidth bounded exponentially in $k$. 
The result is robust in the sense that it allows to include a bounded number of terminal pairs to separate.

This clean structural result allows to bypass the limitations of important separators: similarly as randomized contractions,
it does not require any greedy step (thus handling, e.g., the ``uncut'' contraints) and it easily handles weighted variants.
The most natural problems that can be handled using the treewidth reduction technique, like \nmwcu, usually can be also solved using randomized contractions.

However, we note that there are two shortfalls of treewidth reduction, as compared to randomized contractions. First, it inherently leads
to double-exponential dependency on the parameter: the bound on the treewidth of the ``small cut part'' of $G$ is necessarily exponential,
and on top of that one uses a dynamic programming algorithm whose running time almost always depends at least exponentially on this treewidth.
Second, it requires to specify a bounded number of terminals to start with; hence it is unclear how to use it, e.g., for the \ulcgen problem.

It should be noted that algorithms using the treewidth reduction technique, despite their double-exponential dependency
on the parameter, are usually conceptually much simpler and cleaner than their counterparts obtained using randomized contractions.
This is particularly visible in the case of \textsc{Steiner Multicut}~\cite{ejvl-steinermulticut}.

\paragraph{Branching guided by LP relaxations}
A subset of the current authors, together with Wojtaszczyk~\cite{our-lp}, showed that one can use very strong structural
properties of the LP relaxations of \textsc{Vertex Cover} and \mwc to develop efficient branching algorithms for these problems,
parameterized by the gap above the optimum value of the LP relaxation.
Narayanaswamy et al.~\cite{saket-lp} observed that, in the case of \textsc{Vertex Cover}, one can apply known reduction rules
to improve running time even further. In this manner, quite unexpectedly, they obtained an improvement upon the classic $O(3^k nm)$-time
algorithm for \textsc{Odd Cycle Transversal}~\cite{oct}.
The currently fastest algorithm in this line is due to Lokshtanov et al.~\cite{saket-lp2}.

Wahlstr\"{o}m~\cite{magnus} observed that instead of the very inflexible LP relaxations, one could use ($k$-)submodular relaxation
to a Valued CSP problem, obtaining surprisingly efficient algorithms for a number of problems, including
$|\Sigma|^{2k} n^{O(1)}$-time algorithm \ulcgen.
Subsequently, the dependency on the input size has been improved to linear~\cite{magnus2}. We note that these works~\cite{magnus2,magnus} are subsequent to our work.

The above line of research gave a number of surprisingly efficient algorithms: the running time is
usually single-exponential in the cutsize, and the techniques of~\cite{magnus2} usually give good dependency on the input size. 
On the other hand, to apply them one needs to find a relaxation with strong properties (a $k$-submodular one in most cases),
which is unknown, e.g., for \multicut or \nmwcu.

\paragraph{Limitations of randomized contractions}
This discussion exhibits three limitations of the randomized contractions technique.

First, we do not know how to apply the randomized contractions technique to the \multicut problem without any bound
on the number of terminals; recall that the algorithm of Marx and Razgon~\cite{marx:multicut} makes use of important separators
and shadow removal.
This is mostly due to the fact that our technique, in the recursive step, needs a bound on the
number of possible behaviors on a small cutset, similarly as it is needed to develop a dynamic programming 
algorithm on graphs of bounded treewidth, or to apply the protrusion machinery~\cite{fomin-focs09, fomin-soda12}.
Note that \multicut, in the edge-deletion setting, is NP-hard on trees~\cite{multicut-in-trees}.

Second, our technique is inherently tailored to undirected graphs, whereas both important separators
and shadow removal are well-understood on directed graphs as well. 
It is an interesting question whether one can obtain a convenient structural description of bounded size cuts in directed
graphs, in the spirit of the treewidth reduction technique for undirected graphs~\cite{small-treewidth-marx}. 
A very recent work of one of the authors and Wahlstr\"{o}m~\cite{dirmc-lb} showed that one cannot hope
for bounded treewidth of the underlying undirected graph, but \emph{directed} treewidth can be bounded.
However, the latter is probably insufficient for many algorithmic applications, as~\cite{dirmc-lb} proved also
$W[1]$-hardness of \multicut in directed graphs with only four terminal pairs.

Third, in the case of \ulcgen our technique gives suboptimal running time, both
in terms of the dependency on the parameter and the input size~\cite{magnus2}.
In our approach, we require that $q$, the minimum size of a side in a well-balanced separation, is greater than
the number of possible behaviors of the solution on a cutset (which is usually exponential in the size of the cutset),
and subsequent applications of the color coding technique introduce term $q^k = 2^{\Omega(k^2)}$ to the running time bound.
Furthermore, the recursion scheme, together with multiple needs of finding small cuts, blows up the polynomial factor to quartic.
We conjecture that in the other studied problems, as well as in the case of \textsc{Minumum Bisection}~\cite{bisection},
it is possible to decrease both factors of the running time bound significantly, but possibly using very different techniques.

\subsection{Organization of the paper}

We start with an informal illustration of our technique in Section~\ref{sec:illustration},
using the example of the edge-deletion version of the \ulcgen problem.
We follow the illustration with some formal generic definitions and preliminary
results in Section~\ref{sec:prelims}.
In Sections~\ref{sec:full-ulc}, \ref{sec:full-steiner} and \ref{sec:full-uncut}
we consider~\ulc, \steinercut and \nmwcushort, respectively, proving Theorems~\ref{thm:ulc-main},
\ref{thm:steiner-main} and \ref{thm:cutuncut-main}.
Section~\ref{sec:lower} contains a reduction showing $W[1]$-hardness of the \eulc problem,
when parameterized by the size of the cutset only.
Finally, in Section~\ref{sec:weights} we discuss extensions of our framework to weighted graphs.
As the introduction included an extensive discussion of related work and possible extensions, 
we skip the conclusions section.

\section{Illustration}\label{sec:illustration}

In this section we present the outline of the technique, illustrating it with a running example of the \eulc problem.
Since this section serves as an introduction and illustration, the arguments here are mostly informal. Note that a more general problem, \ulc, is formally proven to be fixed-parameter tractable in Section~\ref{sec:full-ulc}.

\defparproblemu{\eulc}{An undirected (multi)graph $G$, a finite alphabet $\Sigma$ of size $s$,
 an integer $\cutsize$, and for each edge $e \in E(G)$ and each of its endpoints $v$ a permutation $\psi_{e,v}$ of $\Sigma$, such that
if $e = uv$ then $\psi_{e,u} = \psi_{e,v}^{-1}$.}{$\cutsize+s$}{Does there exist
a set $X \subseteq E(G)$ of size at most $\cutsize$ and a function
$\Psi: V(G) \to \Sigma$ such that
for any $uv \in E(G) \setminus X$ we have $(\Psi(u),\Psi(v)) \in \psi_{uv,u}$?}

The permutations $\psi_{e,u}$ are called {\em{constraints}}, the function
$\Psi$ is called a {\em{labeling}} and the set $X$ is the {\em{deletion set}}.

As we consider the edge-deletion version, we use edge
cuts throughout this section. However, as our general framework can be also applied to node-deletion problems, we comment
along the description where additional argumentation is needed in the node-deletion setting.

We assume that the graph given in the input is connected, as it is easy to reduce the problem to considering each connected
component separately. This is true for all the considered problems. Connectivity of the graph will be maintained
during the whole course of the algorithm. Note that this means that the graph after excluding $X$ can have at most $k+1$ connected
components.

The algorithm, at the very high level, closely follows the approach of Kawarabayashi and Thorup~\cite{kt11}.
We distinguish two cases: either the graph has a somewhat balanced separator, or it is highly connected in the following sense:
any cut of bounded size can separate only a very small part of the graph. More formally, we use the following notion of {\em{good edge separation}}.

\begin{definition}
Let $G$ be a connected graph. A partition $(V_1,V_2)$ of $V(G)$ is called a $(q,k)$-{\emph{good edge separation}}, if
%{\emph{(i)}} $|V_1|,|V_2|>q$; {\emph{(ii)}} $|\edges(V_1,V_2)|\leq k$; {\emph{(iii)}} $G[V_1]$ and $G[V_2]$ are connected.
\begin{itemize}
\item $|V_1|,|V_2|>q$;
\item $|\edges(V_1,V_2)|\leq k$, where $\edges(V_1,V_2)$ is the set of edges with one endpoint in $V_1$ and second endpoint in $V_2$;
\item $G[V_1]$ and $G[V_2]$ are connected.
\end{itemize}
\end{definition}

In the first phase of the algorithm, named {\emph{recursive understanding}}, we iteratively find a good edge separation and
reduce one of its sides up to the size bounded by a function of the parameter. We use the lower bound on the number of
vertices of either side to ensure that we indeed make some simplification. The applied reduction step needs introducing a more
general problem, in which, intuitively, we have to prepare for every possible behavior on a bounded number of distinguished
vertices of the graph, called {\emph{border terminals}}.

When no good edge separation can be found, by Menger's theorem we know that between every two disjoint connected subgraphs of
size larger than $q$ we can find $k+1$ edge-disjoint paths. Then we proceed to the second phase, named {\emph{high
connectivity phase}}, where we exploit this highly connected structure to identify the solution.

While the structure of the first phase is the same as in~\cite{kt11}, our work differs in two important aspects.
First, using Lemma~\ref{lem:random} we show a simple efficient way to find a balanced separator to recurse.
Second, we show a general methodology how to apply Lemma~\ref{lem:random} again for the second, high-connectivity phase,
to highlight important parts of the graph and find the solution efficiently.

\subsection{Recursive understanding}

First, we show a simple way how to find a good edge separation in the graph.
A full proof of the following lemma can be found in Section~\ref{sec:prelims}.

\begin{lemma}\label{lem:find-separation-ill}
There exists a deterministic algorithm that, given an undirected, connected graph $G$ on $n$ vertices along with integers $q$
and $k$, in time $O(2^{O(\min(q,k)\log(q+k))}n^3\log n)$ either finds a $(q,k)$-good edge separation, or correctly concludes
that no such separation exists.
\end{lemma}
\begin{proof}
Consider a family $\randfamily$ obtained via Lemma \ref{lem:random}
for the universe $U=E(G)$ and integers $a = 2q$ and $b=k$.
Let $(V_1,V_2)$ be a good separation in $G$ and, for $i=1,2$, let $T_i$ be any tree
with $q$ edges that is a subgraph of $G[V_i]$.
By the properties of $\randfamily$, there exists $S \in \randfamily$ such that
$E(T_1),E(T_2) \subseteq S$, but $S \cap E(V_1,V_2) = \emptyset$.
Consider a (multi)graph $G_S$ obtained from $G$ by contracting the edges of $S$ (we preserve multiple edges in the contraction process),
and let $v \in V(G_S)$ be called {\em{heavy}} if more than $q$ vertices of $G$
were contracted onto it. It is easy to see that the good separation $(V_1,V_2)$
corresponds to a cut between two heavy vertices in $G_S$ of size at most $k$;
moreover, any such cut yields a good separation in $G$.
Such a desired cut can be found in polynomial time;
the claimed running time follows if we first apply the sparsifying technique of
Nagamochi and Ibaraki~\cite{nagamochi-ibaraki} and then the classical algorithm
of Ford and Fulkerson to find a minimum cut between each pair of heavy vertices.
We note that, using instead a variant of the classical Karger's algorithm for minimum cut~\cite{karger-min-cut}, the problem can be solved in
$\tilde{O}(2^{O(\min(q,k)\log(q+k))}(|V(G)|+|E(G)|))$ time at the cost of being randomized.
\end{proof}

The general methodology of the proof of Lemma~\ref{lem:find-separation-ill}: to use color coding to pick a set of ``undeletable'' edges that is disjoint 
from the solution, but highlights it, is the main engine of our work. We will see this idea exploited much
more deeply in the high connectivity phase.

Having found a good edge separation we can proceed to simplification of one of the sides. To this end, following~\cite{kt11}, we consider a more
general problem, where the input graph is equipped with a set of {\emph{border terminals}} $\bterms$, whose number is bounded
by a function of the budget for edge deletions. Intuitively, each considered instance of the border problem corresponds to
solving some small part of the graph, which can be adjacent to the remaining part only via a small boundary --- the border
terminals. Our goal in the border version is, for every fixed behavior on the border terminals, to find some minimum size
solution or to conclude that the size of the minimum solution exceeds the given budget. Of course, the definition of behavior is
problem-dependent; therefore, we present this concept on the example of the \eulc problem.

Luckily, the definition in this case is natural and simple. The behavior on the border terminals, whose number
will be bounded by $4k$, is defined as a function $\Psi_b: \bterms \to \Sigma$ expressing the labeling we expect
on the border terminals.
More formally, for an instance of the border problem $\instance_b = (G,\Sigma,\cutsize,(\psi_{e,v})_{e \in E(G), v \in e})$
with border terminals $\bterms$, by
$\mathbb{P}(\instance_b)$ we denote the set of all possible functions $\Psi_b: \bterms \to \Sigma$.
For any $\Psi_b \in \mathbb{P}(\instance_b)$, we say that a pair $(X,\Psi)$ is a solution to $(\instance_b, \Psi_b)$ 
if it is a solution to \eulc on $\instance_b$ (ignoring the border terminals) and, additionally, $\Psi|_{\bterms} = \Psi_b$.
The border problem is defined as follows.

\defproblemoutput{\beulc}{An \eulc instance $\instance = (G,\Sigma,\cutsize,(\psi_{e,v})_{e \in E(G), v \in e})$ with $G$ being connected,
and a set $\bterms \subseteq V(G)$ of size at most $4k$; denote $\instance_b = (\instance, \bterms)$.}{For each $\Psi_b \in \mathbb{P}(\instance_b)$
output a solution $\sol_{\Psi_b} = X_{\Psi_b}$ to $(\instance_b,\Psi_b)$ with $|X_{\Psi_b}|$ minimum possible, or
output $\sol_{\Psi_b} = \bot$ if such a solution does not exist.}

\beulc generalizes \eulc: we may ask for $\bterms=\emptyset$ and take the output for the empty function $\Psi_b$.

Note that, for a \beulc instance $\instance_b$, we have $|\mathbb{P}(\instance_b)| = |\Sigma|^{|\bterms|} \leq s^{4\cutsize}$,
and the total number of edges output for $\instance_b$ is bounded by $\cutsize s^{4\cutsize}$. Define $q = \cutsize s^{4\cutsize}+1$.
In the recursive understanding phase for the \eulc problem we seek for $(q,2\cutsize)$-good separations.
The reason why we allow cuts of size $2\cutsize$ in the recursion, even though the solution is allowed to cut only $\cutsize$ edges,
will become more clear in the high connectivity phase. There, we would like to rely on the fact that any two connected subgraphs of more
than $q$ vertices are connected by at least $2\cutsize+1$ edge-disjoint paths, and the {\em{majority}} of these paths does not intersect the solution
we are looking for.

Assume that, using the algorithm of Lemma~\ref{lem:find-separation-ill}, we have found a $(q,2\cutsize)$-good separation $(V_1,V_2)$ of the graph $G$,
for the input instance $\instance_b$.
As $|\bterms| \leq 4\cutsize$, at least one of the sides contains at most $2\cutsize$ border terminals.
Without loss of generality we assume that $|V_1 \cap \bterms| \leq 2\cutsize$.
Now consider an instance $\newinst{\instance}_b$ that equals $\instance_b$ restricted to vertices $V_1$, with border
terminals $\newinst{\bterms} = (V_1 \cap \bterms) \cup N_G(V_2)$. In other words, we treat all the endpoints of the
cut $\delta(V_1,V_2)$ that lie in $V_1$ as border terminals. Note that, as $|\delta(V_1,V_2)| \leq 2\cutsize$, we have
$|\newinst{\bterms}| \leq 4\cutsize$ and $\newinst{\instance}_b$ is a valid \beulc instance.

Now, recursively solve the instance $\newinst{\instance}_b$, and let $Z$ be the union of all edges that appear in any
of the solutions output for $\newinst{\instance}_b$; note that $|Z| \leq q-1$.
It is not hard to see that, for any $\Psi_b \in \mathbb{P}(\instance_b)$, if there exists a solution
to $(\instance_b,\Psi_b)$, then there exists one that does not delete any edge of $E(G[V_1]) \setminus Z$.
Indeed, for any solution $(X,\Psi)$ to $(\instance_b,\Psi_b)$ one can replace the part of this solution living in $G[V_1]$
with the output to $\newinst{\instance}_b$ that is consistent with the appropriate behavior on the border terminals
$\newinst{\bterms}$, that is, with a solution to $(\newinst{\instance}_b,\Psi|_{\newinst{\bterms}})$.

Thus, all the edges of $E(G[V_1]) \setminus Z$ can be made undeletable.
In most edge-deletion problems, an undeletable edge can be contracted.
However, in the case of \eulc the situation is slightly more involved, as when contracting an edge $uv$
we need also to adjust the constraints on the edges incident to $u$ and $v$, to take into the account how
$\psi_{uv,v}$ translates the label $\Psi(u)$ into $\Psi(v)$ and vice versa.
This issue, together with a need of some reduction rule to reduce superfluous parallel edges,
causes some technical trouble in the formal proof, but does add any real difficulty to the problem.
Hence, in this illustration we simply assume that the undeletable edges may be contracted.

We remark that the operation applied to reduce parts of the graph determined to be undeletable is problem-dependent.
More complex problems, in particular node-deletion versions, may require even more careful simplification rules.

We now note that the assumptions $|Z| \leq q-1$ and $|V_1| > q$ ensure that at least one edge is contracted and we make a progress
due to the recursion step. Even more, we infer that there are only at most $q$ vertices left in $V_1$ after the contraction.
With this observation, we proceed to the estimation of the running time of the algorithm. By Lemma~\ref{lem:find-separation-ill} the time required to
find a $(q,2k)$-good edge separation is $O(2^{O(k\log q)}n^{3}\log n)=O(2^{O(k^2\log s)}n^{3}\log n)$; hence, the total running time is 
$O(2^{{O(k^2\log s)}}n^{4}\log n)$. We note that if $q$, more or less equal the bound on
the number of behaviors on the border terminals, is only a function of $\cutsize$, then we always obtain a running time of
the form $O(g(\cutsize)n^{4}\log n)$ for some function $g$.

\subsection{High connectivity phase}

We are left with the more involved part of our approach, namely, what to do when no $(q,2\cutsize)$-good edge separation is present in
the graph. Note that we can assume that the graph has more than $q(\cutsize+1)$ vertices, as otherwise a brute-force search, which
checks all the subsets of edges of size at most $\cutsize$, runs within the claimed time complexity bound.

The following simple lemma formalizes the structural properties of the graph after removing the solution. Note that this
structure is precisely the gain of the first phase of the algorithm.

\begin{lemma}\label{lem:high-structure-ill}
Let $G$ be a connected graph that admits no $(q,2\cutsize)$-good edge separation. Let $F$ be a set of edges of size at most $\cutsize$, such
that $G\setminus F$ has connected components $C_0,C_1,\ldots,C_\ell$. Then (i) $\ell\leq \cutsize$, and (ii) all the components $C_i$
except at most one contain at most $q$ vertices.

Moreover, for any two connected subgraphs $Z_1$ and $Z_2$ of $G$ that are vertex-disjoint and both of them contain more than $q$ vertices,
there exist at least $2\cutsize+1$ edge-disjoint paths between vertices of $Z_1$ and vertices of $Z_2$.
\end{lemma}

We would like to remark that if we apply the framework
directly to the node-deletion problems, we do not have any bound on $\ell$, i.e., the number of components --- in the
node-deletion setting we need additional tools here.

Fix some behavior on the border terminals $\Psi_b : \bterms \to \Sigma$; we iterate through all of them, which
gives $2^{O(k\log s)}$ overhead to the running time. Assume that there exists a solution $X\subseteq E(G)$ for this particular
choice. Without loss of generality let $X$ be of minimum size. Let $C_0,\ldots,C_\ell$ be components of $G\setminus X$, as in
Lemma~\ref{lem:high-structure-ill}, where $|V(C_i)|\leq q$ for $i=1,2,\ldots,\ell$.
Note that the assumption $|V(G)| > q(\cutsize+1)$ implies that $|V(C_0)| > q$, that is, the connected component of unbounded
size is actually huge. We call $C_0$ the {\emph{big component}}, and other components are {\emph{small components}}.

We now explain the general methodology how to highlight the solution $X$, using Lemma~\ref{lem:random}.
Let $V(X)$ denote the set of endpoints of the edges of $X$.
For every component $C_i$, choose its arbitrary spanning tree
$T_i$. Let $A_1=\bigcup_{i=1}^\ell E(T_i)$ be the set of edges of the spanning trees of small components. As
$\ell\leq \cutsize$, we have that $|A_1|\leq (q-1)\cutsize$. For every vertex $u\in V(X)\cap V(C_0)$ construct an arbitrary subtree $T_0^u$ of $T_0$
such that $u\in V(T_0^u)$ and $|V(T_0^u)|=q+1$, and let $A_2=\bigcup_{u\in V(X)\cap V(C_0)} E(T_0^u)$. We
have that $|V(X)|\leq 2\cutsize$ and hence $|A_2|\leq 2q\cutsize$.

We say that a set $S \subseteq E(G)$ \emph{interrogates} the solution $X$ if $S \cap X = \emptyset$
but $A_1 \cup A_2 \subseteq S$. Note that a family $\randfamily$ constructed by Lemma~\ref{lem:random}
for the universe $E(G)$ and constants $a = (3q-1)\cutsize$ and $b=\cutsize$ contains a set that interrogates $X$.
Hence we may branch into $|\randfamily|$ cases, guessing a set $S$ that interrogates the solution we are looking for. We refer to Figure~\ref{fig_spojnosc1} for an illustration.

\begin{figure}[ht]
\begin{center}
\includegraphics[width=5in]{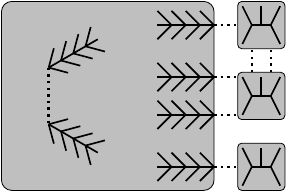}%
\caption{An illustration of the application of Lemma~\ref{lem:random} in the high connectivity phase.
The edges of the solution $X$ are dotted, and the edges required to be in the interrogating
set $S$ are thick.
We require that $S$ contains spanning trees of all small connected components
and large subgraphs attached to endpoints of the edges of $X$ in the big connected component.
Note that it is possible that an edge of the solution has both endpoints in the big connected component.
In this case we require that $S$ contains large subgraphs attached to its both endpoints.
\label{fig_spojnosc1}}
\end{center}
\end{figure}

The set $S$ is our way to highlight the solution $X$.
Note that there are three main properties of an interrogating set:
\begin{enumerate}
\item it is disjoint with the solution;
\item it spans all the small connected components; and
\item it spans a large connected subgraph around each endpoint of an edge of $X$ that belongs to the big connected component.
\end{enumerate}
In the subsequent arguments we will heavily exploit all three properties.
Our goal is to deduce $X$ using its interrogating set $S$; formally, we are going to find a minimum solution
to $(\instance_b, \Psi_b)$ that is additionally interrogated by $S$.

We analyze connected components of the graph $(V(G),S)$. Each such connected component
is called a {\emph{stain}}. A stain is \emph{big} if it contains more than $q$ vertices, and \emph{small} otherwise.
Note that any (unknown to us) connected component $C_i$ for $i \geq 1$ is a small stain, whereas
all big stains are contained in $C_0$.
Let $S^\bigcc$ be the union of vertex sets of all big stains.
The following structural observation greatly limits the number of possible sets $X$ to consider.
\begin{lemma}\label{lem:ill-options}
For any connected component $D$ of $G \setminus S^\bigcc$, exactly one of the following is true:
\begin{enumerate}
\item no edge incident to $D$ is contained in $X$, and $D \subseteq C_0$;
\item $D$ contains no vertex of $C_0$, and the small stains contained in $D$
are in one-to-one correspondence with components $C_i$ of $G \setminus X$ that are contained in $D$.
\end{enumerate}
\end{lemma}
\begin{proof}
If $D$ contains no vertex of $C_0$, the second property in the second point follows from the assumption
that $S$ contains a spanning tree $T_i$ of each connected component $C_i$ for $i \geq 1$, and that
$S$ is disjoint from the solution.

If $D$ contains a vertex of $C_0$, but the first point is not satisfied, then there exists a vertex
$v$ that is both in $D \cap C_0$ and is an endpoint of an edge of $X$. However, then $S$ should contain $T_0^v$
and $v$ belongs to a big stain, a contradiction to the definition of $D$.
\end{proof}
We remark that in the node-deletion setting the situation is a bit more complex, but an equivalent
of Lemma~\ref{lem:ill-options} can still be proven and exploited.
We also refer to Figure~\ref{fig_spojnosc2} for an illustration for Lemma~\ref{lem:ill-options}.

\begin{figure}[ht]
\begin{center}
\includegraphics[width=5in]{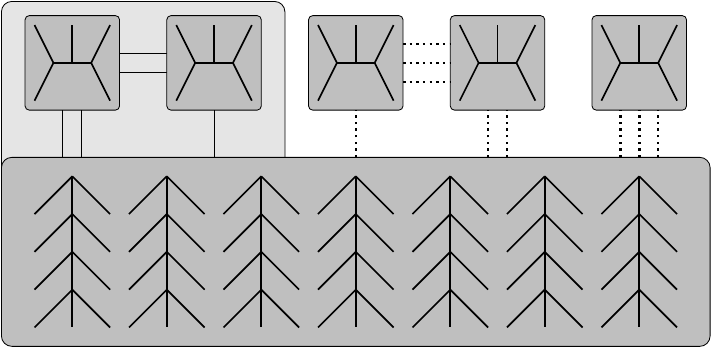}%
\caption{
An illustration of the situation after the set $S$ is guessed in the high connectivity phase:
the bottom half is the union of all big stains, $S^\bigcc$; for each connected component
of $G \setminus S^\bigcc$ we need to decide whether it goes entirely
to $C_0$ (the component on the left) or whether we cut it according to the stains
(the components in the middle and on the right).\label{fig_spojnosc2}}
\end{center}
\end{figure}

The main difficulty of the high connectivity phase is to deduce, for each connected component $D$
of $G \setminus S^\bigcc$, which option of Lemma~\ref{lem:ill-options} is true for $D$.
Once this decision is made, in all problems considered by us it is easy to deduce the entire set $X$.
Let us now illustrate this claim with the running example of \eulc.
We remark that we now work closely in the \eulc setting; the following argumentation is highly problem-dependent.

We need the following observation: as we seek for a solution $X$ disjoint with $S$, if we fix
a label of a vertex $v \in Z$ for some stain $Z$, the constraints on the edges of $S$ in $Z$
propagate the labeling to the whole stain $Z$. 
Here, we heavily rely on the fact that the constraints in the \ulcgen problem are permutations.
A labeling of $Z$ that originated from a labeling
of a single vertex, propagated through the constraints of $S$, is called a \emph{reasonable}
labeling of $Z$. Note that there are at most $|\Sigma|$ reasonable labelings of a single stain.

Thus, if for a connected component $D$ of $G \setminus S^\bigcc$ we know that the second option of Lemma~\ref{lem:ill-options}
is true, for each small stain $Z$ contained in $D$ we may find (by enumerating the reasonable labelings of $Z$) a labeling of $Z$ that minimizes the number
of unsatisfied constraints in $G[Z]$. 
On the other hand, if the first option of Lemma~\ref{lem:ill-options} is true for $D$, then
we may forget $D$ for a moment, solve the problem in the rest of the graph, and extend the obtained labeling of $S^\bigcc$ to $D$.
The last step should be possible, as we decided not to delete any constaint incident to $D$.

Thus, we are left with the quest to decide, for each connected component of $G \setminus S^\bigcc$, which
option of Lemma~\ref{lem:ill-options} to choose.
In \eulc, the main trick in this quest is to correctly label $S^\bigcc$.
Consequently, we now focus on big stains. As in the high connectivity phase our graph does not admit $(q,2\cutsize)$-good separation,
any two big stains $Z_1$ and $Z_2$ are connected by a family $\mathcal{P}$ of at least $2\cutsize+1$ edge-disjoint paths.
Let us focus on one path $P \in \mathcal{P}$ and assume that $P$ is disjoint with the solution $X$.
Using again the fact that all constraints in the \ulcgen{} problem are permutations, we infer that for any label assigned to the first vertex of $P$,
there exists a unique way to label all the vertices of $P$ while satisfying all the constraints on the edges of $P$.

Fix now one of at most $|\Sigma|$ reasonable labelings of $Z_1$; denote it $\Psi_1$.
Assuming $Z_1$ is labeled according to $\Psi_1$, there is a unique way to label the vertices of a path $P \in \mathcal{P}$
assuming $P$ is disjoint with the solution $X$. Moreover, the obtained (unique) label of the endpoint of $P$ yields
a unique reasonable labeling of $Z_2$.
As the \emph{majority} of the paths of $\mathcal{P}$ are disjoint with the solution $X$, the \emph{majority}
of paths of $\mathcal{P}$ should yield \emph{the same} labeling of $Z_2$, given the labeling $\Psi_1$ of $Z_1$.
Consequently, fixing a reasonable labeling on one big stain, provides us with a unique way to label
all other big stains, even without knowing the set $X$.
Therefore, we may branch into at most $|\Sigma|$ ways, guessing the labeling of \emph{all} big stains, that is, of the set $S^\bigcc$.

We remark that the argumentation in the previous paragraph is the sole reason for considering cuts of size $2\cutsize$
instead of only $\cutsize$ in the recursive understanding phase.

Hence, we have obtained a labeling $\Psi^\bigcc$ of $S^\bigcc$. Consider a component $D$ of $G \setminus S^\bigcc$
and assume that the first option of Lemma~\ref{lem:ill-options} is true for $D$. Consequently, the labeling
$\Psi^\bigcc$ can be (uniquely) extended to $D$ without violating any constraint incident to $D$.
Moreover, observe that an implication is true in the other direction as well: if $\Psi^\bigcc$ can be extended
to $D$ without violating any constraint incident to $D$, then we can greedily
choose the first option of Lemma~\ref{lem:ill-options} for $D$, as there is no need to delete any edge
incident to $D$ (assuming labeling $\Psi^\bigcc$).

This finishes the description of high connectivity phase
and the entire algorithm for \eulc.

\section{Preliminaries}\label{sec:prelims}

In this section prepare ground for formal proofs of the theorems stated in the introduction.
We start with setting up the notation, and then we give definitions and preliminary results
on ``good separations'', both in the edge- and node-deletion variants.

\subsection{Notation}

We use standard graph notation. As the definitions vary among the algorithms, we introduce problem-specific notation at the beginning of each corresponding section, describing whether we work on graphs, multigraphs or some other structures. Generally, by a graph we denote the pair $G=(V,E)$ consisting vertex set $V$ and edge set $E$. By $V(G)$ we denote the vertex set of $G$ and by $E(G)$ the edge set. For $F\subseteq E(G)$ by $V(F)$ we denote the set of endpoints of $F$. For $V_1,V_2\subseteq V(G)$, by $\edges(V_1,V_2)$ we denote the set of edges with one endpoint in $V_1$ and second in $V_2$. For $W\subseteq V(G)$, by $G[W]$ we denote the graph induced by $W$. For $u\in V(G)$, by $N(u)$ we denote the neighborhood of $u$, i.e., $N(u)=\{v\ |\ uv\in E(G)\}$, and the closed neighborhood is defined by $N[u]=N(u)\cup \{u\}$. We extend this notion to subsets in the following manner: for $W\subseteq V(G)$, $N[W]=\bigcup_{u\in W} N[u]$, and $N(W)=N[W]\setminus W$. If $X$ is a set of vertices or edges, by $G\setminus X$ we denote the graph $G$ with edges or vertices of $X$ removed.

\subsection{Contractions}

In this section we gather the definitions and simple facts connected to the notion of an {\emph{edge contraction}}. Our definition works in multigraphs.

\begin{definition}
Given a multigraph $G$ and an edge $uv\in E(G)$, {\emph{contraction}} of $uv$ is the operation that yields a new multigraph $G'$ with following properties:
\begin{itemize}
\item $V(G')=V(G)\setminus \{u,v\}\cup \{w_{uv}\}$, where $w_{uv}\notin V(G)$ is a new vertex;
\item $E(G')$ is first constructed from $E(G)$ by deleting all edges $uv$, and then substituting all occurrences of $u$ or $v$ by $w_{uv}$ in all the other edges.
\end{itemize}
\end{definition}

In other words, we preserve multiple edges but delete loops. With contraction of an edge $uv$ we can associate a mapping $\iota_{uv}:V(G)\to V(G')$ by setting $\iota_{uv}(u)=\iota_{uv}(v)=w_{uv}$ and $\iota_{uv}(t)=t$ for all $t\in V(G)\setminus \{u,v\}$. For $w\in V(G)$, we say that vertex $w$ is {\emph{contracted onto}} $\iota_{uv}(w)$. By somewhat abusing the notation we identify all the edges of $E(G')$ with the edges from $E(G)$ in which they originated. By contracting the edge set $S\subseteq E(G)$ we mean consecutively contracting edges of $S$ in an arbitrary order. Note that if some edge already disappeared from the graph because of becoming a loop, we omit this contraction. We usually use $\iota$ to denote the composition of all the mappings $\iota_{uv}$ corresponding to the performed contractions. The following lemma, which can be considered a folklore, implies that the order of performing the contractions does not matter.

\begin{lemma}\label{lem:contract-properties}
Let $G$ be a multigraph, $D\subseteq E(G)$ be a set of edges and $G'$ be the graph obtained by contracting $D$ in an arbitrary order. Then the following holds:
\begin{itemize}
\item $\iota(u)=\iota(v)$ if and only if $u$ and $v$ can be connected via a path consisting of edges from $D$, for $u,v\in V(G)$.
\item $\iota^{-1}(v)$ induces a connected subgraph of $G$, for $v\in V(G')$;
\item $E(G')\subseteq E(G)$;
\item an edge $vw\in E(G)$ is contained also in $E(G')$ if and only if $\iota(v)\neq \iota(w)$;
\item if $X\subseteq V(G')$, then $G'[X]$ is a maximal connected component if and only if $G[\iota^{-1}(X)]$ is;
\item in particular, $G$ is connected if and only if $G'$ is;
\item for every set $F$ such that $D\cap F=\emptyset$, $G'\setminus F$ can be obtained by contracting $D$ in $G\setminus F$.
\end{itemize}
\end{lemma}

From Lemma~\ref{lem:contract-properties} it follows, that given a graph $G=(V,E)$ and the set $D\subseteq E$, in time $O(|V|+|E|)$ we can construct the graph $G'$ obtained by contracting edges of $D$. We simply find connected components of the graph $(V,D)$, construct a new vertex for each of them, and for every edge of $E$ check whether it should be introduced in $G'$, and where.

\subsection{Preliminary results}

We start with a formal proof of Lemma~\ref{lem:random}.

\begin{proof}[Proof of Lemma~\ref{lem:random}]
For $a = 0$ or $b = 0$ the lemma is trivial; assume then $a,b \geq 1$.

We use the standard technique of splitters.
A $(n,r,r^2)$-{\em{splitter}} is a family of functions from $\{1,2,\ldots,n\}$ to
$\{1,2,\ldots,r^2\}$,
such that for any subset $X \subseteq \{1,2,\ldots,n\}$ of size $r$, one of the functions
in the family is injective on $X$. Naor et al. \cite{naor-schulman-srinivasan-derandom}
gave an explicit
construction of an $(n,r,r^2)$-splitter of size $O(r^6 \log r \log n)$
using $O(\textrm{poly}(r)\cdot n \log n)$ time.

Without loss of generality, assume that $a \leq b$ and that $U = \{1,2,\ldots,n\}$. Let $c=\min(a+b,n)$.
We construct a $(n,c,c^2)$-splitter using the algorithm of Naor et al. and,
for each function $f$ in the splitter and for each subset 
$S' \subseteq \{1,2,\ldots,c^2\}$ of size $a$, we put into the family
$\randfamily$ the set $f^{-1}(S') \subseteq U$. Assume now that we have $A,B\subseteq U$ such that $|A|\leq a$ and $|B|\leq b$. Obtain $A'$ and $B'$ by adding arbitrary elements of $U\setminus (A\cup B)$ to $A$ and $B$ so that $|A'|+|B'|=c$. By definition of the splitter, there exists some $f$ in the splitter that is injective on $A'\cup B'$. To finish the proof one needs to observe that if we take $S=f^{-1}(f(A'))$, then $A\subseteq S$ and $B\cap S=\emptyset$.

The time bound and the size of the constructed family $\randfamily$ follow
from the bound on the size of the splitter and the fact that
there are at most $\binom{(a+b)^2}{a} = 2^{O(a \log (a+b))}$ choices for the set $S'$;
note that for fixed $f$ and $S'$, the set $f^{-1}(S')$ can be computed
in linear time.
\end{proof}

A well-known result by Nagamochi and Ibaraki~\cite{nagamochi-ibaraki} states that the graph can be efficiently sparsified while preserving all the essential connectivity.

\begin{lemma}[\cite{nagamochi-ibaraki}]
\label{lem:japanese}
Given an undirected graph $G=(V,E)$ and an integer $k$, in $O(k(|V|+|E|))$ time we can obtain 
a set of edges $E_0 \subseteq E$ of size at most $(k+1)(|V|-1)$,
such that for any edge $uv \in E\setminus E_0$ in the graph
$(V,E_0)$ there are at least $k+1$ edge-disjoint paths between $u$ and $v$.
\end{lemma}

\begin{proof}
The algorithm performs exactly $k+1$ iterations.
In each iteration it finds a spanning forest $F$ of the graph $G$, adds
all the edges of $F$ to $E_0$ and removes all the edges of $F$ from the graph $G$.

Observe that for any edge $uv$ remaining in the graph $G$, the vertices $u$ and $v$
are in the same connected components in each of the forests found.
Hence in each of those forests we can find a path between $u$ and $v$;
thus, we obtain $\cutsize+1$ edge-disjoint paths between $u$ and $v$.
\end{proof}

\subsection{Good separations in edge-deletion problems}\label{sec:edge-cuts}

For sake of completeness, let us recall the definition of a $(q,k)$-good edge separation.

\begin{definition}
Let $G$ be a connected graph. A partition $(V_1,V_2)$ of $V(G)$ is called a $(q,k)$-{\emph{good edge separation}}, if
\begin{itemize}
\item $|V_1|,|V_2|>q$;
\item $|\edges(V_1,V_2)|\leq k$;
\item $G[V_1]$ and $G[V_2]$ are connected.
\end{itemize}
\end{definition}

We are ready to present proofs of lemmas regarding algorithms finding
good edge separations.

\begin{lemma}\label{lem:find-separation}
There exists a deterministic algorithm that, given an undirected, connected graph $G$ on $n$ vertices along with integers $q$
and $k$, in time $O(2^{O(\min(q,k)\log(q+k))}n^3\log n)$ either finds a $(q,k)$-good edge separation, or correctly concludes
that no such separation exists.
\end{lemma}

\begin{proof}
The algorithm iterates through all the sets from the family $\randfamily$, obtained from Lemma~\ref{lem:random} for universe $U=E(G)$ and constants $a=2q$ and $b=k$. For a set $S\in \randfamily$, we obtain a new graph $H$ by contracting all the edges of $S$. Let $\iota:V(G)\to V(H)$ be the mapping that maps every vertex of $G$ to the vertex it is contracted onto. We say that a vertex $u\in V(H)$ is {\emph{big}} if $|\iota^{-1}(u)|> q$. Now, for every pair of big vertices $u_1,u_2\in V(H)$ we compute some minimum edge cut between $u_1$ and $u_2$ if it is of size at most $k$, or find that it has to have larger size. This can be done in $O(k^2n^3)$ time, since first we can sparsify 
the graph by removing all the edges outside of the set $E_0$ returned by Lemma~\ref{lem:japanese}, and next for each of the $O(n^2)$ pairs of big vertices using the classical algorithm by Ford an Fulkerson in $O(k^2n)$ time find a cut of size at most $k$ if it exists.
Assume that for some pair of big vertices $u_1,u_2$ we have found a minimum edge cut $F_{u_1,u_2}$, of size at most $k$. We claim that $F_{u_1,u_2}$ induces a $(q,k)$-good edge separation of $G$, which can be returned as the output of the algorithm.

Let $v_1\in \iota^{-1}(u_1)$ and $v_2\in \iota^{-1}(u_2)$ be arbitrary vertices. Let $V_1,V_2$ be the sets of vertices reachable from $v_1,v_2$ in $G\setminus F_{u_1,u_2}$, respectively. We claim that $(V_1,V_2)$ is a $(q,k)$-good edge separation of $G$. Firstly, observe that $V_1$ and $V_2$ are disjoint. Otherwise there would be a path from $v_1$ to $v_2$ in $G$ that avoids $F_{u_1,u_2}$, which after applying the contractions would become a path from $u_1$ to $u_2$ in $H$ that avoids $F_{u_1,u_2}$. Secondly, observe that $V_1\cup V_2=V(G)$. It follows from the well-known properties of minimum cuts that in $H\setminus F_{u_1,u_2}$ every vertex is reachable either from $u_1$ or from $u_2$. As graphs $G[\iota^{-1}(u)]$ are connected for $u\in H$, we find that in $G$ every vertex is reachable either from $v_1$ or from $v_2$. Thirdly, observe that $|V_1|,|V_2| > q$, as $\iota^{-1}(u_1)\subseteq V_1$ and $\iota^{-1}(u_2)\subseteq V_2$.

We are left with proving that if the graph admits a $(q,k)$-good edge separation, then for at least one set $S_0\in \randfamily$ we obtain two big vertices that can be separated by an edge cut of size at most $k$. This ensures that if no solution has been found for any $S\in \randfamily$, then the algorithm can safely provide a negative answer. Fix some $(q,k)$-good edge separation $(V_1,V_2)$ and let $T_1,T_2$ be arbitrary subtrees of $G[V_1]$ and $G[V_2]$, respectively, each having exactly $q+1$ vertices. By the choice of family $\randfamily$, there exists $S_0\in\randfamily$ that contains all the edges of $T_1$ and $T_2$, but is disjoint with $\delta(V_1,V_2)$. In the step when $S_0$ is considered, after applying contractions all the vertices of $T_1$ are contracted onto one vertex $u_1$, all the vertices of $T_2$ are contracted onto one vertex $u_2$, but edges from $\delta(V_1,V_2)$ are not being contracted. Hence, we obtain big vertices $u_1,u_2$ that can be separated by an edge cut of size at most $k$.
\end{proof}

\begin{lemma}\label{lem:high-structure}
Let $G$ be a connected graph that admits no $(q,k)$-good edge separation. Let $F$ be a set of edges of size at most $k$, such
that $G\setminus F$ has connected components $C_0,C_1,\ldots,C_\ell$. Then (i) $\ell\leq k$, and (ii) all the components $C_i$
except at most one contain at most $q$ vertices.
\end{lemma}

\begin{proof}
Claim (i) follows directly from the fact, that removing an edge from the graph can increase the number of connected components by at most one. For Claim (ii), observe that if two components had at least $q$ vertices, then $F$ could serve as an edge cut between their vertex sets of size at most $k$. It follows that the minimum edge cut between their vertex sets would also have size bounded by $k$, hence it would induce a $(q,k)$-good edge separation in $G$.
\end{proof}

We now show that we can improve the polynomial factor in the running time of the procedure of Lemma~\ref{lem:find-separation}, at the cost of randomization.

\begin{lemma}
\label{lem:find-separation-fast}
There exists a randomized algorithm that, given an undirected, connected graph $G=(V,E)$ along with integers $q$ and $k$, in time $\tilde{O}(2^{O(\min(q,k)\log(q+k))}(|V|+|E|))$ either finds a $(q,k)$-good edge separation, or correctly concludes that no such separation exists with probability at least $(1-1/|V|^2)$.
\end{lemma}

\begin{proof}
Let $(V_1,V_2)$ be a $(q,k)$-good edge separation.
Intuitively, we want to have an edge-contraction process
such that no edge of $\delta(V_1,V_2)$ is contracted and each vertex which remains is big, because
then any cut of size at most $k$ gives a $(q,k)$-good edge separation, which we can find by using Karger's algorithm.
We use Lemma~\ref{lem:random} in a very similar fashion to the proof of Lemma~\ref{lem:find-separation}; however, as a few details
are different, we repeat the entire proof.

The algorithm iterates through all the sets from the family $\randfamily$,
obtained from Lemma~\ref{lem:random} for universe $U=E(G)$ and constants $a=2qk$ and $b=k$. For a set $S\in \randfamily$, we obtain a new graph $H'$ by contracting all the edges of $S$. Let $\iota':V(G)\to V(H')$ be the mapping that maps every vertex of $G$ to the vertex it is contracted onto. We say that a vertex $u'\in V(H')$ is {\emph{big}} if $|\iota'^{-1}(u')| > q$ and {\emph{small}} otherwise.
Let $S'\subseteq E(H')$ be the set of edges of $H'$ having at least one small endpoint.
We construct a graph $H$, by contracting all the edges of $S'$ in $H'$.
Let $\iota:V(G) \to V(H)$ be the mapping from the graph $G$ to the graph $H$.
Note that after contracting all the edges of $S'$ all the vertices are big in
the graph $H$ with respect to $\iota$.
By using Karger's algorithm~\cite{karger-min-cut},
in $\tilde{O}(k\log(qk)(|V|+|E|))$ time we find the minimum cut in the graph $H$
with probability at least $\left(1-\frac{1}{2^{ck\log(qk)}|V|^2\log |V|}\right)$, for some constant $c$.
If the minimum cut found is of size at most $k$, it immediately gives a
$(q,k)$-good edge separation in the graph $G$, since all the vertices of $H$ are big.

We are left with proving that if $G$ admits a $(q,k)$-good edge separation $(V_1,V_2)$, then for at least one set $S_0\in \randfamily$ the graph $H$ contains a cut of size at most $k$, providing
some (possibly different) $(q,k)$-good edge separation.
This ensures that if no solution has been found for any $S\in \randfamily$, then the algorithm can safely provide a negative answer.
For each vertex $u \in N(V_2) \subseteq V_1$ let $T^u$ be an arbitrary subtree of $G[V_1]$ containing the vertex $u$, having exactly $q+1$ vertices.
Similarly, for each vertex $u \in N(V_1)$ let $T^u$ be an arbitrary subtree of $G[V_2]$ containing $u$, having exactly $q+1$ vertices.
By the choice of the family $\randfamily$, there exists $S_0\in\randfamily$ that contains all the edges of $T^u$ for each $u \in V(\delta(V_1,V_2))$, 
but at the same time $S_0$ is disjoint with $\delta(V_1,V_2)$. 
In the step when $S_0$ is considered, after applying contractions, for each $u \in V(\delta(V_1,V_2))$ all the vertices of $T^u$ are contracted onto one vertex $u'$, which is big.
However, the edges from $\delta(V_1,V_2)$ are not being contracted.
Observe, that in the graph $H'$ no edge of $\delta(V_1,V_2)$ has a small endpoint,
and consequently all of the edges of $\delta(V_1,V_2)$ are present in the graph $H$,
and they induce a cut of size at most $k$.

Note that, the algorithm of Karger is used $O(2^{O(k\log(qk))}\log |V|)$ times,
and therefore, by the union bound, if our algorithm does not find a $(q,k)$-good edge separation,
with probability at least $(1-1/|V|^2)$ it does not exist.
\end{proof}

\newcommand{\undelV}{V^\infty}

\subsection{Good separations in node-deletion problems}\label{sec:node-cuts}

As we consider node-deletion problems in most of our results, we need to define
an appropriate variant of good separations; that is the main goal of this section.
In the edge-deletion variant,
we might have assumed that we only consider cuts that separate the graph into exactly
two connected components; this is no longer a case in the node-deletion variant.
Moreover, the applications require us to handle the possibility that some vertices
are undeletable.

It turns our that in the node-deletion problems we need to use two types of separations.
In the first one,
we require that, after removal of the separator, at least two connected components
are large.

\begin{definition}\label{def:node-good}
Let $G$ be a connected graph and $\undelV \subseteq V(G)$ a set of undeletable vertices.
A triple $(Z,V_1,V_2)$ of subsets of $V(G)$ is called
a $(q,\cutsize)$-{\em{good node separation}}, if 
\begin{itemize}
\item $|Z| \leq \cutsize$, 
\item $Z \cap \undelV = \emptyset$,
\item $V_1$ and $V_2$ are vertex sets of two different connected components of $G \setminus Z$; and
\item $|V_1 \setminus \undelV|,|V_2 \setminus \undelV| > q$.
\end{itemize}
\end{definition}

In the second one we require a bunch of connected components with the
same neighbourhood.

\begin{definition}\label{def:flowercut}
Let $G$ be a connected graph, $\undelV \subseteq V(G)$ a set of undeletable vertices, and $\bterms \subseteq V(G)$ a set of border terminals in $G$.
A pair $(Z,(V_i)_{i=1}^\ell)$ is called a $(q,\cutsize)$-{\em{flower separation}} in $G$
(with regard to border terminals $\bterms$), if the following holds:
\begin{itemize}
\item $1 \leq |Z| \leq \cutsize$ and $Z \cap \undelV = \emptyset$; the set $Z$ is the {\em{core}} of the flower separation $(Z,(V_i)_{i=1}^\ell)$;
\item $V_i$ are vertex sets of pairwise different connected components of $G \setminus Z$, each set $V_i$ is a {\em{petal}} of the flower separation $(Z,(V_i)_{i=1}^\ell)$;
\item $V(G) \setminus (Z \cup \bigcup_{i=1}^\ell V_i)$, called a {\em{stalk}}, contains more than $q$ vertices of $V\setminus \undelV$;
\item for each petal $V_i$ we have $V_i \cap \bterms = \emptyset$,
  $|V_i \setminus \undelV| \leq q$ and $N_G(V_i) = Z$;
\item $|(\bigcup_{i=1}^\ell V_i) \setminus \undelV| > q$.
\end{itemize}
\end{definition}

We now show how to detect
the aforementioned separations using Lemma \ref{lem:random},
similarly as it is done in the case of good edge separations.

\begin{lemma}\label{lem:detect-good-node}
Given a connected graph $G$ with undeletable vertices $\undelV \subseteq V(G)$
and integers $q$ and $\cutsize$, one may find in
$O(2^{O(\min(q,\cutsize) \log (q + \cutsize))} n^3 \log n)$ time
a $(q,\cutsize)$-good node separation of $G$, or correctly conclude that no such separation exists.
\end{lemma}
\begin{proof}
The algorithm iterates through all the sets from the family $\randfamily$,
obtained from Lemma~\ref{lem:random} for universe $U=V(G) \setminus \undelV$ and constants $a=2q+2$ and $b=\cutsize$.
For a set $S\in \randfamily$, we obtain a new graph $H$ by contracting all the edges between
vertices of $S \cup \undelV$ in $G$. Let $\iota:V(G)\to V(H)$ be the mapping that maps every vertex of $G$
to the vertex it is contracted onto, and let $S' = \iota(S \cup \undelV)$.
We say that a vertex $u\in S'$ is {\emph{big}} if $|\iota^{-1}(u) \setminus \undelV|> q$.

In the graph $H$, we assign weight $\infty$ to all vertices of $S'$ and weight $1$
to all vertices of $V(H) \setminus S'$. In this weighted graph, for every pair
of big vertices $u_1$ and $u_2$, we compute a minimum node cut between $u_1$ and $u_2$
if it is of size at most $\cutsize$, or find that it has to have larger size.
This can be done in $O(\cutsize n^3)$ time using the Gomory-Hu tree extended to node weighted
separations by Granot and Hassin~\cite{granot-hassin}.
That is we can use $|V(H)|-1$ applications of 
the classic Ford-Fulkerson algorithm, each of which consuming $O(\cutsize n^2)$ time,
since after finding $\cutsize+1$ vertex disjoint paths we may stop the algorithm.
Assume that for some pair of big vertices $u_1,u_2$ we have found
a minimum node cut $F_{u_1,u_2}$, of size at most $\cutsize$.
We claim that $F_{u_1,u_2}$ induces a $(q,\cutsize)$-good node separation of $G$,
which can be returned as the output of the algorithm.

Let $v_1\in \iota^{-1}(u_1) \setminus \undelV$ and $v_2\in \iota^{-1}(u_2) \setminus \undelV$
be arbitrary vertices.
Note that $F_{u_1,u_2} \subseteq V(G) \setminus \undelV$,
     as only vertices of $V(H) \setminus S' = V(G) \setminus (S \cup \undelV)$
have finite weights.
Let $V_1,V_2$ be the sets of vertices reachable from $v_1,v_2$ in $G\setminus F_{u_1,u_2}$,
    respectively.
We claim that $(F_{u_1,u_2},V_1,V_2)$ is a $(q,\cutsize)$-good node separation of $G$.
Indeed: $V_1$ and $V_2$ are defined as vertex sets of
two connected components of $G \setminus F_{u_1,u_2}$; moreover, $V_1 \neq V_2$
as $F_{u_1,u_2}$ separates $u_1$ from $u_2$ in $H$, and therefore $v_1$ from $v_2$ in $G$.
Finally, observe that $|V_1 \setminus \undelV|,|V_2 \setminus \undelV|> q$, as $\iota^{-1}(u_1)\subseteq V_1$ and $\iota^{-1}(u_2)\subseteq V_2$.

We are left with proving that if the graph admits a $(q,\cutsize)$-good node separation,
then for at least one set $S_0\in \randfamily$ we obtain two big vertices that can be
separated by a node cut of size at most $\cutsize$.
This ensures that if no solution has been found for any
$S\in \randfamily$,
then the algorithm can safely provide a negative answer.
Fix some $(q,\cutsize)$-good node separation $(Z,V_1,V_2)$ and let $T_1,T_2$
be arbitrary subtrees of $G[V_1]$ and $G[V_2]$, respectively,
each having exactly $q+1$ vertices that are in $V(G) \setminus \undelV$.
As $|Z|\leq \cutsize$, by the choice of family $\randfamily$, there exists $S_0\in\randfamily$
that contains
$(V(T_1) \cup V(T_2)) \setminus \undelV$, but is disjoint with $Z$.
In the step when $S_0$ is considered, after applying contractions
all the vertices of $T_1$ are contracted onto one vertex $u_1$,
all the vertices of $T_2$ are contracted onto one vertex $u_2$,
but vertices of $Z$ get weight $1$.
Hence, we obtain big vertices $u_1,u_2$ that can be separated by a node cut
of size at most $\cutsize$ (note that the algorithm does not necessarily find precisely
    the cut $Z$ in this step).
\end{proof}

\begin{lemma}\label{lem:detect-flower-cut}
Given a connected graph $G$ with undeletable vertices $\undelV \subseteq V(G)$ and border
terminals $\bterms \subseteq V(G)$
and integers $q$ and $\cutsize$, one may find in
$O(2^{O(\min(q,\cutsize) \log (q + \cutsize))} n^3 \log n)$ time
a $(q,k)$-flower separation in $G$ w.r.t. $\bterms$, or correctly conclude that no such flower separation exists.
\end{lemma}
\begin{proof}
We first note that, given a set $Z \subseteq V(G)$ of size at most $\cutsize$,
we can in $O(n^2)$ time verify whether there exists a $(q,\cutsize)$-flower separation with $Z$ as the core, that is,
$(Z,(V_i)_{i=1}^\ell)$ for some choice of the family of petals $(V_i)_{i=1}^\ell$.
Indeed, we may simply iterate over connected components of $G \setminus Z$ using a simple dynamic program. 
For each prefix of the sequence of connected components and for each $n'\leq n$ we compute, 
whether some of the components can be chosen to be petals so that the total number of vertices of $V(G)\setminus \undelV$ 
in the petals is equal to $n'$. When we consider the next connected component, if it does not satisfy requirements for a 
petal then we cannot take it as a petal (and we take the value of the cell computed in the last iteration
for the same value of $n'$). However, if it does satisfy these requirements, then we either not take it to be a petal 
(and do the same as previously) or take it (and we take the value of the cell computed in the last iteration
for the value $n'$ decremented by the number of vertices from $V\setminus \undelV$ in the considered component).
There exists a flower separation with $Z$ as the centre if and only if some of the values for $q+1\leq n'\leq |V(G)\setminus (\undelV\cup Z)|-q-1$ is
true in the last iteration. It is trivial to augment the dynamic program with backlinks, so that the flower separation
can be retrieved.

To prove the lemma, we iterate
through all the sets from the family $\randfamily$,
obtained from Lemma~\ref{lem:random} for universe $U=V(G) \setminus \undelV$ and constants $a=q$ and $b=\cutsize$.
For a set $S\in \randfamily$, we obtain a new graph $H$ by contracting all the edges between
vertices of $S \cup \undelV$ in $G$. Let $\iota:V(G)\to V(H)$ be the mapping that maps every vertex of $G$
to the vertex it is contracted onto, and let $S' = \iota(S \cup \undelV)$.
We say that a vertex $u\in S'$ is {\emph{interesting}} if $|\iota^{-1}(u) \setminus \undelV| \leq q$.
For each interesting vertex $u$ with  $|N_H(u)| \leq \cutsize$ we verify
whether there exists a $(q,\cutsize)$-flower separation $(Z,(V_i)_{i=1}^\ell)$ in $G$ w.r.t. $\bterms$
with the core $Z = N_H(u)$; note that $N_H(u) \subseteq V(G)$.
We output such a flower separation if we find one. If no flower separation is found
for any choice of $S$ and $u$, we conclude that no $(q,\cutsize)$-flower separation
exists in $G$ w.r.t. $\bterms$.
The time bound follows from the fact that for each vertex $u$, we can verify whether $u$ is interesting
and compute $N_H(u)$ in $O(n^2)$ time, and then within the same complexity check if $N_H(u)$ is the core of some $(q,\cutsize)$-flower separation.
To finish the proof of the lemma we need to show
that if the algorithm concludes that there is no appropriate flower separation in the graph,
then this conclusion is correct.

To this end, assume that there exists a $(q,\cutsize)$-flower separation $(Z,(V_i)_{i=1}^\ell)$
in $G$ w.r.t. $\bterms$. Note that $|V_1 \setminus \undelV|\leq q$ ($\ell \geq 1$ since $|(\bigcup_{i=1}^\ell V_i) \setminus \undelV| > q$)
and $|Z|\leq \cutsize$, so by the properties of the family $\randfamily$ there exists
a set $S_0 \in \randfamily$ with $(V_1 \setminus \undelV) \subseteq S_0$ and $Z \cap S_0 = \emptyset$.
Recall that $N_G(V_1) = Z$; thus, in the graph $H$ constructed for the set $S_0$
there is a vertex $u \in V(H)$ with $\iota^{-1}(u) = V_1$. Note that
$u$ is an interesting vertex (as $|V_1 \setminus \undelV| \leq q$) and $N_H(u) = Z$ (as $N_G(V_1) = Z$
by the definition of the flower separation).
Therefore the algorithm considers $Z = N_H(u)$ and finds a $(q,\cutsize)$-flower separation
in $G$ w.r.t. $\bterms$.
\end{proof}

We conclude this section with a lemma that shows that if we do not have any good node or flower
separations, then any $\cutsize$-cut not only cannot split the graph into two large components,
  but also cannot split the graph into too many small ones.

\begin{lemma}\label{lem:node-no-separation}
If a connected graph $G$ with undeletable vertices $\undelV \subseteq V(G)$ and border terminals
$\bterms \subseteq V(G)$
does not contain a $(q,\cutsize)$-good node separation or
a $(q,\cutsize)$-flower separation w.r.t. $\bterms$ then,
  for any $Z \subseteq V(G) \setminus \undelV$ of size at most $\cutsize$,
the graph $G \setminus Z$ contains at most $(2q+1)(2^\cutsize-1) + |\bterms| + 1$ connected components containing a vertex of $V \setminus \undelV$,
out of which at most one has more than $q$ vertices not in $\undelV$.
\end{lemma}
\begin{proof}
Let $Z \subseteq V(G) \setminus \undelV$, $|Z| \leq \cutsize$.
First, if there are at least two connected components of $G \setminus Z$ with more than $q$ vertices in $V(G) \setminus \undelV$, then a minimal subset of $Z$ separating these two components would induce a $(q,k)$-good node separation in $G$.
Thus, in $G \setminus Z$ we have at most one connected component with more than $q$ vertices outside $\undelV$ and at most $|\bterms|$ connected components
that contain a vertex from $\bterms$. We denote the remaining connected components containing at least one vertex of $V \setminus \undelV$ as {\em{nice}} ones;
they have at most $q$ vertices outside $\undelV$
each. Let us partition them with respect to their neighbourhood (which is a subset of $Z$).
Note that, if there exists a set $Z' \subseteq Z$, such that at least $2q+2$ nice connected components of $G \setminus Z$ that are adjacent to exactly $Z'$,
then there exists a $(q,k)$-flower separation in $G$ w.r.t. $\bterms$ with core $Z'$
and petals being $q+1$ of aforementioned nice connected components of $G \setminus Z$.
As there are at most $2^\cutsize-1$ nonempty subsets of $Z$, the lemma follows.
\end{proof}

\section{The algorithm for \ulc}\label{sec:full-ulc}

This section is devoted to fixed-parameter tractability
of the \ulc{} problem,
   parameterized by both the size of the cutset and the size of the alphabet.
We solve a bit more general version of the problem,
where we allow arbitrary unary relations and we allow the binary relations to be only partial permutations.
This generalizations appear naturally in our branching and reduction rules.

More formally, for an alphabet $\Sigma$, a binary relation $\psi \subseteq \Sigma \times \Sigma$
is called a {\em{partial permutation}}
if for every $\alpha \in \Sigma$ both $(\{\alpha\} \times \Sigma) \cap \psi$
and $(\Sigma \times \{\alpha\}) \cap \psi$ are of size at most one;
in other words, for every $\alpha \in \Sigma$, at most one value $\beta$ satisfies $\phi(\alpha,\beta)$
and at most one value $\beta'$ satisfies $\phi(\beta',\alpha)$.
For a partial permutation $\psi$, its reverse is defined
as $\psi^{-1} = \{(\beta,\alpha): (\alpha,\beta) \in \psi\}$.
For any two partial permutations $\psi_1,\psi_2$ their composition
is defined as
$$\psi_2 \circ \psi_1 = \{(\alpha,\gamma): \exists_{\beta \in \Sigma} (\alpha,\beta) \in \psi_1 \wedge (\beta,\gamma) \in \psi_2\}.$$
Note that a composition of two partial permutations is a partial permutation itself, and behaves 
in a similar manner as a composition of functions.

It is more convienient notationally to treat the deletion of a vertex as another, special label $\skull$.
Formally, we consider the following problem.

\defparproblemu{\ulc}{An undirected graph $G$, a finite alphabet $\Sigma$ of size $s$,
 an integer $\cutsize$, for each vertex $v \in V(G)$ a set $\phi_v \subseteq \Sigma$ and
 for each edge $e \in E(G)$ and each its endpoint $v$ a partial permutation $\psi_{e,v}$ of $\Sigma$, such that
if $e = uv$ then $\psi_{e,u} = \psi_{e,v}^{-1}$.}{$\cutsize+s$}{Does there exist
a function $\Psi: V(G) \to \Sigma \cup \{\skull\}$ such that
at most $\cutsize$ vertices are assigned value $\skull$,
for every $v \in V(G)$ we have $\Psi(v) \in \phi_v \cup \{\skull\}$ and for every
$uv \in E(G \setminus X)$ we have $\Psi(u) = \skull$, $\Psi(v) = \skull$, or $(\Psi(u),\Psi(v)) \in \psi_{uv,u}$?}

The relations $\psi_{e,u}$ are called {\em{edge constraints}}, the sets $\phi_v$ are called {\em{vertex constraints}}, the function
$\Psi$ is called a {\em{labeling}} and the set $\Psi^{-1}(\skull)$ is the {\em{deletion set}}.
For a vertex $v$ with $\Psi(v) = \skull$, we say that \emph{$v$ is deleted by $\Psi$}.

Before we start, we note that the edge-deletion variant (where we look
for a deletion set being a subset of edges; we are to label all vertices,
but we do not need to satisfy the constraints on the deleted edges)
reduces to the defined above node-deletion variant.

Indeed, first observe that in the \ulc{} problem we can assume that additionally
we are given in the input a set of undeletable vertices $V^\infty \subseteq V(G)$
and we are to find a labeling $\Psi$ that does not delete any vertex of $V^\infty$: we can reduce this
variant to the original one by replacing each undeletable vertex with a clique
on $\cutsize+1$ vertices, with constraints on the edges of the clique being identities.
Second, given an \edgeulc{} instance $(G,\Sigma,\cutsize,(\phi_v)_{v \in V(G)},(\psi_{e,v})_{e \in E(G),v\in e})$,
we can first make all vertices of $G$ undeletable, and then subdivide each edge,
so that the edge constraints on the two halves of the edge $e=uv$ of $G$ compose
to the constraint $\psi_{e,u}$; the new vertices introduced in this operation are kept
deletable.

We remark that the aforementioned reduction blows up the number of vertices and edges
of the input graph, thus the obtained algorithm for \eulc{} has worse
dependency on $n$ than $n^4\log n$ we obtain for the node-deletion variant.
Note that Section~\ref{sec:illustration} contains a sketch of an algorithm
of \eulc{} with $n^4 \log n$ dependency on $n$ in the running time.
As the main result of our work is fixed-parameter tractability of considered problems,
and the proven dependency on $n$ is far from being linear or even quadratic,
we refrain from formally showing an algorithm for \eulc{} with $n^4 \log n$ dependency
on $n$ in the running time.

Note that we may assume that the input graph $G$ in the \ulc{} problem is connected;
otherwise, we may solve the problem on each connected component, for all budgets
between $0$ and $\cutsize$, separately.
During the course of the algorithm, we maintain the connectivity of $G$.
We denote by $n$ the number of vertices of the graph of the currently considered \ulc{}
instance.

We also assume that
%we are working in the random access memory model, and that
the elements of $\Sigma$ can be compared in constant time.

The description of the algorithm consists of a sequence of {\em{steps}}.
Each step is accompanied with some lemmas and a discussion that justifies its correctness and verifies complexity bounds.

\subsection{Labelings}

We first extend the notion of labeling to arbitrary subsets of $V(G)$.
\begin{definition}
Given an instance $\instance = (G,\Sigma,\cutsize,(\phi_v)_{v \in V(G)}, (\psi_{e,v})_{e \in E(G),v \in e})$
and a set $S \subseteq V(G)$, a function $\Psi^S : S \to \Sigma \cup \{\skull\}$ is called a {\em{labeling}}
if it satisfies all constraints on $G[S]$ in $\instance$, that is:
for each $v \in S$ we have $\Psi^S(v) \in \phi_v \cup \{\skull\}$ and for each $uv \in E(G[S])$
we have $\Psi^S(u) = \skull$, $\Psi^S(v) = \skull$, or $(\Psi^S(u),\Psi^S(v)) \in \psi_{uv,u}$.

A labeling is \emph{deletion-free} if it does not assign the value $\skull$ to any vertex.
\end{definition}
For a labeling $\Psi$, by $\domain(\Psi)$ we denote its domain.

The following lemma is a straightforward corollary of the fact that the edge constraints are partial permutations.
\begin{lemma}\label{lem:ulc-propagate-labeling}
Let $\instance = (G,\Sigma,\cutsize,(\phi_v)_{v \in V(G)}, (\psi_{e,v})_{e \in E(G),v \in e})$ be a \ulc{} instance
and let $A \subseteq V(G)$ be an arbitrary subset of the vertex set that induces a connected subgraph of $G$.
Then, for every $v \in A$ and $\alpha \in \Sigma$ there exists at most one deletion-free labeling $\Psi^A: A \to \Sigma$
such that $\Psi^A(v) = \alpha$.
Furthermore, in $O(s|A|^2)$ time one can find such a labeling or correctly conclude that it does not exist.
Consequently, for each set $A \subseteq V(G)$ such that $G[A]$ is connected, there are at most $s$ deletion-free labelings of $A$ and those can be enumerated in $O(s^2|A|^2)$ time.
\end{lemma}
\begin{proof}
Note that for every $uw \in E(G[A])$, if $\Psi^A(u)$ is fixed, then there exists at most one value $\Psi^A(w)$ such that
$(\Psi^A(u),\Psi^A(w)) \in \psi_{uw,u}$, and such a value $\Psi^A(w)$ can be found in $O(s)$ time.
The first claim of the lemma follows from the assumption that $G[A]$ is connected: the labeling $\Psi^A$ can be found using a breadth-first search,
and then verified to satisfy all the constraints in $O(s|A|^2)$ time.
For the second claim, we simply iterate over all possible values $\Psi^A(v)$ for one fixed vertex $v \in A$.
\end{proof}

\subsection{Operations on the input graph}

In this section we define two basic operations the algorithm repetitively applies on the graph
and show their key properties.

\begin{definition}\label{def:ulc-update-edge}
Let $\instance = (G,\Sigma,\cutsize,(\phi_v)_{v \in V(G)},(\psi_{e,v})_{e \in E(G),v \in e})$ be a \ulc{} instance,
let $u,v \in V(G)$ and $\psi$ be a partial permutation of $\Sigma$.
By {\em{updating an edge $uv$ with a constraint $\psi$}} we mean the following operation:
if $uv \notin E(G)$, then we add an edge $uv$ to the graph $G$ with constraints
$\psi_{uv,u} = \psi$, $\psi_{uv,v} = \psi^{-1}$;
otherwise, we modify the constraints on the edge $uv$ in $G$ by replacing
$\psi_{uv,u}$ with $\psi_{uv,u} \cap \psi$ and $\psi_{uv,v}$ with $\psi_{uv,v} \cap \psi^{-1}$.
\end{definition}

Informally speaking, updating an edge $uv$ with $\psi$ is equivalent to adding a new edge between $u$ and $v$ with this constraint;
however, we use the definition above to avoid multiple edges in $G$. Note that obviously updating an edge cannot spoil the assumption
of connectivity of $G$. The following lemma is immediate.

\begin{lemma}\label{lem:ulc-update-edge}
Let $\instance'$ be a \ulc{} instance obtained from $\instance$
by updating and edge $uv$ with a constraint $\psi$. Then
$\Psi$ is a solution to $\instance'$ if and only if it is a solution to $\instance$
that satisfies the following additional property:
either $\Psi(u) = \skull$, $\Psi(v) = \skull$, or $(\Psi(u),\Psi(v)) \in \psi$.
\end{lemma}

The second operation allows us to remove a vertex that, for some reason, will not be deleted by an optimum solution.
\begin{definition}\label{def:ulc-bypass-vertex}
Let $\instance = (G,\Sigma,\cutsize,(\phi_v)_{v \in V(G)},(\psi_{e,v})_{e \in E(G),v \in e})$ be a \ulc{} instance and $v \in V(G)$.
By {\em{bypassing the vertex $v$}} we mean the following operation:
\begin{enumerate}
\item remove the vertex $v$ with its incident edges from the graph $G$;
\item for each $u \in N_G(v)$ we replace $\phi_u$ with $\phi_u \cap \{\beta: \exists_{\alpha \in \phi_v} (\alpha,\beta) \in \psi_{uv,v}\}$;
\item for each $u_1,u_2 \in N_G(v)$, $u_1 \neq u_2$, we update an edge $u_1u_2$ with a constraint $\psi_{vu_2,v} \circ \psi_{vu_1,u_1}$.
\end{enumerate}
\end{definition}

In the next lemma we formally check that bypassing a vertex has the same meaning as proclaiming it undeletable.

\begin{lemma}\label{lem:ulc-bypass-vertex}
Let $\instance'$ be a \ulc{} instance obtained from an instance $$\instance = (G,\Sigma,\cutsize,(\phi_v)_{v \in V(G)}, (\psi_{e,v})_{e \in E(G),v \in e})$$
by bypassing a vertex $v$ with $\phi_v \neq \emptyset$. Then the following holds:
\begin{itemize}
\item if $\Psi$ is a solution to $\instance'$, then there exists $\alpha \in \Sigma$ such that
$\Psi \cup \{(v,\alpha)\}$ is a solution to $\instance$;
\item if $\Psi$ is a solution to $\instance$ that satisfies $\Psi(v) \neq \skull$, then
$\Psi|_{V(G) \setminus \{v\}}$ is a solution to $\instance'$.
\end{itemize}
\end{lemma}
\begin{proof}
For the first claim, pick $\alpha$ as follows.
If there exists a neighbor of $v$ whose value in $\Psi$ is not $\skull$, pick any such neighbor $w$, 
set $\alpha$ such that $(\Psi(w),\alpha) \in \psi_{vw,w}$ and $\alpha \in \phi_v$. Note that
such $\alpha$ exists as, by the definition of the bypassing operation, in $\instance'$ the vertex constraint for $w$
are contained in $\{\beta: \exists_{\alpha' \in \phi_v} (\beta,\alpha') \in \psi_{vw,w}\}$.
If such a neighbor $w$ does not exist, pick $\Psi(v)$ to be an arbitrary element of $\phi_v$.

We claim that $\Psi \cup \{(v,\alpha)\}$
is a solution to $\instance$. Clearly, $\Psi \cup \{(v,\alpha)\}$ satisfies all vertex constraints of $V(G)$ as well as all edge
constrains on edges not incident to $v$, as those constrains in $\instance$ are supersets of the corresponding constraints in $\instance'$.
Moreover, the choice of $\alpha$ ensures that $\alpha \in \phi_v$.
We are left with verifying edge constraints $\psi_{uv,v}$ for $u \in N_G(v)$.
If $\Psi(u) = \skull$, then we are done, and if $u=w$, then clearly $(\alpha,\Psi(w)) \in \psi_{vw,v}$ by the choice of $\alpha$. Otherwise, by the definition of the bypassing operation,
$(\Psi(w),\Psi(u)) \in \psi_{uv,v} \circ \psi_{vw,w}$. Since $(\Psi(w),\alpha) \in \psi_{vw,w}$, we infer that $(\alpha,\Psi(u)) \in \psi_{uv,v}$ and the claim is proven.

For the second claim, denote $\alpha = \Psi(v)$. To prove the claim we need to verify that $\Psi|_{V(G) \setminus \{v\}}$ satisfies vertex
constraints on $N_G(v)$ (that may shrink during the bypassing operation) and edge constraints on edges between vertices in $N_G(v)$ (that
are updated during the bypassing operation).
First consider a vertex $u \in N_G(v)$. If $\Psi(u) = \skull$, there is nothing to check, so assume otherwise.
Since $\Psi$ is a solution to $\instance$ and $\Psi(v) \neq \skull$, we have
$\Psi(u) \in \phi_u$, $\alpha \in \phi_v$ and $(\Psi(u),\alpha) \in \psi_{uv,u}$. Thus $\Psi(u) \in \{\beta: \exists_{\alpha' \in \phi_v} (\alpha',\beta) \in \psi_{uv,v}\}$
and $\Psi(u)$ satisfies the vertex constraint at $u$ in $\instance'$.
Second, consider two vertices $u_1,u_2 \in N_G(v)$, $u_1 \neq u_2$ and $\Psi(u_1) \neq \skull$, $\Psi(u_2) \neq \skull$.
Since $\Psi$ is a solution to $\instance$ and $\Psi(v) \neq \skull$,
we have $(\Psi(u_1),\alpha) \in \psi_{vu_1,u_1}$ and $(\alpha,\Psi(u_2)) \in \psi_{vu_2,v}$. Therefore
$(\Psi(u_1),\Psi(u_2)) \in \psi_{vu_2,v} \circ \psi_{vu_1,u_1}$ and the claim is proven.
\end{proof}

During the course of the algorithm we perform bypassing operations multiple times, which can drastically increase the number of edges, even if the graph was sparse in the beginning. Therefore, we measure the complexity of our algorithm only in $n$, the number of vertices, and always use only the trivial quadratic bound on the number of edges.

\subsection{Borders and recursive understanding}

As discussed in the illustration, in the case of the \ulc{} problem, the definition of the border variant is completely natural:
informally speaking, for each vertex on the border, we need to know whether it is deleted
and if not, what label is assigned to it. More formally,
given a \ulc{} instance $\instance = (G,\Sigma,\cutsize,(\phi_v)_{v \in V(G)},(\psi_{e,v})_{e \in E(G),v \in e})$
and a set of border terminals $\bterms \subseteq V(G)$, for every function
$\cP:\bterms \to \Sigma \cup \{\skull\}$, we say that a solution $\Psi$ to $\instance$
{\em{is consistent with}} $\cP$ if $\Psi|_{\bterms} = \cP$.
Let $\mathbb{P}(\instance)$ be the set of all functions from $\bterms$ to $\Sigma \cup \{\skull\}$.
We define the border problem as follows.

\defproblemoutput{\bulc}{A \ulc{} instance
$\instance = (G,\Sigma,\cutsize,(\phi_v)_{v \in V(G)}, (\psi_{e,v})_{e \in E(G),v \in e})$
  with $G$ being connected,
and a set $\bterms \subseteq V(G)$ of size at most $4\cutsize$.}{
For each $\cP \in \mathbb{P}(\instance)$, output a solution $\sol_\cP = \Psi_\cP$
to the instance $\instance$ that is consistent with $\cP$ and deletes (assigns $\skull$)
to minimum possible number of vertices, or output $\sol_\cP = \bot$ if no such solution
exists.}

Note that $|\mathbb{P}(\instance)| \leq (s+1)^{4\cutsize}$ and all output solutions
$\Psi_\cP$ delete at most $\cutsize(s+1)^{4\cutsize}$ different vertices in total.
Let $q = \cutsize(s+1)^{4\cutsize} + 2\cutsize$; if $|V(G)| > q+2\cutsize$, then
there are at least $|V(G)|-q-2\cutsize$ vertices in $G$ that are not in $\bterms$
nor are deleted by any of the output solutions $\Psi_\cP$.

In the next lemma we formalize how a recursive step looks like, and verify its correctness.
The statement and its proof, although technical and notationally quite heavy, is completely standard:
we essentially need to verify that the information carried by the boundary terminals
in the definiton of \bulc{} is sufficient to independently substitute partial solutions
on different sides of a separation.

\begin{lemma}\label{lem:ulc-rekur}
Assume we are given a \bulc{} instance
$$\instance_b = (G,\Sigma,\cutsize,(\phi_v)_{v \in V(G)},(\psi_{e,v})_{e \in E(G),v \in e},\bterms)$$
and two disjoint sets of vertices $Z,\newinst{V} \subseteq V(G)$, such that
$|Z| \leq 2\cutsize$, $N_G(\newinst{V}) \subseteq Z$, $|\newinst{V} \cap \bterms| \leq 2\cutsize$
and the subgraph of $G$ induced by $W := \newinst{V} \cup Z_W$ is connected, where $Z_W := N_G(\newinst{V})$.
Denote $\newinst{\bterms} = (\bterms \cup Z_W) \cap W$ and
$$\newinst{\instance}_b = (G[W], \Sigma,\cutsize,(\phi_v)_{v \in W},(\psi_{e,v})_{e \in E(G[W]),v\in e},\newinst{\bterms}).$$
Then $\newinst{\instance}_b$ is a proper \bulc{} instance.
Moreover, if we denote by 
$(\newinst{\sol}_\cP)_{\cP \in \mathbb{P}(\newinst{\instance}_b)}$ an arbitrary output
to the \bulc{} instance $\newinst{\instance}_b$ and
$$U(\newinst{\instance}_b) = \newinst{\bterms} \cup \bigcup \{\newinst{\Psi}_\cP^{-1}(\skull) : \cP \in \mathbb{P}(\newinst{\instance}_b),
\newinst{\sol}_\cP = \newinst{\Psi}_\cP \neq \bot\},$$
then there exists a correct output $(\sol_\cP)_{\cP \in \mathbb{P}(\instance_b)}$
to the \bulc{} instance $\instance_b$ such that
every vertex that deleted by some solution $\sol_\cP \neq \bot$
belongs to $U(\newinst{\instance}_b)$.
\end{lemma}
\begin{proof}
The claim that $\newinst{\instance}_b$ is a proper \bulc{} instance follows directly
from the assumptions that $G[W]$ is connected, $|Z_W| \leq |Z| \leq 2\cutsize$
and $|\newinst{V} \cap \bterms| \leq 2\cutsize$. In the rest of the proof we justify the second claim
of the lemma.

Fix $\cP \in \mathbb{P}(\instance_b)$. Assume that there exists a solution
to the instance $\instance_b$ that is consistent with $\cP$; let $\Psi_\cP$
be such a solution with minimum possible number of deleted vertices.
To prove the lemma we need to show
a second solution $\Psi_\cP'$ to $\instance_b$ that deletes no more vertices than $\Psi_\cP$ does, is consistent with $\cP$,
and such that all vertices from $W$ that are deleted by $\Psi_\cP'$ lie in $U(\newinst{\instance}_b)$.

Let $\newinst{\cP}$ be the restriction of $\Psi_\cP$ to $\newinst{\bterms}$.
Note that $\newinst{\cP} \in \mathbb{P}(\newinst{\instance}_b)$
and $\Psi_\cP|_W$, is a solution to $\newinst{\instance}_b$
consistent with $\newinst{\cP}$.
Therefore the output $\newinst{\sol}_{\newinst{\cP}}$ to $\newinst{\instance}_b$ is different than $\bot$;
denote it $\newinst{\Psi}_{\newinst{\cP}}$.
By definition, this output deletes minimum possible number of vertices; in particular
$$|\newinst{\Psi}^{-1}_{\newinst{\cP}}(\skull)| \leq |\Psi^{-1}_\cP(\skull) \cap W|.$$

Define $\Psi'_\cP: V(G) \to \Sigma \cup \{\skull\}$ as follows: $\Psi'_\cP(v) = \newinst{\Psi}_{\newinst{\cP}}(v)$
for $v \in W$ and $\Psi'_\cP(v) = \Psi_\cP(v)$ otherwise.
By the optimality of $\newinst{\Psi}_{\newinst{\cP}}$, the number of vertices deleted
by $\Psi'_\cP$ is not larger than the number of vertices deleted by $\Psi_\cP$. 
Since $\Psi'_\cP$ is a blend of two labelings, it clearly satisfies all vertex constraints.
The fact that $Z_W = N(\newinst{V}) \subseteq \newinst{\bterms}$ ensures that $\Psi'_\cP$, $\newinst{\Psi}_{\newinst{\cP}}$,
and $\Psi_\cP$ agree on $Z_W$ and thus $\Psi'_\cP$  satisfies all edge constraints.
Finally, $\bterms \cap W \subseteq \newinst{\bterms}$ and $\newinst{\cP}$ is a restriction of $\Psi_\cP$, which
in turn is consistent with $\cP$, the labeling $\Psi'_\cP$ is consistent with $\cP$ as well.
This finishes the proof of the lemma.
\end{proof}

Note that in Lemma \ref{lem:ulc-rekur} we have $|U(\newinst{\instance}_b) \cap \newinst{V}|  \leq q$.
We are now ready to present the recursive steps of the algorithm.

\begin{step}\label{step:ulc-rekur}
Assume we are given a \bulc{} instance
$$\instance_b = (G,\Sigma,\cutsize,(\phi_v)_{v \in V(G)}, (\psi_{e,v})_{e \in E(G),v \in e}, \bterms).$$
Invoke first the algorithm of Lemma \ref{lem:detect-good-node} in a search for a $(q,2\cutsize)$-good node separation (with $\undelV = \emptyset$).
If it returns a good node separation $(Z,V_1,V_2)$, let $j \in \{1,2\}$ be such that $|V_j \cap \bterms| \leq 2\cutsize$ and denote $\newinst{Z} = Z$, $\newinst{V} = V_j$.
Otherwise, if it returns that no such good node separation exists in $G$,
invoke the algorithm of Lemma \ref{lem:detect-flower-cut} in a search for a $(q,\cutsize)$-flower separation w.r.t. $\bterms$ (with $\undelV = \emptyset$ again).
If it returns that no such flower separation exists in $G$,
pass the instance $\instance_b$ to the next step. Otherwise, if it returns a flower separation $(Z,(V_i)_{i=1}^\ell)$, denote $\newinst{Z} = Z$ and $\newinst{V} = \bigcup_{i=1}^\ell V_i$.

In the case we have obtained $\newinst{Z}$ and $\newinst{V}$ (either from Lemma \ref{lem:detect-good-node} or Lemma \ref{lem:detect-flower-cut}), 
invoke the algorithm recursively for the \bulc{} instance $\newinst{\instance}_b$ defined as in the statement
of Lemma \ref{lem:ulc-rekur} for separator $\newinst{Z}$ and set $\newinst{V}$, obtaining an output $(\newinst{\sol}_\cP)_{\cP \in \mathbb{P}(\newinst{\instance}_b)}$.
Compute the set $U(\newinst{\instance}_b)$. Bypass (in an arbitrary order) all vertices of $\newinst{V} \setminus U(\newinst{\instance}_b)$ to obtain a new instance $\instance_b'$ (observe that for each bypassed vertex $v$ we have $\phi_v \neq \emptyset$, which is a necessary condition for bypassing). Recall that $\newinst{\bterms}\subseteq U(\newinst{\instance}_b)$, so no border terminal get bypassed.
Restart the algorithm on the new instance $\instance_b'$ and obtain a family of solutions $(\sol_\cP')_{\cP \in \mathbb{P}(\instance_b)}$. For every $\cP \in \mathbb{P}(\instance_b)$, if $\sol_\cP'=\bot$ then output $\sol_\cP=\bot$ as well, while if $\sol_\cP=\Psi_\cP'$ then obtain $\Psi_\cP$ by extending $\Psi_\cP'$ on $U(\newinst{\instance}_b)$ using Lemma~\ref{lem:ulc-propagate-labeling} (we justify that such an extension exists in Lemma~\ref{lem:ulc-rekur-correctness}) and output $\sol_\cP=\Psi_\cP$.
\end{step}

Let us first verify that the application of Lemma \ref{lem:ulc-rekur} is justified. Indeed, by the definitions of
the good node separation and the flower separation, as well as the choice of $\newinst{V}$, we have in both cases $|\newinst{V} \cap \bterms| \leq 2\cutsize$
and that $G[\newinst{V} \cup N_G(\newinst{V})]$ is connected. Moreover, note that the recursive call is applied to the graph with strictly smaller number of vertices
than $G$: in the case of a good node separation, $V_2$ is removed from the graph, and in the case of a flower separation, recall that the definition
of the flower separation requires $Z \cup \bigcup_{i=1}^\ell V_i$ to be a proper subset of $V(G)$.
Finally, in both cases $|\newinst{V}| > q$, and $|\newinst{V} \setminus U(\newinst{\instance}_b)| \geq |\newinst{V}|-q \geq 1$ vertices are bypassed in Step \ref{step:ulc-rekur}.

The following lemma verifies the correctness of Step \ref{step:ulc-rekur}.

\begin{lemma}\label{lem:ulc-rekur-correctness}
Assume we are given a \bulc{} instance
$$\instance_b = (G,\Sigma,\cutsize,(\phi_v)_{v \in V(G)}, (\psi_{e,v})_{e \in E(G),v \in e}, \bterms)$$
on which Step \ref{step:ulc-rekur} is applied, and let $\instance_b'$ be an instance after all bypassing operations of Step \ref{step:ulc-rekur} are applied.
Let $(\sol_\cP')_{\cP \in \mathbb{P}(\instance_b')}$ be a correct output to $\instance_b'$.
Then there exists a correct output $(\sol_\cP)_{\cP \in \mathbb{P}(\instance_b)}$
to $\instance_b$, such that:
\begin{itemize}
\item $\sol_\cP=\bot$ if $\sol_\cP'=\bot$;
\item if $\sol_\cP'=\Psi_\cP'$ then $\Psi_\cP'$ can be consistently extended to $V(G)$
and for every such extension $\Psi_\cP$ is a correct output for $\cP$ in $\instance_b$;
\end{itemize}
\end{lemma}
\begin{proof}
The lemma is a straightforward corollary of Lemma \ref{lem:ulc-rekur} and the properties of the bypassing operation described in Lemma \ref{lem:ulc-bypass-vertex}. Lemma \ref{lem:ulc-rekur} ensures us that each vertex of $\newinst{V} \setminus U(\newinst{\instance}_b)$ is omitted by some optimal solution for every $\cP\in \mathbb{P}(\instance_b)$, which enables us to use Lemma \ref{lem:ulc-bypass-vertex}. Note that existence of the extension is asserted by the second claim of Lemma~\ref{lem:ulc-bypass-vertex}.
\end{proof}

We are left with an analysis of the time complexity of Step \ref{step:ulc-rekur}.
The applications of Lemmas \ref{lem:detect-good-node} and \ref{lem:detect-flower-cut} use $O(2^{O(\min(q,2\cutsize) \log(q+2\cutsize))} n^3 \log n) = O(2^{O(\cutsize^2 \log s)} n^3 \log n)$
time. Let $n' = |\newinst{V}|$; the recursive step is applied to the graph with at most $n' + 2\cutsize$ vertices and, after bypassing, there are at most $n-n'+q$ vertices
left. Moreover, each bypassing operation takes $O(sn^2)$ time, the computation of $U(\newinst{\instance}_b)$
takes $O((s+1)^{4\cutsize} n)$ time, and extending the labelings from the trimmed instance takes $O((s+1)^{4\cutsize} sn^2)$ time. The values of $s = |\Sigma|$ and $\cutsize$ do not change in this step.
Therefore, we have the following recursive formula for time complexity as a function of the number of vertices of $G$:
\begin{equation}
T(n)\leq \max_{q+1\leq n'\leq n-2k-1} \Big( 2^{O(\cutsize^2 \log s)} n^3\log n + T(n'+2\cutsize) + T(n-n'+q)   \Big).
\end{equation}
Note that the function $p(t)=t^4\log t$ is convex, so the maximum of the expression is attained at one of the ends. A straightforward inductive check of both of the ends proves that we have indeed the claimed bound on the complexity, i.e., $T(n) = O(2^{O(\cutsize^2 \log s)} n^4 \log n)$.

We conclude this section with a note that Lemma \ref{lem:node-no-separation} asserts that,
   if Step \ref{step:ulc-rekur} is not applicable, then for every set $Z \subseteq V(G)$
  of size at most $\cutsize$, the graph $G \setminus Z$ contains at most
 $t := (2q+1)(2^\cutsize-1) + 4\cutsize + 1$ connected components, out of which at most one
has more than $q$ vertices.

\subsection{Brute force approach}\label{sec:ulc-brute}
If the graph is reduced by Step \ref{step:ulc-rekur} is small,
   the algorithm may apply a straightforward brute-force approach
to the \bulc{} problem. In this section we describe this method formally.
\begin{lemma}\label{lem:ulc-brute-Psi}
Assume we are given a \bulc{} instance $$\instance_b = (G,\Sigma,\cutsize,(\phi_v)_{v \in V(G)}, (\psi_{e,v})_{e \in E(G),v \in e}, \bterms).$$
Let $X \subseteq V(G)$ be a set of size at most $\cutsize$ and let $\cP \in \mathbb{P}(\instance_b)$. Then, in time $O(s^2n^2)$, one can compute a function $\Psi:V(G) \to \Sigma \cup \{\skull\}$
such that $\Psi^{-1}(\skull) = X$ and $\Psi$ is a solution to $\instance_b$ consistent with $\cP$,
or correctly conclude that no such function exists.
\end{lemma}
\begin{proof}
We apply Lemma \ref{lem:ulc-propagate-labeling} to the vertex set of every connected component of the graph induced by $A = V(G) \setminus X$. For each output labeling, we verify if it is consistent with $\cP$.
\end{proof}
\begin{lemma}\label{lem:ulc-brute}
A correct output to a \bulc{} instance $$\instance_b = (G,\Sigma,\cutsize,(\phi_v)_{v \in V(G)}, (\psi_{e,v})_{e \in E(G),v \in e}, \bterms)$$
can be computed in $O((s+1)^{4\cutsize} s^2\cutsize n^{\cutsize+2})$ time.
\end{lemma}
\begin{proof}
We apply Lemma \ref{lem:ulc-brute-Psi} for each $\cP \in \mathbb{P}(\instance_b)$
(there are at most $(s+1)^{4\cutsize}$ choices) and for each deletion set $X \subseteq V(G)$
with  $|X| \leq \cutsize$ (at most $(\cutsize + 1) n^\cutsize$ choices).
\end{proof}
\begin{step}\label{step:ulc-brute}
If $|V(G)| \leq qt + \cutsize$, apply
Lemma \ref{lem:ulc-brute} to find a correct output
to a \bulc{} instance $\instance_b = (G,\Sigma,\cutsize,(\phi_v)_{v \in V(G)}, (\psi_{e,v})_{e \in E(G),v \in e}, \bterms)$.
\end{step}
Recall that $q = (s+1)^{2\cutsize} \cutsize = 2^{O(\cutsize \log s)}$
and $t = (2q+2)(2^\cutsize-1) + 2\cutsize + 1 = 2^{O(\cutsize \log s)}$.
Thus, if Step \ref{step:ulc-brute} is applicable,
its running time is $2^{O(\cutsize^2 \log s)}$.

\subsection{High connectivity phase}\label{sec:ulc-high}

Assume we have a \bulc{} instance $\instance_b = (G,\Sigma,\cutsize,(\phi_v)_{v \in V(G)}, (\psi_{e,v})_{e \in E(G),v \in e}, \bterms)$ where Steps \ref{step:ulc-rekur} and \ref{step:ulc-brute}
are not applicable.
In this section we show how to
exploit high connectivity of the graph implied by Lemma \ref{lem:node-no-separation}
to compute a correct output to $\instance_b$.
To this end, fix $\cP \in \mathbb{P}(\instance_b)$; we focus on
finding the solution $\sol_\cP$. First, let us solve some simple cases.

\begin{step}\label{step:ulc-empty-X}
For each $\cP \in \mathbb{P}(\instance_b)$, verify using Lemma \ref{lem:ulc-brute-Psi}
whether there exists solution $\sol_\cP = \Psi_\cP$ that does not delete any vertex at all.
If yes, output such a solution.
\end{step}

Note that, if $|V(G)|$ is too large for Step \ref{step:ulc-brute} to be applicable,
for every set $Z \subseteq V(G)$ of size at most $\cutsize$, 
the bound on the number of connected components from Lemma \ref{lem:node-no-separation}
implies that there exists exactly one connected component of $G \setminus Z$ with
more than $q$ vertices; denote its vertex set by $\bigcc(Z)$.
We extend this notion to labelings: for a labeling $\Psi$ that deletes at most $\cutsize$ vertices,
we denote $\bigcc(\Psi) = \bigcc(\Psi^{-1}(\skull))$.

%The set $\bigcc(Z)$ is an analogue of the component $C_0$ in the algorithm for the \steinercut problem. From now on the notation differs, as we try to avoid in the subsequent description usage of the bypassing operation, which is an analogue of edge contraction. We remark that we could proceed similarly to the algorithm from Section~\ref{sec:full-steiner}; however, we choose not to, as we find this way overcomplicated and harder to digest for this particular problem.

\subsubsection{Interrogating sets}

We now use Lemma \ref{lem:random} to get some more structure of the graph $G$.
\begin{definition}
Let $Z \subseteq V(G)$ be a set of size at most $\cutsize$ and let $S \subseteq V(G)$.
We say that $S$ {\em{interrogates}} $Z$ if the following holds:
\begin{enumerate}
\item $S \cap Z = \emptyset$;
\item for every connected component $C$ of $G \setminus Z$ with at most $q$ vertices,
  all vertices of $C$ belong to $S$;
\item for every $v \in Z$, such that $N_G(v) \cap \bigcc(Z) \neq \emptyset$,
  there exists a connected component of $G[S]$ that has more than $q$ vertices
  and contains at least one neighbour of $v$.
\end{enumerate}
We say that $S$ \emph{interrogates} a labeling $\Psi$ if it interrogates the deletion set $\Psi^{-1}(\skull)$.
\end{definition}

Note that in the third point, the considered component has to be entirely contained in $\bigcc(Z)$ due to its size.

\begin{lemma}\label{lem:ulc-random-applicable}
Let $\randfamily$ be a family obtained
by the algorithm of Lemma \ref{lem:random} for universe $U=V(G)$ and constants $a=qt + (q+1)\cutsize$ and $b=\cutsize$,
Then, for every $Z \subseteq V(G)$ with $1 \leq |Z| \leq \cutsize$, there exists a set
$S \in \randfamily$ that interrogates $Z$.
\end{lemma}
\begin{proof}
Fix $Z \subseteq V(G)$ with $|Z| \leq \cutsize$. Let $A_1$ be the union of vertex
sets of all connected components of $G \setminus Z$ that have at most $q$ vertices;
by Lemma \ref{lem:node-no-separation}, $|A_1| \leq qt$.
For each $v \in Z$ such that $N_G(v) \cap \bigcc(Z) \neq \emptyset$,
fix $w_v \in N_G(v) \cap \bigcc(Z)$ and a tree $T_v$ with exactly $q+1$ vertices
that contains $w_v$ and is contained in $\bigcc(Z)$; note that this is possible due to $|\bigcc(Z)|>q$. Let $A_2$ be the union
of vertex sets of all trees $T_v$ for $v \in Z$; clearly $|A_2| \leq (q+1)\cutsize$.
By Lemma \ref{lem:random}, as $|A_1 \cup A_2| \leq qt+(q+1)\cutsize$ and $|Z| \leq \cutsize$,
there exists a set $S \in \randfamily$ that contains $A_1 \cup A_2$ and 
is disjoint with $Z$. By the construction of the sets $A_1$ and $A_2$, the set
$S$ interrogates $Z$ and the lemma is proven.
\end{proof}

Note that, as $q,t = 2^{O(\cutsize \log s)}$,
     the family $\randfamily$ of Lemma \ref{lem:ulc-random-applicable}
is of size $2^{O(\cutsize^2 \log s)} \log n$ and can be computed
in $O(2^{O(\cutsize^2 \log s)} n \log n)$ time.
Therefore we may branch, guessing a set $S$ that interrogates
the solution $\sol_\cP = \Psi_\cP$ we are looking for.
\begin{step}\label{step:ulc-guess-S}
Compute the family $\randfamily$ from Lemma \ref{lem:ulc-random-applicable}
and branch into $|\randfamily|$ subcases, indexed by sets $S \in \randfamily$.
In a branch $S$ we seek for a solution $\Psi_\cP$ with minimum possible number of deleted vertices,
that not only is a solution to $\instance_b$ consistent with $\cP$,
but is also interrogated by $S$.
\end{step}
Lemma \ref{lem:ulc-random-applicable} verifies the correctness of the branching
of Step \ref{step:ulc-guess-S}; as discussed, the step is applied in
$O(2^{O(\cutsize^2 \log s)} n \log n)$ time and leads to
$O(2^{O(\cutsize^2 \log s)} \log n)$ subcases.

After choosing a set $S$, we may now slightly modify the set $S$ to make it more regular.
\begin{definition}
A vertex $v \in V(G)$ is said to be {\em{forsaken}}, if
\begin{itemize}
\item $v \in \bterms$ and $\cP(v) = \skull$; or
\item $\phi_v = \emptyset$.
\end{itemize}
\end{definition}
The forsaken vertices are those that are necessarily deleted by
any solution $\Psi_\cP$ consistent with $\cP$.
\begin{step}\label{step:ulc-clean-S}
As long as there exists a vertex $v \in V(G)$ that is not forsaken and $N_G[v] \cap S = \emptyset$,
add $v$ to $S$.
\end{step}
Step \ref{step:ulc-clean-S} can clearly be applied in $O(sn^2)$ time
(for all vertices it is applied to; note that Step \ref{step:ulc-clean-S}
 is applied to one vertex $v$ at a time and, by its application to the vertex $v$,
 it may become not applicable to the neighbours of $v$).
We now discuss its correctness.
Let $\Psi_\cP$ be a solution to $\instance_b$
that is interrogated by $S$ and consistent with $\cP$. Then $\Psi_\cP$ is interrogated
by $S \cup \{v\}$ unless $\Psi_\cP(v) = \skull$.
If this is the case, then since $N_G[v] \cap S = \emptyset$, by the last property of an interrogating set
$v$ is not adjacent to any vertex of $\bigcc(\Psi_\cP)$. Moreover, by the second property
of an interrogating set, $v$ is not adjacent to any vertex of connected component
of $G \setminus \Psi^{-1}_\cP(\skull)$ of size at most $q$. We infer that all
vertices of $N_G[v]$ are deleted by the labeling $\Psi_\cP$.
Since $v$ is not a forsaken vertex, there exists a value $\alpha \in \phi_v$ such that
if we assing $\alpha$ to $v$ in the labeling $\Psi_\cP$, we obtain a solution to $\instance_b$
that is consistent with $\cP$ (but not necessarily interrogated by $S$).
Therefore $\Psi_\cP$ is not a solution to $\instance_b$ consistent with $\cP$
with minimum possible number of deleted vertices, and we may omit it from consideration.

Step \ref{step:ulc-clean-S} gives us the following property of the set $S$.
\begin{lemma}\label{lem:ulc-clean-S}
After Step \ref{step:ulc-clean-S} is applied, any vertex $v$ that is not forsaken
is contained in $N_G[S]$.
\end{lemma}
\begin{proof}
If $v$ is not forsaken and $v \notin N_G[S]$, then $N_G[v] \cap S = \emptyset$
and Step \ref{step:ulc-clean-S} is applicable to $v$.
\end{proof}

\subsubsection{Labelings of big stains}

Let us now focus on a fixed branch $S \in \randfamily$.
\begin{definition}
Each connected component of $G[S]$ is called a {\em{stain}}.
A stain is {\em{big}} if it has more than $q$ vertices, and {\em{small}} otherwise.
\end{definition}

Let $S^\bigcc \subseteq S$ be the union of all vertex sets of big stains of $G[S]$.
We now establish a crucial observation that the fact that $G$ admits no $(q,2\cutsize)$-good
node separations implies that there are only very few reasonable labelings for $S^\bigcc$.
\begin{lemma}\label{lem:ulc-big-labelings}
One can in $O(\cutsize sn^3 + \cutsize s^2 n^2)$ time compute a family $\Psifamily$
of at most $s$ deletion-free labelings of $S^\bigcc$,
such that for every solution $\Psi$
to $\instance_b$
such that $S$ interrogates $\Psi$,
there exists $\Psi^\bigcc \in \Psifamily$ with $\Psi|_{S^\bigcc} = \Psi^\bigcc$.
\end{lemma}
\begin{proof}
If $S^\bigcc = \emptyset$ the lemma is trivial; assume then $S^\bigcc \neq \emptyset$.
For any big stain with vertex set $C$ in $G[S]$, by Lemma \ref{lem:ulc-propagate-labeling}
there are at most $s$ deletion-free labelings of $C$ and all these labelings can be computed in $O(s^2|C|^2)$ time.
Moreover, we know that any such a labeling is induced by fixing a label of one vertex of $C$;
in other words, for every two different labelings $\Psi',\Psi''$ of $G[C]$
and for every $v \in C$, we have $\Psi'(v) \neq \Psi''(v)$.

Let $C_1$ and $C_2$ be vertex sets of two different big stains in $G[S]$.
As $G$ admits no $(q,2\cutsize)$-good node separations, by Menger's theorem
there exists a sequence
$P^0,P^1,\ldots,P^{2\cutsize}$ of $2\cutsize+1$ paths, where each path starts in
$C_1$, ends in $C_2$ and the sets of internal vertices of those paths are pairwise disjoint.
Moreover, such a sequence of paths $P^0,P^1,\ldots,P^{2\cutsize}$ can be found
in $O(\cutsize n^2)$ time by the classic Ford-Fulkerson algorithm.

Let $\Psi$ be a solution to $\instance_b$ that is interrogated by $S$.
The crucial observation is that for at least $\cutsize+1$ indices $0 \leq i \leq 2\cutsize$ (i.e., for majority of them),
the path $P^i$ does not contain a vertex deleted by $\Psi$ (note that the endpoints
   of $P^i$ are in $C_1 \cup C_2 \subseteq S$, and thus not deleted by $\Psi$).
Denote the endpoints of $P^i$ as $v^i_1 \in C_1$ and $v^i_2 \in C_2$.
If $P^i$ does not contain any vertex deleted by $\Psi$, the composition of all edge constraints on $P^i$
(denote it by $\psi^i$) is a partial permutation such that
$(\Psi(v^i_1),\Psi(v^i_2)) \in \psi^i$. We infer that
for every $0 \leq i \leq 2\cutsize$ and every labeling $\Psi'$ of $G[C_1]$
there exists at most one labeling $\Psi'_{C_2,i}$ of $G[C_2]$ such that
$(\Psi'(v^i_1),\Psi'_{C_2,i}(v^i_2)) \in \psi^i$. Moreover, given $\Psi'$, all labelings
$\Psi'_{C_2,i}$ for all $0 \leq i \leq 2\cutsize$ can be computed in 
$O(\cutsize s(n + s|C_2|^2))$ time using Lemma~\ref{lem:ulc-propagate-labeling}.

Let $\Psi' = \Psi|_{C_1}$. As at least $\cutsize+1$ paths $P^i$ do not contain any vertices deleted by $\Psi$,
for a majority of indices $0 \leq i \leq 2\cutsize$, the labelings
$\Psi'_{i,C_2}$ are the same labelings. For a fixed big stain $C_1$ and for each
labeling $\Psi'$ of $G[C_1]$, we can compute this majority labeling $\Psi'_{\textrm{maj},C_2}$
of $G[C_2]$ for every big stain $C_2 \neq C_1$
in time $O(\cutsize n^2 + \cutsize s(n+s|C_2|^2))$ (including the time needed to compute paths $P^i$).
As there are at most $n$ big stains, and $s$ labelings of a fixed big stain $C_1$,
   the lemma follows.
\end{proof}

Note that for every $Z \subseteq V(G)$, $1 \leq |Z| \leq \cutsize$,
there exists the component $\bigcc(Z)$ and, if $S$ interrogates $Z$, then
there exists at least one big stain in $G[S]$ (note that we require here that $Z$ is nonempty;
the solutions with empty deletion sets are found by Step \ref{step:ulc-empty-X}).
This observation, together with Lemma \ref{lem:ulc-big-labelings}, justifies the following
step.
\begin{step}\label{step:ulc-big-labelings}
For each $\cP \in \mathbb{\instance_b}$, in a branch with index $S$,
if $G[S]$ contains no big stains, terminate this branch, and otherwise
invoke Lemma \ref{lem:ulc-big-labelings} to obtain a family $\Psifamily$
and branch into at most $s$ subcases, indexed by labelings $\Psi^\bigcc \in \Psifamily$.
For fixed $\cP$, in a branch with indices $S$ and $\Psi^\bigcc$,
we seek for a solution $\Psi_\cP$ with minimum possible size of the deletion set,
such that $\Psi_\cP$ is a solution to $\instance_b$ consistent with $\cP$,
interrogated by $S$ and $\Psi_\cP|_{S^\bigcc} = \Psi^\bigcc$.
\end{step}
Each application of Step \ref{step:ulc-big-labelings}
takes $O(\cutsize sn^3+\cutsize s^2n^2)$ time and leads to $O(2^{O(\cutsize^2 \log s)} \log n)$ subcases
in total.

\subsubsection{Final bounded search tree algorithm}

In this section we show how to finish the search for an appropriate output $\sol_\cP$
for $\cP \in \mathbb{P}(\instance_b)$, in a fixed branch with indices $S$ and $\Psi^\bigcc$.
This is done in a standard framework of a bounded search tree algorithm.
Informally speaking, the goal is to look at all connected components of $G$ after removal
of all already deleted vertices and $N_G[S^\bigcc]$, and decide, one by one, whether it should
be merged into the $\bigcc(\sol_\cP)$ or not. 

Formally speaking, we maintain a labeling $\Psi$, initiated as $\Psi^\bigcc \cup \cP$ together with all forsaken vertices deleted,
and our goal is to extend it to the desired solution via a bounded search algorithm. 
That is, at every recursive call of the branching algorithm, we look for a solution with minimum number of deleted vertices that additionally extends $\Psi$.

At every step, either no branching is performed and some new value is added to $\Psi$,
or we branch into $O(\Sigma)$ directions, in every branch deleting at least one new vertex.

The branching algorithm is described as a set of four {\em{rules}}.
At each moment, we apply the first applicable rule.

First, let us define the stopping condition for the branching algorithm.

\begin{algrule}[Finishing Rule]
If there is some inconsistency in $\Psi$: it deletes more than $\cutsize$ vertices or violates one of the constraints, terminate the current branch.
If $\Psi$ can be extended to $V(G)$ without deleting any additional vertex, then return such an extension as a solution for the current branch.
\end{algrule}
Note that recognizing whether Finishing Rule can be applied takes $O(s^2n^2)$ time using Lemma~\ref{lem:ulc-propagate-labeling}.

Given the labeling $\Psi$, let $N(\Psi)$ be the set of vertices that do not belong to the domain of $\Psi$,
but have a neighbor that is assigned a value from $\Sigma$ by $\Psi$ (i.e., in the domain of $\Psi$ but not deleted by $\Psi$).
Furthermore define $N[\Psi] = N(\Psi) \cup \domain(\Psi)$.
Consider now the following task: given a labeling $\Psi$, we would like to extend $\Psi$ to $N(\Psi)$ without deleting any new vertex.
Observe that there is essentially at most only one candidate $\Psi^N$ for such an extension: for every $v \in N(\Psi)$, we fix a neighbor $w(v)$ of $v$ with $\Psi(w(v)) \in \Sigma$,
and assign to $v$ the unique value that satisfies the constraint $\psi_{vw(v),v}$; note that such a value may not exist as $\psi_{vw(v),v}$ is a partial permutation or may not belong to $\phi_v$.

We use this observation in the following two rules.
First, since we are looking for a solution interrogated by $S$, we can immediately extend $\Psi$ to vertices in $N(\Psi) \cap S$.
\begin{algrule}[Extension Rule]
If there exists a vertex $v \in N(\Psi) \cap S$ with a neighbor $w(v)$ satisfying $\Psi(w(v)) \in \Sigma$, 
then assign to $v$ the unique value that satisfies the constraint $\psi_{vw(v),v}$.
\end{algrule}
Note that we can check if the Extension Rule can be applied in $O(sn^2)$ time, and in total it cannot be applied more than $n$ times in a single root-to-leaf path in the bounded search tree.

For the second rule, we observe in the next lemma that if $\Psi^N$ does not exist (i.e., it violates some constraint), then 
there is a witnessing contradiction on at most two vertices of $N(\Psi)$.

\begin{lemma}\label{lem:ulc-neighbourhood}
If $\Psi^N$ is undefined or violates some constraint (and hence is not a labeling), then 
there exists a set $B \subseteq N(\Psi)$ of size at most two
such that any labeling extending $\Psi$ needs to delete at least one vertex of $B$.
Moreover, such a set $B$ can be found on $O(sn^2)$ time.
\end{lemma}
\begin{proof}
First, if for some vertex $v \in N(\Psi)$, the value $\Psi^N(v)$ is undefined 
(there is no value matching $\Psi(w(v))$ in the constraint $\psi_{vw(v),v}$)
or such a value does not belong to $\psi_v$, we can take $B = \{v\}$.
Otherwise, $\Psi^N$ needs to violate some edge constraint, say $\psi_{vu,v}$ for some $v,u \in N[\Psi]$.
As $\Psi$ satisfies all edge constraints, either $u$ or $v$ is not in $\domain(\Psi)$, assume then $v \notin \domain(\Psi)$. If $u \in \domain(\Psi)$, then
$B = \{v\}$ satisfies the conditions of the lemma: any assignment of a value from $\Sigma$ to $v$
would violate either $\psi_{vu,v}$ or $\psi_{vw(v),v}$.
If $u \notin \domain(\Psi)$, then $B = \{u,v\}$ satisfies the conditions of the lemma:
however we assign values from $\Sigma$ to $u$ and $v$, we would either violate $\psi_{vw(v),v}$,
violate $\psi_{uw(u),u}$, or assign the values as in $\Psi^N$ and violate $\psi_{uv,v}$.

We note that we can compute $\Psi^N$ and find a constraint not satisfied by $\Psi^N$ in $O(sn^2)$ time.
\end{proof}

\begin{algrule}[Neighborhood Branching Rule]
Find $\Psi^N$. If it is undefined or is not a labeling, apply Lemma~\ref{lem:ulc-neighbourhood}
obtaining a set $B$, and branch into $|B|$ subcases, deleting one of the vertices of $B$ (i.e., assigning it value $\skull$ in $\Psi$).
\end{algrule}

Unfortunately, the Neighborhood Branching Rule is not always applicable.
However, we now show that, if it is not applicable, then we can make decisions on different components of $G\setminus N[\Psi]$ nearly independently.
Henceforth we assume that $\Psi^N$ is well-defined and is a labeling; see Figure~\ref{fig:ulc-final} for an illustration.

\begin{figure}[ht]
\begin{center}
\includegraphics{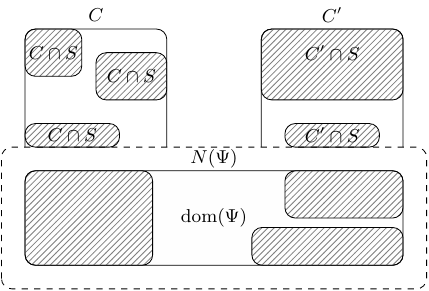}%
\caption{Situation before the last step of the branching algorithm. 
  The stripped rectangles represent the set $S$.}\label{fig:ulc-final}
\end{center}
\end{figure}

%Let $D^\bigcc = N_G(S^\bigcc) \subseteq V(G) \setminus S$,
%    $D^\smallcc = V(G) \setminus (S \cup D^\bigcc)$ and
%    $D = D^\bigcc \cup D^\smallcc = V(G) \setminus S$.
%Moreover, denote $S^\smallcc = S \setminus S^\bigcc$.
%The following lemma shows that on $G \setminus N_G[S^\bigcc]$,
%    we have only very restricted options.
% and is an analogue of Lemma~\ref{lem:steinercut-full-or-empty}.
\begin{lemma}\label{lem:ulc-branch-smallcc}
Let $C$ be a vertex set of an arbitrary connected component of $G \setminus N[\Psi]$
and let $\Psi_0$ be a labeling of $V(G)$ that extends $\Psi$ and is interrogated by $S$.
Then either:
\begin{itemize}
\item $C \subseteq \bigcc(\Psi_0)$, in particular no vertex of $C$ is deleted by $\Psi_0$, or
\item $C \setminus S$ is exactly the set of vertices from $C$ that are deleted by $\Psi_0$ and all vertices of $N(C)$ are deleted by $\Psi_0$.
\end{itemize}
Furthermore, if one can extend $\Psi^N$ to a labeling $\Psi^C$ of $N[\Psi] \cup C$ without deleting any new vertex, then 
a function $\Psi_1$ defined as equal $\Psi^C$ on $C$ and equal $\Psi_0$ on $V(G) \setminus C$ is a labeling of $V(G)$
extending $\Psi$ that deletes not more vertices than $\Psi_0$ does.
\end{lemma}
\begin{proof}
For the first part of the lemma, assume there exists $v \in C \cap \bigcc(\Psi_0)$.
From the connectivity of $G[C]$, we have that either $C \subseteq \bigcc(\Psi_0)$ or
there exists $uw \in E(G[C])$ with $\Psi_0(w) = \skull$ and $u \in \bigcc(\Psi_0)$.
Recall that $S$ interrogates $\Psi_0$; from the last property of the interrogating set
we infer that $w$ is adjacent to a big stain of $S$, a contradiction with the fact that all big stains are in $\domain(\Psi)$ but $w \notin N[\Psi]$.
Therefore either $C \subseteq \bigcc(\Psi_0)$ or $C \cap \bigcc(\Psi_0) = \emptyset$.

In the latter case, pick any $v \in C$ that is not deleted by $\Psi_0$.
As $C \cap \bigcc(\Psi_0) = \emptyset$, $v$ is contained in a connected component $C(v)$
of $G \setminus \Psi_0^{-1}(\skull)$ with at most $q$ vertices. By the second property of the interrogating
set $S$, $C(v)$ is a small stain of $S$. Consequently, we have $C(v) \subseteq C$, as any vertex of $C(v) \cap N(\Psi)$ would be amenable to the Extension Rule.

Consider now a vertex $w \in N(C)$ with a neighbor $u \in C$.
Assume that $w$ is not deleted by $\Psi_0$, that is, $\Psi_0(w) \neq \skull$.
If $w \in \bigcc(\Psi_0)$, then $\Psi_0(u) = \skull$ as we have assumed that $C \cap \bigcc(\Psi_0) = \emptyset$.
However, then, from the last property of the interrogating set, we infer that $u$ is adjacent to a big stain of $S$,
a contradiction with the fact that all big stains are in $\domain(\Psi)$ but $u \notin N[\Psi]$. 
In the other case, if $w \notin \bigcc(\Psi_0)$, then the Extension Rule is applicable to $w$ if $w \in N(\Psi)$
or to $u$ if $w \in \domain(\Psi)$. As we have reached a contradiction in all subcases, we infer that $w$ is deleted by $\Psi_0$.
Since the choice of $w$ was arbitrary, we have that all vertices of $N(C)$ are deleted by $\Psi_0$.
This finishes
the proof of the first part of the lemma.

For the second part of the lemma, first note that $\Psi_1$ deletes a subset of the vertices deleted by $\Psi_0$,
as $\Psi^C$ does not delete any vertex of $C$. Furthermore, since both $\Psi_0$ and $\Psi_C$ are labelings, $\Psi_1$ satisfies all vertex contraints.

As for edge constraints, clearly $\Psi_1$ satisfies all edge constraints for edges completely contained either in $G[C]$ or $G\setminus C$. 
Consider now an edge $uv$ with $u \in C$ and $v \in N(C)$. If $\Psi_0(v) = \skull$ then we are done. Otherwise, as $C$ is a connected component of $G \setminus N[\Psi]$,
we have $v \in N(\Psi)$. Since $\Psi_0$ extends $\Psi$, we have $\Psi_0(v) = \Psi^N(v)$. Since $\Psi^C$ extends $\Psi^N$ and is a labeling, we have that $\Psi_1$
satisfies the constaint on the edge $uv$. This completes the proof of the lemma.
\end{proof}

Consider now the following step: for every connected component $C$ of $G \setminus N[\Psi]$ we apply Lemma~\ref{lem:ulc-propagate-labeling} to check if there exists an extension $\Psi^C$ of $\Psi^N$ to $N[\Psi] \cup C$.
Note that if such an extension exists, Lemma~\ref{lem:ulc-branch-smallcc}, together with the fact that $\Psi^N$ exists and is a labeling, implies that the union of all such labelings $\Psi^C$ is a labeling of $V(G)$
that extends $\Psi$ and does not delete any new vertex, and the Finishing Rule is applicable. Note that such an output labeling may not be interrogated by $S$, but it is not an issue at this point.
Otherwise, we can use the limited options given by Lemma \ref{lem:ulc-branch-smallcc} to perform branching.

\begin{algrule}[Small Stains Rule]\label{rule:ulc-branch-smallcc}
For every connected component $C$ of $G \setminus N[\Psi]$
apply Lemma~\ref{lem:ulc-propagate-labeling} to check if there exists an extension $\Psi^C$ of $\Psi^N$ to $N[\Psi] \cup C$.
Since the Finishing Rule is not applicable, there exists a component for which such an extension does not exist; we pick such a component $C$ and branch into at most $s+1$ cases.
\begin{enumerate}
\item Apply Lemma~\ref{lem:ulc-propagate-labeling} to find at most $s$ deletion-free labelings of $G[C]$, and for every such labeling consider a branch with $\Psi$ extended with this labeling.
\item Furthermore, in an additional branch assign $\skull$ to every vertex of $N(C) \cup (C \setminus S)$, and for every connected component
of $G[C \cap S]$ apply Lemma~\ref{lem:ulc-propagate-labeling} to find a deletion-free labeling of this component.
\end{enumerate}
\end{algrule}
First, note that the Small Stains Rule is applicable in $O(s^2n^2)$ time.

The first option of the Small Stains Rule corresponds to the case $C \subseteq \bigcc(\Psi_0)$ of Lemma~\ref{lem:ulc-branch-smallcc}.
Note that the inexistence of $\Psi^C$ implies that, in every subcase of the first option, the Neighborhood Branching Rule will execute on some vertices of $N(C)$, leading to an increase in the number of deleted vertices.

The second option corresponds to the second case of Lemma~\ref{lem:ulc-branch-smallcc}. Note that the inexistence of $\Psi^C$ implies that either this branch is not executed at all
(due to the fact that some component of $G[C \cap S]$ does not admit any deletion-free labeling) or either $C \setminus S$ or $N(C)$ is nonempty, leading to an increase in the number of deleted vertices.
Note that after all vertices $N(C) \cup (C \setminus S)$ are assigned $\skull$, every connected component $C'$ of $G[C \setminus S]$ has all its neighbors deleted, and can freely choose any deletion-free
labeling output by Lemma~\ref{lem:ulc-propagate-labeling}.

Consequently, the Small Stains Rule, coupled with a possible following application of the Neighborhood Branching Rule, gives at most $2s+1$ subcases, and in every such subcase at least one new vertex is deleted.
As the Small Stains Rule is always applicable if no previous rule is applicable, we have concluded the description of the branching algorithm.
Recall that we terminate the algorithm after more than $\cutsize$ vertices are deleted.
As each rule is applicable in $O(s^2n^2)$ time and the number of rule applications before reaching a leaf is bounded by $n$ the whole branching algorithm  works in $O(2^{O(\cutsize \log s)} n^3)$ time.
This finishes the description of the fixed-parameter algorithm for \bulc{} and the proof of Theorem~\ref{thm:ulc-main}.

\section{The algorithm for \steinercut}\label{sec:full-steiner}

This section is devoted to the proof of Theorem~\ref{thm:steiner-main}, i.e., to the \steinercut problem, parameterized by the size of the cutset.

\defparproblemu{\steinercut}{A graph $G$, a set of terminals $T \subseteq V(G)$,
and integers $\numcc$ and $\cutsize$.}{$\cutsize$}{Does there exist a set $X$ of at most $\cutsize$
edges of $G$, such that in $G \setminus X$ at least $\numcc$ connected components
contain at least one terminal?}

For a \steinercut instance $(G,T,\numcc,\cutsize)$,
a set $X \subseteq E(G)$ is called a {\em{solution}}
if $|X| \leq \cutsize$ and $G \setminus X$ contains at least $\numcc$ connected
components that contain at least one terminal.

First, observe that by Lemma~\ref{lem:japanese} in $O(\cutsize n^2)$ time we can ensure 
that the graph $G$ has $O(\cutsize n)$ edges, by removing the edges outside of the set $E_0$.
Correctness of this step follows from the fact that in this operation
all cuts of size at most $k$ are preserved and moreover no new cut of size
at most $\cutsize$ appears, since for each of the removed edges at least $\cutsize+1$
edge disjoint paths between its endpoints remain.
We are going to use the assumption that there are $O(\cutsize n)$ edges in the graph
during the course of our algorithm.
To this end, we use Lemma~\ref{lem:japanese} after each reduction, thus always ensuring that the graph has at most $O(\cutsize n)$ edges, where $n$ is the current number of vertices.
We note that the only reason of using Lemma~\ref{lem:japanese} is caring about the polynomial factor.

Second, we observe that in the \steinercut{} problem we may assume that the graph $G$
is connected. Indeed, otherwise we may add to $G$ a clique on $\cutsize+2$ vertices
(so that the clique cannot be split with removal of $\cutsize$ edges),
make all vertices of the clique adjacent to exactly one vertex of each connected
component of $G$, and decrease $\numcc$ by the number of connected components of $G$
containing a terminal minus one.

If $G$ is connected, removal of $\cutsize$ edges may lead to at most $\cutsize+1$
connected components. Thus, we may assume that $\numcc \leq \cutsize + 1$, as otherwise
the answer is trivially negative.

Moreover, in the course of the algorithm we repetitively contract 
some edges of $G$. In the process of contraction, we remove loops, but we keep multiple edges.
Thus we allow $G$ to be a multigraph with multiple edges, but without loops.
Note that, if an edge $uv \in E(G)$ has multiplicity more than $\cutsize$,
the aforementioned sparsification process using Lemma~\ref{lem:japanese} reduces
the multiplicity down to at most $\cutsize+1$.
%$u$ and $v$ cannot be separated by an edge cut of size $\cutsize$. Thus, for such an edge
%$uv$, we may reduce the multiplicity of $uv$ to $\cutsize+1$.

The algorithm performs a number of steps. Description of each step is
accompanied by discussion of correctness and analysis of running time.

\subsection{Operations on the input graph}

The basic operation the algorithm performs on the graph is an edge contraction. As mentioned in the last section, we assume that after performing a series of contractions we reduce all multiplicities of the multiedges that exceed $\cutsize+1$ down to $\cutsize+1$. Moreover, if in $G$ a set of terminals $T \subseteq V(G)$ is given, if $T \cap \{u,v\} \neq \emptyset$, we replace $T$ with $(T \setminus \{u,v\}) \cup \{w_{uv}\}$, i.e., we put $w_{uv}$ into $T$ if and only if $u$ or $v$ belongs to $T$.

The following straightforward corollary of Lemma \ref{lem:contract-properties}
shows when we may contract an edge of $G$ in the \steinercut case.
\begin{lemma}\label{lem:steiner-contract}
Let $\instance = (G,T,\numcc,\cutsize)$ be a \steinercut instance,
let $D \subseteq E(G)$ and let $\instance' = (G',T',\numcc,\cutsize)$
be the instance $\instance$ with the edges $D$ contracted in an arbitrary order.
Then:
\begin{enumerate}
\item Any solution $X$ to $\instance'$ is a solution to $\instance$ as well
(recall that we treat $E(G')$ as a subset of $E(G)$).
\item For any solution $X$ to $\instance$ that is disjoint with $D$,
the set
$$X' = \{\iota(u)\iota(v): uv \in X, \iota(u) \neq \iota(v)\} \subseteq E(G')$$
is a solution to $\instance'$.
\end{enumerate}
\end{lemma}

We also use the notion of {\emph{identifying}} two vertices.

\begin{definition}
Given a multigraph $G$ and two vertices $u,v$, {\emph{identification}} of $u$ and $v$ is the operation of adding an edge $uv$ and contracting it.
\end{definition}

As identification is modelled by edge addition and contraction, we apply the same terminology also to this notion.

\subsection{Borders and recursive understanding}

In the border problem the graph is additionally equipped with at most $2\cutsize$
border terminals $\bterms$. For a border \steinercut{} instance
$\instance_b = (G,T,\cutsize,\bterms)$,
we need to remember an equivalence relation $\rel_b$ on $\bterms$,
that corresponds to how the border terminals are to be distributed among
connected components, a set $Y_b \subseteq \bterms$,
that carries information which border terminal is in connected component with
some terminal of $T$, and an integer $\numbcc_b$, which means that
in $\instance_b$ we are to obtain $\numbcc_b$ connected components
that contain a terminal. Formally speaking,
we define $\mathbb{P}(\instance_b)$ as the set of all triples $\cP = (\rel_b,Y_b,\numbcc_b)$,
where $\rel_b$ is an equivalence relation on $\bterms$,
$Y_b \subseteq \bterms$ and $0 \leq \numbcc_b \leq \cutsize+1$ is an integer.
Moreover, we require that if $(u,v) \in \rel_b$, then $u \in Y_b$ if and only if $v \in Y_b$.

We say that a set $X \subseteq E(G)$ is {\em{a solution}}
to $(\instance_b,\cP)$ for a triple $\cP = (\rel_b,Y_b,\numbcc_b) \in \mathbb{P}(\instance_b)$
if 
\begin{itemize}
\item two border terminals $u,v \in \bterms$
are in the same connected component of $G \setminus X$ if and only if $(u,v) \in \rel_b$;
\item for any border terminal $u \in \bterms$, the connected
component of $G \setminus X$ which contains $u$, contains a vertex of $T$ if
and only if $u \in Y_b$;
\item $G \setminus X$ contains exactly $\numbcc_b$ connected components that
contain a vertex of $T$;
\item $|X| \leq \cutsize$.
\end{itemize}

We formally define the border problem as follows.

\defproblemoutput{\bsteinercut}{A \steinercut{} instance
$\instance = (G,T,\numcc,\cutsize)$ with $G$ being connected,
and a set $\bterms \subseteq V(G)$ of size at most $2\cutsize$;
denote $\instance_b = (G,T,\cutsize,\bterms)$.}{
For each $\cP \in \mathbb{P}(\instance_b)$ output a solution $\sol_\cP = X_\cP$
  to $(\instance_b,\cP)$ with minimum possible $|X_\cP|$,
  or $\sol_\cP = \bot$ if such a solution does not exist.}

Observe that \bsteinercut generalizes \steinercut as we may ask for $\bterms=\emptyset$
and check the value of a solution consistent
with $(\emptyset,\emptyset,\numcc)$, as we can assume that after removing the minimum
size solution there are exactly $\numcc$ connected components containing a terminal.

Note that $|\mathbb{P}(\instance_b)| \leq (2\cutsize)^{2\cutsize} \cdot 2^{2\cutsize}
\cdot (\cutsize+2)$, as there are at most $|\bterms|^{|\bterms|}$ choices
for an equivalence relation $\rel_b$, $2^{|\bterms|}$ choices for $Y_b$
and $\cutsize+2$ choices for the value of $\numbcc_b$.
Denote
$$q = \cutsize (2\cutsize)^{2\cutsize} 2^{2\cutsize} (\cutsize + 2) + 1 = 2^{O(\cutsize \log \cutsize)}.$$

Let $\instance_b=(G,T,\cutsize,\bterms)$ be the given instance of \bsteinercut.
Assume that $G$ admits a $(q,\cutsize)$-good edge separation $(V_1,V_2)$.

As $V_1$ and $V_2$ are disjoint, at least one of them contains at most $\cutsize$ border
terminals from $\bterms$.
Without loss of generality assume that $|\bterms \cap V_1|\leq \cutsize$.
Let $\newinst{G} = G[V_1]$, $\newinst{T} = T \cap V_1$ and
$\newinst{\bterms} = (\bterms \cap V_1) \cup N_G(V_2)$.
Consider an instance $\newinst{\instance}_b = (\newinst{G},\newinst{T},\cutsize,\newinst{\bterms})$.
Note that $\newinst{\instance}_b$ is a correct instance of \bsteinercut,
as $|(\bterms\cap V_1)\cup N_G(V_2)|\leq 2\cutsize$.
Apply the algorithm recursively to the instance $\newinst{\instance}_b$
(note that it is strictly smaller instance as the vertex set $V_2$ is removed)
and let $U(\newinst{\instance}_b)$ be the set of edges that are contained in any output solution
for any behaviour on the border terminals of $\newinst{\instance}_b$.
Observe that $|U(\newinst{\instance}_b)|\leq q-1$.
Let $R=E(\newinst{G})\setminus U(\newinst{\instance}_b)$ be the set of remaining edges in
$\newinst{G} = G[V_1]$.
Contract the edges of $R$ in $G$ to obtain the new graph $G'$ with terminals $T'$
and border terminals $\bterms'$.
Let $V_1'$ be the set of vertices of $G'$ onto which vertices of $V_1$ were contracted.
Observe that $G'[V_1']$ is still connected as a contraction of a connected graph,
and has at most $q-1$ edges, as $|U(\newinst{\instance}_b)|\leq q-1$.
Therefore, $|V_1'|\leq q$. The contraction induces a mapping $\iota:V(G)\to V(G')$
that maps every vertex of $G$ to the vertex of $G'$ onto which it is contracted. 

The following lemma is useful in arguing safeness of the described operation.

\begin{lemma}\label{lem:equiv-steiner}
Let $\cP \in \mathbb{P}(\instance_b)$ and let $X_\cP$ be a solution to $(\instance_b,\cP)$.
Then there exists a second solution $X_\cP'$ to $(\instance_b,\cP')$, such that
$|X_\cP'| \leq |X_\cP|$ and additionally $X_\cP' \cap R = \emptyset$.
\end{lemma}
\begin{proof}
Consider the graph $\newinst{G} = G[V_1]$ and the set $X_\cP \cap E(\newinst{G})$.
Define:
\begin{itemize}
\item $\newinst{\rel}_b$ to be an equivalence relation on $\newinst{\bterms}$
such that for any $u,v \in \newinst{\bterms}$ we have $(u,v) \in \newinst{\rel}_b$
if and only if $u$ and $v$ are in the same connected component of $\newinst{G} \setminus X_\cP$.
\item $\newinst{Y}_b$ to be a set of those vertices $v \in \newinst{\bterms}$,
  such that the connected component of $\newinst{G} \setminus X_\cP$ that contains $v$
  contains a terminal from $\newinst{T} = T \cap V_1$ as well.
\item $\newinst{\numbcc}_b$ to be the number of connected components of
$\newinst{G} \setminus X_\cP$ that contain a vertex of $\newinst{T}$.
\end{itemize}
Let $\newinst{\cP} = (\newinst{\rel}_b,\newinst{Y}_b,\newinst{\numbcc}_b)$.
Clearly, $\newinst{\cP} \in \mathbb{P}(\newinst{\instance}_b)$
and $X_\cP \cap E(\newinst{G})$ is a solution to $(\newinst{\instance}_b,\newinst{\cP})$
(note that $\newinst{\numbcc}_b \leq |X_\cP \cap E(\newinst{G})| + 1\leq \cutsize + 2$,
 as $\newinst{G}$ is connected).
Therefore $\sol_{\newinst{\cP}} = \newinst{X}_{\newinst{\cP}} \neq \bot$,
that is, there exists a solution $\newinst{X}_{\newinst{\cP}}$ to $(\newinst{\instance}_b,\newinst{\cP})$,
such that $|\newinst{X}_{\newinst{\cP}}| \leq |X_\cP \cap E(\newinst{G})|$
and $\newinst{X}_{\newinst{\cP}} \cap R = \emptyset$.

Define $X_\cP' = (X_\cP \setminus E(\newinst{G})) \cup \newinst{X}_{\newinst{\cP}}$.
Clearly $|X_\cP'| \leq |X_\cP|$.
To finish the proof of the lemma we need to show that
$X_\cP'$ is a solution to $(\instance_b,\cP)$.

First, we show the following claim: for any $u,v \in \bterms \cup \newinst{\bterms}$,
  $u$ and $v$ are in the same connected component
of $G \setminus X_\cP$ if and only if $u$ and $v$ are in the same connected component
of $G \setminus X_\cP'$. We show only a proof in one direction, as proofs
in both directions are totally symmetric and use only the facts
that $X_\cP \setminus E(\newinst{G}) = X_\cP' \setminus E(\newinst{G})$
and that both $X_\cP \cap E(\newinst{G})$ and $X_\cP' \cap E(\newinst{G})$
are solutions to $(\newinst{\instance}_b,\newinst{\cP})$.

Let $u,v \in \bterms \cup \newinst{\bterms}$ be two vertices that are connected by
a path $P$ in $G \setminus X_\cP$. Let $u=v_0,v_1,\ldots,v_r=v$ be the
sequence of all vertices on $P$ that belong to $\bterms \cup \newinst{\bterms}$,
in the order they appear on $P$.
To prove the claim we need to show that for any $0 \leq i < r$, the vertices
$v_i$ and $v_{i+1}$ belong to the same connected component of $G \setminus X_\cP'$.
By definition, as $N_G(V_2) \subseteq \newinst{\bterms}$,
the subpath $P_i$ of $P$ between $v_i$ and $v_{i+1}$ lies
entirely in $\newinst{G}$ or entirely in $G \setminus E(\newinst{G})$.
In the first case, we infer that $v_i,v_{i+1} \in \newinst{\bterms}$,
$(v_i,v_{i+1}) \in \newinst{\rel}_b$
and $v_i$ and $v_{i+1}$ are in the same connected component of $\newinst{G} \setminus X_\cP'$
as $X_\cP' \cap E(\newinst{G})$ is a solution to $(\newinst{\instance}_b,\newinst{\cP})$.
In the second case, we infer that $v_i$ and $v_{i+1}$ are in the same connected
component of $(G \setminus E(\newinst{G})) \setminus X_\cP'$,
as $X_\cP \setminus E(\newinst{G}) = X_\cP' \setminus E(\newinst{G})$.
This finishes the proof of the claim.

As a straightforward corollary of the aforementioned claim,
we infer that for any $u,v \in \bterms$, we have $(u,v) \in \rel_b$
if and only if $u$ and $v$ are in the same connected component of $G \setminus X_\cP'$.
We now show that for any $v \in \bterms \cup \newinst{\bterms}$,
   its connected component of $G \setminus X_\cP$
contains a vertex of $T$ if and only if its connected component of $G \setminus X_\cP'$
contains a vertex of $T$.
The proofs in both directions are again totally symmetric, thus we present only
the forward implication.

Let $P$ be a path that connects the vertex $v$ with a terminal $w \in T$ in $G \setminus X_\cP$.
Let $u$ be the last (closest to $w$) vertex on $P$ that belongs to $\bterms \cup \newinst{\bterms}$
(as $v \in \bterms \cup \newinst{\bterms}$, such a vertex $u$ exists).
From the claim we infer that $u$ and $v$ are in the same connected component
of $G \setminus X_\cP'$. Let $P_u$ be the subpath of $P$ from $u$ to $w$.
We have two cases: either $P_u$ is contained in $\newinst{G}$, or in $G \setminus E(\newinst{G})$.
In the first case, we infer that $u \in \newinst{\bterms}$, $u \in \newinst{Y}_b$
(as $X_\cP \cap E(\newinst{G})$ is a solution to $(\newinst{\instance}_b,\newinst{\cP})$)
and that the connected component of $\newinst{G} \setminus X_\cP'$ that contains
$u$ contains a vertex from $\newinst{T} = T \cap V_1$ (not necessarily the vertex $w$).
In the second case, we infer that the path $P_u$ is present in
$(G \setminus E(\newinst{G})) \setminus X_\cP'$.
This finishes the proof that $v \in Y_b$ if and only if there exists a terminal
in the connected component of $G \setminus X_\cP'$ that contains $v$.

To finish the proof of the lemma we need to show that the number of connected components
of $G \setminus X_\cP'$ that contain a terminal equals $\numbcc_b$.
To this end, we partition the connected components of $G \setminus X_\cP$
and $G \setminus X_\cP'$ containing terminals into three types:
\begin{enumerate}
\item those that contain a vertex from $\bterms \cup \newinst{\bterms}$;
\item those that do not contain such a vertex, but are contained in $\newinst{G}$;
\item and the rest --- those that do not contain a vertex from $\bterms \cup \newinst{\bterms}$,
and are contained in $G \setminus V_1$.
\end{enumerate}
Our goal is to prove that for each type, the numbers of connected components containing terminals of corresponding
types in $G \setminus X_\cP$ and $G \setminus X_\cP'$ are equal.

For the first type, the claim is a straightforward corollary of
already proven facts that {\emph{(i)}} any two vertices
$u,v \in \bterms \cup \newinst{\bterms}$ are in the same connected component
of $G \setminus X_\cP$ if and only if they are in the same connected component
of $G \setminus X_\cP'$, and {\emph{(ii)}} for every vertex $u\in \bterms \cup \newinst{\bterms}$, the connected component of $G \setminus X_\cP$ containing $u$ contains a terminal if and only if the connected component of $G \setminus X_\cP'$ containing $u$ contains a terminal.

For the second type, note that we are to count the number of connected components
of $\newinst{G} \setminus X_\cP$ and $\newinst{G} \setminus X_\cP$ that do not contain a vertex
of $\bterms \cup \newinst{\bterms}$, or, equivalently, $\newinst{\bterms}$, but contain a terminal from $\newinst{T}$.
As both $X_\cP \cap E(\newinst{G})$ and $X_\cP' \cap E(\newinst{G})$ are solutions
to $(\newinst{\instance}_b,\newinst{\cP})$, this number is equal to $\newinst{\numbcc}_b$
minus the number of equivalence classes of $\newinst{\rel}_b$ that contain vertices from $Y_b$.

For the third type, recall that $X_\cP \setminus E(\newinst{G}) = X_\cP' \setminus E(\newinst{G})$,
so the sets of connected components of the third type in $G \setminus X_\cP$
and $G \setminus X_\cP'$ are equal.
This finishes the proof of the lemma.
\end{proof}

We now show that the output for the new instance $\instance_b'=(G',T',\cutsize,\bterms')$ can be easily transformed to the output for the original instance $\instance_b$.

\begin{lemma}\label{lem:steinercut-transform}
Let $(\sol'_\cP)_{\cP\in \mathbb{P}(\instance'_b)}$ be a correct output for the instance
$\instance'_b$.
For any $\cP=(\rel_b,Y_b,\numbcc_b) \in \mathbb{P}(\instance_b)$ define $\sol_\cP$ as follows.
\begin{itemize}
\item If $\iota$ maps 
two border terminals $u,v \in \bterms$ with $(u,v) \notin \rel_b$ to the same vertex of $\bterms'$,
then we take $\sol_\cP=\bot$.
\item Otherwise, we define $\cP' = (\rel_b',Y_b',\numbcc_b)$ as follows:
$(u',v') \in \rel_b'$ if $\iota^{-1}(u') \cap \bterms$ are contained in the same equivalence
class as $\iota^{-1}(v') \cap \bterms$, and $v' \in Y_b'$ if $\iota^{-1}(v') \cap \bterms \subseteq Y_b$;
      and take $\sol_\cP = \sol'_{\cP'}$.
\end{itemize}
Then the sequence $(\sol_\cP)_{\cP\in \mathbb{P}(\instance_b)}$
is a correct output to the instance $\instance_b$.
\end{lemma}
\begin{proof}
The lemma is a straightforward corollary
of the contraction properties of Lemmas \ref{lem:contract-properties} and \ref{lem:steiner-contract}, as well as Lemma \ref{lem:equiv-steiner}.
\end{proof}

We can now formally define the first step of the algorithm.

\begin{step}\label{step:steiner-rekur}
Using Lemma~\ref{lem:find-separation} we check, whether $G$ admits a $(q,\cutsize)$-good edge separation.
If this is not the case, we proceed to the second phase, i.e., high connectivity phase.
Otherwise let $(V_1,V_2)$ be this separation and without loss of generality assume
that $|\bterms\cap V_1|\leq \cutsize$. Construct the instance
$\newinst{\instance}_b=(G[V_1],T\cap V_1,\cutsize,(\bterms \cap V_1)\cup N_G(V_2))$, apply Lemma~\ref{lem:japanese} to it, solve it recursively
and compute $U(\newinst{\instance}_b)$, the set of edges that appear in any
solution given in the output. Contract all the remaining edges of $G[V_1]$ in $G$
to obtain new graph $G'$ with terminals $T'$ and border terminals $\bterms'$.
Define $\instance_b'=(G',T',\cutsize,T_b')$; recall that a vertex belongs to (border)
terminals if and only if some (border) terminal was contracted onto it.
Apply Lemma~\ref{lem:japanese} to $\instance_b'$, recursively solve the instance $\instance_b'$ and transform
the output according to Lemma~\ref{lem:steinercut-transform}.
\end{step}

Let us now estimate the running time. First, we spend $O(2^{O(\cutsize^2\log \cutsize)}n^3\log n)$
time to check, whether there exists a $(q,\cutsize)$-good edge separation. We apply the algorithm recursively to the instance $\newinst{\instance}_b$, which has $n'$ vertices for some $q+1\leq n'\leq n-q-1$. Construction of the instance $\newinst{\instance}_b$ takes $O(\cutsize n)$ time, construction of $U(\newinst{\instance}_b)$ takes $O(2^{O(\cutsize\log \cutsize)}n)$ time,
and construction of the instance $\instance_b'$ takes $O(\cutsize n)$ time.
Then, we apply the algorithm recursively to the instance $\instance_b'$ that has at most $n-n'+q$ vertices.
Therefore, we can derive the following recurrential inequality:
\begin{equation}
T(n)\leq \max_{q+1\leq n'\leq n-q-1} \Big( O(2^{{O(\cutsize^2\log \cutsize)}}n^{3}\log n) + T(n') + T(n-n'+q)   \Big)\,,
\end{equation}
Note that the function $p(t)=t^3\log t$ is convex, so the maximum of the expression is attained at one of the ends. A straightforward inductive check of both of the ends proves that we have indeed the claimed bound on the complexity, i.e., $O(2^{{O(\cutsize^2\log \cutsize)}}n^{4}\log n)$.

Observe, that if we use the randomized algorithm for finding good edge separations from Lemma~\ref{lem:find-separation-fast}, we obtain $T(v) \leq \tilde{O}(2^{{O(\cutsize^2\log \cutsize)}}n^{2})$ time complexity with success probability at least $(1-1/n)$,
since the graph is partitioned using good edge separations less than $n$ times.

\subsection{Brute force approach}

If the graph returned by Step~\ref{step:steiner-rekur} turns out to be small, we apply a brute-force approach. In this section we describe this step formally.

\begin{lemma}\label{lem:steiner-brute}
A correct output to a \bsteinercut{} instance $\instance_b = (G,T,\cutsize,\bterms)$
can be computed in $O(2^{O(\cutsize \log \cutsize)} n^{2\cutsize+2})$ time.
\end{lemma}
\begin{proof}
For every $\cP\in \mathbb{P}(\instance_b)$, and for every set $X \subseteq E(G)$
of at most $\cutsize$ edges that, for all $u,v \in V(G)$, takes either all or zero edges
$uv$, in $O(n^2)$ time we verify whether $X$ is a solution to $(\instance_b,\cP)$.
The time bound follows from the fact that $|\mathbb{P}(\instance_b)| \leq 2^{O(\cutsize \log \cutsize)}$
and there are at most $(\cutsize+1) n^{2\cutsize}$ choices of the set $X$.
\end{proof}

We are ready to provide the step of the algorithm that finishes resolving the problem, providing that the graph is sufficiently small.

\begin{step}\label{step:steiner-brute}
If $|V(G)|\leq (\cutsize+1)q$, then apply Lemma~\ref{lem:steiner-brute} to resolve the given \bsteinercut instance $\instance_b = (G,T,\cutsize,\bterms)$.
\end{step}

The correctness of this step is obvious, while from Lemma~\ref{lem:steiner-brute} we find that the running time is $O(2^{O(\cutsize^2 \log \cutsize)})$ as $q = 2^{O(\cutsize \log \cutsize)}$.
Therefore, from now on we can assume that $|V(G)|>(\cutsize+1)q$.

\subsection{High connectivity phase}

We now show how to solve \bsteinercut in $\tilde{O}(2^{{O(\cutsize^2\log \cutsize)}}n \log n)$ time
for the remaining case, when the graph $G$ does not admit a $(q,\cutsize)$-good edge
separation, yet is still too big to apply brute-force.
We need to output answers for all the possible triples $\cP\in \mathbb{P}(\instance_b)$.
We iterate through all such $\cP$; note that this gives only $2^{O(\cutsize\log \cutsize)}$
overhead in the running time.
Therefore, from now on we may assume that $\cP = (\rel_b,Y_b,\numbcc_b)$ is fixed.

Firstly, we make a quick check whether an empty deletion set is sufficient for our needs.
We formally need this step in order to be able to use nontriviality
of the solution in some technical reasonings.

\begin{step}\label{step:steiner-empty}
Given \bsteinercut instance $\instance_b = (G,T,\cutsize,\bterms)$ and
$\cP\in \mathbb{P}(\instance_b)$, verify in $O(\cutsize n)$ time
whether $\emptyset$ is a solution to $(\instance_b,\cP)$. If this is the case,
        output $\sol_\cP = \emptyset$.
\end{step}

The described step requires $O(\cutsize n)$ time and its correctness is obvious. From now on we may assume that the minimum deletion set is nonempty.

\subsubsection{Interrogating sets}

We now prepare ourselves to use Lemma~\ref{lem:random} to extract more structure of the graph $G$.

\begin{definition}
Let $X\subseteq E(G)$, $1\leq |X|\leq \cutsize$, and let $C_0,C_1,\ldots,C_\ell$ be connected components of $G\setminus X$, where, due to Lemma~\ref{lem:high-structure}, $\ell\leq \cutsize$ and $|V(C_i)|\leq q$ for $i\geq 1$. We say that a set of edges $S\subseteq E(G)$ {\emph{interrogates}} $X$ if the following properties are satisfied:
\begin{itemize}
\item $X\cap S=\emptyset$;
\item for every component $C_i$, $i\geq 1$, $S$ contains a spanning tree of $C_i$;
\item for every vertex $u\in V(C_0)\cap V(X)$, $u$ is contained in a connected component of $(V(G),S)$ of size at least $q+1$.
\end{itemize}
\end{definition}

Note that the first property together with $|V(C_i)|\leq q$ for $i\geq 1$ imply that the connected component considered in the third property has to be entirely contained in $C_0$. We now prove that a sufficiently large family given by Lemma~\ref{lem:random} contains a set interrogating a solution.

\begin{lemma}\label{lem:interrogation}
Let $\randfamily$ be a family obtained
by an application of the algorithm of Lemma \ref{lem:random} for universe $U=E(G)$ and constants $a=3q\cutsize$ and $b=\cutsize$,
Then, for any $X \subseteq E(G)$ with $1 \leq |X| \leq \cutsize$, there exists a set
$S \in \randfamily$ that interrogates $X$.
\end{lemma}
\begin{proof}
Let $C_0,C_1,\ldots,C_\ell$ be connected components of $G\setminus X$, where,
    due to Lemma~\ref{lem:high-structure}, $\ell\leq \cutsize$ and $|V(C_i)|\leq q$
    for $i\geq 1$.
As the algorithm did not finish when performing Step~\ref{step:steiner-brute},
we have that $|V(G)|>(\cutsize+1)q$, so $|V(C_0)|\geq q+1$.
Fix a spanning tree $T_i$ of each component $C_i$.
Let $A_1=\bigcup_{i=1}^\ell E(T_i)$, note that $|A_1|\leq q\cutsize$.
For every vertex $u\in V(C_0)\cap V(X)$ fix an arbitrary subtree $T_0^u$ of $T_0$
that contains exactly $q+1$ vertices, and define $A_2=\bigcup_{u\in V(C_0)\cap V(X)} E(T_0^u)$;
this is possible due to $|V(C_0)|\geq q+1$.
By Lemma~\ref{lem:random}, there exists a set $S\in \randfamily$
such that $A_1\cup A_2\subseteq S$ and $S\cap X=\emptyset$. It follows from the construction that $S$ interrogates $X$.
\end{proof}

This gives raise to the following branching step.

\begin{step}\label{step:steinercut-branching}
Using Lemma~\ref{lem:random} generate family $\randfamily$ for universe $U=E(G)$,
      and constants $a=3q\cutsize$ and $b=\cutsize$.
      Branch into $|\randfamily|$ subcases, labeled with $S\in \randfamily$.
      In branch $S$ we seek a solution $X$ to $(\instance_b,\cP)$
      such that $S$ interrogates $X$ and, moreover, $|X|$ is minimum among these.
\end{step}

Lemma~\ref{lem:interrogation} asserts that the deletion set of an optimal solution is interrogated by some set from family $\randfamily$. Therefore, in order to find a solution with minimum possible size of deletion set it suffices to take minimum over solutions given by the branches. Note that in this manner we introduce $O(2^{O(\min(q,\cutsize)\log(q+\cutsize))}\log n)=O(2^{O(\cutsize^2\log \cutsize)}\log n)$ branches. Moreover, as the family $\randfamily$ can be computed in $O(2^{O(\cutsize^2\log \cutsize)}n\log n)$ time and the construction of every branch takes $O(\cutsize n)$ time, the whole branching procedure takes $O(2^{O(\cutsize^2\log \cutsize)}n \log n)$ time. We now describe the subroutine performed in each branch, let it be labeled by $S\in \randfamily$. To simplify the presentation we assume that $S$ interrogates $X_0$ for some minimum solution $X_0$ and examine what happens with $X_0$ during the operations performed on the graph.

Let us contract all the edges from $S$ to obtain a new graph $H_0$. Let $\iota_0$ be the mapping from $V(G)$ to $V(H_0)$ corresponding to these contractions. Then, we obtain the new graph $H$ by identifying all the vertices $u\in H_0$ for which $|\iota_0^{-1}(u)| > q$ into a single vertex; such vertices $u$ are called {\emph{heavy}}. If there are no heavy vertices, we can safely terminate the branch, as a set that interrogates an nonempty solution must induce at least one connected component that has at least $q+1$ vertices. Otherwise, denote $b$ the vertex resulting in their identification; we will further refer to it as to the {\emph{core}} vertex. Let $\iota_1$ be the mapping from $V(H_0)$ to $V(H)$ corresponding to these identifications. Moreover, let $\iota=\iota_1\circ \iota_0$ be the mapping from $V(G)$ to $V(H)$ corresponding to the composition of these operations. 

We claim that the feasible deletion set $X_0$ 'survives' both steps.

\begin{lemma}\label{lem:steinercut-not-crossing}
Let $v,w\in V(G)$ such that $\iota(v)=\iota(w)$. Then $v$ and $w$ are in the same connected component of $G\setminus X$ for any set $X \subseteq E(G)$ interrogated by $S$.
\end{lemma}
\begin{proof}
Assume otherwise, that is, that we have $\iota(v)=\iota(w)$ but $v\in V(C_i)$ and $w\in V(C_j)$ for $i\neq j$. Without loss of generality assume that $i\neq 0$.

Assume first that $\iota_0(v)=\iota_0(w)$. As $v$ and $w$ are contracted onto the same vertex, there exists a path from $v$ to $w$ in $G$ that consists of edges of $S$. As $S\cap X=\emptyset$, this means that $v$ and $w$ are in the same connected component of $G\setminus X$, which is a contradiction.

Assume now that $\iota_0(v)\neq \iota_0(w)$. This means that $\iota_0(v)$ and $\iota_0(w)$ have to be identified while constructing graph $H$ from $H_0$. It follows that $|\iota_0^{-1}(v)|\geq q+1$. Therefore, there are at least $q$ vertices of $G$ that are reachable from $v$ via paths contained in $S$, hence disjoint with $X$. However, $i\neq 0$ so $|V(C_i)|\leq q$, which is a contradiction.
\end{proof}

From Lemma~\ref{lem:steinercut-not-crossing} we infer that all the edges of $X_0$
are still present in $H$, as, from the minimality of $X_0$ we may assume
that the edges of $X_0$ connect different connected components of $G \setminus X_0$.
Let us define the sets of terminals $T'$ and border terminals $\bterms'$ in $H$
by setting $u\in T'$ if and only if $\iota^{-1}(u)\cap T\neq \emptyset$
and $u\in \bterms'$ if and only if $\iota^{-1}(u)\cap \bterms \neq \emptyset$.

Moreover, due to Lemma \ref{lem:steinercut-not-crossing} and the fact
that $X_0$ is a solution to $(\instance_b,\cP)$,
we infer that for any $v,w \in \bterms$ with $\iota(v) = \iota(w)$,
we have $(v,w) \in \rel_b$ and $v \in Y_b$ iff $w \in Y_b$.
Thus we can project $\rel_b$ and $Y_b$ on $\bterms' \subseteq V(H)$,
by defining a relation $\rel_b'$ by taking $(\iota(v),\iota(w)) \in \rel_b'$ iff $(v,w) \in \rel_b$
and a set $Y_b'$ by taking $\iota(v) \in Y_b'$ iff $v \in Y_b$.

Define $\instance_b' = (H,T',\cutsize,\bterms')$ and $\cP' = (\rel_b',Y_b',\numbcc_b)$;
note that $\cP' \in \mathbb{P}(\instance_b')$.
The next lemma can be proven by a straightforward check of the definition of the solution, as Lemma~\ref{lem:steinercut-not-crossing} asserts connected components of $G\setminus X_0$ correspond in a one-to-one manner to connected components of $H\setminus X_0$.

\begin{lemma}\label{lem:steinercut-new-solution}
The set $X_0$ is a solution to $(\instance_b',\cP')$.
\end{lemma}

In the same manner we can also obtain a converse implication.

\begin{lemma}\label{lem:steinercut-new-solution-con}
If a set $X' \subseteq E(H)$ is a solution to $(\instance_b',\cP')$ of minimum
possible size, then $X'$ is a solution to $(\instance_b,\cP)$ as well.
\end{lemma}
%\begin{proof}
%The lemma is a straightforward corollary of the
%properties of contractions
%(Lemmas \ref{lem:contract-properties} and \ref{lem:steiner-contract}).
%\end{proof}

Lemmas \ref{lem:steinercut-new-solution} and \ref{lem:steinercut-new-solution-con} justify correctness of the following step, which we now state formally.

\begin{step}\label{step:steinercut-contraction}
Contract the edges of $S$ to obtain the graph $H_0$.
If there are no heavy vertices in $H_0$, terminate the branch as $S$ cannot interrogate any feasible deletion set. Otherwise, identify all heavy vertices into one core vertex $b$ and
denote by $H$ the resulting graph.
Define the instance $\instance_b'=(H,T',\cutsize,\bterms')$
and $\cP'$ in a natural manner, described in this section.
Run the remaining part of the algorithm on the instance $\instance_b'$ and triple $\cP'$
to obtain a solution $\sol$, which then output as $\sol_\cP$.
\end{step}

Note that Step~\ref{step:steinercut-contraction} can be performed in $O(\cutsize n)$ time.

\subsubsection{Connected components of $H \setminus \{b\}$ and dynamic programming}

We now establish some structural properties of the behaviour of $X_0$ in $H$.
The goal is to limit the class of possible solutions in $\instance_b'$
we need to search through.
Note that in the constructed graph $H$ the vertex $b$ plays a special role,
as we know that $V(X_0)\cap V(C_0)\subseteq \iota^{-1}(b)$.

Let $B_1',B_2',\ldots,B_p'$ be the components of $H\setminus \{b\}$ and let $B_i=H[V(B_i')\cup \{b\}]$ for $i=1,2,\ldots,p$. Observe that $B_i$ are connected, edge-disjoint and $b$ separates them. Moreover, we can compute them in $O(\cutsize n)$ time. We now claim that for each component $B_i$, the solution either takes $E(B_i)$ entirely, or is disjoint with it.

\begin{lemma}\label{lem:steinercut-full-or-empty}
For each $i=1,2,\ldots,p$, either $E(B_i)\cap X_0=\emptyset$ or $E(B_i)\subseteq X_0$.
\end{lemma}
\begin{proof}
We consider two cases.
Assume first that no vertex in $V(B_i')$ is in the same connected component
of $H \setminus X_0$ as the core $b$. In particular, this implies that edges connecting $b$ with $V(B_i')$ belong to $X_0$.
Moreover, in $G$ the set $\iota^{-1}(V(B_i'))$
is a union of vertex sets of some components $C_i$ for $i \geq 1$.
As $S$ interrogates $X_0$, each component $C_i$ for $i \geq 1$
is projected by $\iota$ onto a single vertex in $H$.
Consider an edge $e \in E(B_i')$. We infer that $e$ connects two different connected components of $G \setminus X_0$,
thus $e \in X_0$. Therefore $E(B_i) \subseteq X_0$ in this case.

Assume now that there exists a vertex $v \in V(B_i')$ that is in the same connected
component of $H \setminus X_0$ as $b$, i.e., $\iota^{-1}(v) \subseteq V(C_0)$.
Let $w \in N_{B_i'}(v)$.
As $v \neq b$, we have $|\iota^{-1}(v)| \leq q$ and $|\iota^{-1}(w)|\leq q$, so we have $vw \notin X_0$,
as otherwise $S$ does not interrogate $X_0$.
Since the choice of $v$ and its neighbour $w$ was arbitrary, and since
$B_i'$ is connected, we infer that all vertices of $B_i'$ belong to the same
connected component of $H \setminus X_0$. As $\iota^{-1}(v) \subseteq V(C_0)$, we find that $\iota^{-1}(V(B_i')) \subseteq V(C_0)$.
From the minimality of $X_0$ we infer that $E(B_i) \cap X_0 = \emptyset$.
\end{proof}

Lemma \ref{lem:steinercut-full-or-empty} ensures us that we may
seek the optimal solution among the ones that do not intersect edge
sets of components $B_i$ nontrivially.
We now show how to construct the optimal solution in the
remaining instance in $O(\cutsize^2 n)$ time.
First, we resolve components $B_i'$ that contain border terminals.

\begin{step}\label{step:steinercut-borders}
We define $D \subseteq \bterms'$ as follows.
If $b \in \bterms'$, we define $D$ to be the equivalence class
of $\rel_b'$ that contains $b$.
Otherwise, we branch into at most $1+|\bterms'| \leq 1 + 2\cutsize$ subcases,
taking $D$ to be an empty set or one of the equivalence classes of $\rel_b'$.
Given $D$, we seek for a solution $X$ where the set of border terminals
being in the in the same connected component of $H \setminus X$ as $b$
equals $D$.

For a fixed choice of $D$, we may immediately resolve
the connected components $B_i'$ that contain a border terminal of $\bterms'$.
Initiate a counter $\numcc_0 = 0$.
For each component $B_i'$ with $V(B_i') \cap \bterms' \neq \emptyset$ perform
the following.
\begin{enumerate}
\item If there exists a border terminal in $V(B_i') \cap D$, as well as
a border terminal in $(V(B_i') \cap \bterms') \setminus D$, terminate this branch,
  as for any of the two cases given by Lemma \ref{lem:steinercut-full-or-empty},
  we cannot satisfy the conditions implied by $\rel_b'$ and the set $D$.
\item If all border terminals in $V(B_i') \cap \bterms'$ belong to $D$,
  contract all edges of $B_i$
  (we have $E(B_i) \cap X = \emptyset$ in this case).
\item If all border terminals in $V(B_i') \cap \bterms'$ do not belong to $D$,
  include $E(B_i)$ into the constructed solution: decrease $\cutsize$ by $|E(B_i)|$
  and increase $\numcc_0$ by $|V(B_i') \cap T'|$ (by including $E(B_i)$ into a solution,
  we delete $|E(B_i)|$ edges and create $|V(B_i') \cap T'|$ new connected components
  that contain a terminal).
\end{enumerate}
\end{step}

Note that Step~\ref{step:steinercut-borders} can be performed in $O(\cutsize^2 n)$ time and
results in at most $2\cutsize+1$ branches. Its correctness is asserted by Lemma~\ref{lem:steinercut-full-or-empty}.
After Step \ref{step:steinercut-borders} is applied, all terminals of $\bterms'$ are either contracted onto $b$
(if they belong to $D$), or became isolated vertices after the removal of edges included to the
constructed solution. If some equivalence class of $\rel_b'$ different than $D$ is larger than
a single vertex of $H$, or we do not satisfy the conditions implied by the set $Y_b'$
for some vertex of $\bterms'$, we may immediately reject the current branch. Otherwise,
we may forget the relation $\rel_b'$ (as all conditions imposed by it are already satisfied).
Moreover we may also forget almost all information carried by the set $Y_b'$, except for the
fact whether $D \subseteq Y_b'$. This is done in the following step.

\begin{step}\label{steinercut-clean-rel}
Terminate the current branch if one of the following conditions is satisfied:
\begin{enumerate}
\item There exists a equivalence class of $\rel_b'$ that is different than $D$ and contains at least two border terminals.
\item There is a vertex $v \in \bterms' \setminus D$, such that $v$ is exactly in one of the sets $T'$ and $Y_b'$.
\item $D \neq \emptyset$, $D \cap Y_b' = \emptyset$ and $b \in T'$ after Step \ref{step:steinercut-borders} is applied.
\end{enumerate}
Otherwise, denote $\alpha = \bot$ if $D \neq \emptyset$ and $D \cap Y_b' = \emptyset$, and $\alpha = \top$ otherwise.
\end{step}
We note that, from the minimality of $X$ and the connectivity of $G$,
we have that any connected component of $G \setminus X$ that does not contain a border terminal, contains a terminal from $T$.
Indeed, otherwise, if $G \setminus X$ contains a connected component $C$ that does not contain a terminal nor a border terminal,
one edge incident to $C$ may be removed from $X$ and still $X$ would be a solution to $(\instance_b,\cP)$, a contradiction.
Therefore, if $D = \emptyset$, we may assume that the connected component of $G \setminus X$ that contains $\iota^{-1}(b)$,
contains at least one terminal.
 
From now on we know that all the remaining components $B_i'$ do not contain border terminals.
Without loss of generality, let $B_1',B_2',\ldots,B_{p'}'$ be the remaining components.
For every remaining component $B_i$ we have two numbers: $a_i=|E(B_i)|$, the cost of incorporating it to the solution, and $b_i=|V(B_i')\cap T'|$,
the number of separated terminals.
Computation of $a_i,b_i$ can be done in $O(\cutsize n)$ time.
We would like to know what is the optimal number of edges needed to separate exactly $\numcc_1 = \numbcc_b - \numcc_0$ connected components with terminals,
with the additional constraint that the connected component containing $b$ contains a terminal if and only if $\alpha = \top$.

This can be solved in time $O(\numbcc_b p')$ via a standard dynamic programming routine.
We create a 3-dimensional table $T[j,\ell,\ett]$ for $\ell=0,1,\ldots,\numcc_1$, $j=0,1,\ldots,p'$, $\ett=\{\bot,\top\}$ with the following meaning: $T[j,\ell,\ett]$
is the minimum cost of a solution contained in the prefix $B_1,B_2,\ldots,B_j$ that separates exactly $\ell$
isolated vertices being terminals and $\ett$ denotes whether the remaining connected component with $b$ contains a terminal different than $b$ (or $+\infty$ if such a solution does not exist).
Formally,
$$T[j,\ell,\ett]=\min\left\{\sum_{\gamma \in \Gamma}a_\gamma\ |\ \Gamma \subseteq \{1,2,\ldots,j\} \wedge \sum_{\gamma \in \Gamma} b_\gamma=\ell \wedge (\ett=\top \Leftrightarrow \sum_{\gamma \in \{1,2,\ldots,j\}\setminus \Gamma} b_\gamma >0) \right\}.$$
Observe that $T$ admits the following recurrential formula (by somewhat abusing notation, we assume that cells of $T$ with negative coordinates contain $+\infty$):
$$T[j,\ell,\bot]=\begin{cases}+\infty \textrm{ if } j=0 \textrm{ and } \ell > 0, \\ 0 \textrm{ if } j=\ell=0, \\ \min(T[j-1,\ell,\bot],a_j+T[j-1,\ell-b_j,\bot]) \textrm{ if } j>0 \textrm{ and } b_j=0, \\ a_j+T[j-1,\ell-b_j,\bot] \textrm{ otherwise.}\end{cases}$$
$$T[j,\ell,\top]=\begin{cases}+\infty \textrm{ if } j=0, \\ \min(T[j-1,\ell,\top],a_j+T[j-1,\ell-b_j,\top]) \textrm{ if } j>0, \textrm{ and } b_j=0,  \\ \min(T[j-1,\ell,\bot],T[j-1,\ell,\top],a_j+T[j-1,\ell-b_j,\top]) \textrm{ otherwise.}\end{cases}$$
Hence, we can fill the table $T$ in time $O(\numcc_1 p')=O(\cutsize n)$; the optimal value can be deduced from the cells $T[p',\numcc_1,\bot]$ and $T[p',\numcc_1,\top]$. Although we presented here only the algorithm for computing the optimal value, it is straightforward to implement the dynamic program so that it also maintains backlinks via which one can retrieve the corresponding set $\Gamma$ from the definition of $T$. Thus, we can formally present the final step of our algorithm.

\begin{step}\label{step:steinercut-core-term}
Compute numbers $a_i$ and $b_i$ and fill table $T$ in $O(\cutsize n)$ time. Let $\ett \in \{\bot,\top\}$ be defined as: $\ett = \bot$ if $\alpha = \bot$,
  $\ett = \top$ if $\alpha=\top$ and $b \notin T'$, and otherwise pick $\ett \in \{\bot,\top\}$ to minimize the value $T[p',\numcc_1,\ett]$.
  Let $\Gamma \subseteq \{1,2,\ldots,p'\}$ be the set from the definition of the value $T[p',\numcc_1,\ett]$, computed in $O(n)$ time by following backlinks in the table $T$.
  If the value $T[p',\numcc_1,\ett]$ exceeds the remaining budget $\cutsize$, terminate the branch.
  Otherwise, incorporate the set $\bigcup_{\gamma\in \Gamma} E(B_\gamma)$ to the constructed solution.
\end{step}

Step \ref{step:steinercut-core-term} can be performed in $O(\cutsize n)$ time and its correctness follows from the definition of the table $T$ and the previous steps of the algorithm.

This finishes the description of the fixed-parameter algorithm for \steinercut and the proof of Theorem~\ref{thm:steiner-main}.

\newcommand{\opccl}{\mathtt{cc}}
\newcommand{\ccrel}{\ell}

\section{The algorithm for \nmwcu}\label{sec:full-uncut}

In this section we show an FPT algorithm for the following generalization of the well-known \textsc{Multiway Cut} problem.

\defparproblemu{\nmwcu (\nmwcushort)}{A graph $G$ together with a set of terminals $T\subseteq V(G)$, an equivalence relation $\rel$ on the set $T$, and an integer $\cutsize$.}{$\cutsize$}{Does there exist a set $X \subseteq V(G) \setminus T$ of at most $\cutsize$ nonterminals such that for any $u,v \in T$, the vertices $u$ and $v$ belong to the same connected component of $G \setminus X$ if and only if $(u,v) \in \rel$?}

In other words, we are to delete at most $k$ vertices from the graph, so that the terminals are split between connected components
exactly as it is given by the equivalence relation $\rel$.
Given a \nmwcushort instance $\instance=(G,T,\rel,\cutsize)$, a set of vertices $X$ is 
called a {\em{solution}} to $\instance$, if $|X| \leq \cutsize$ and for any $u,v \in T$,
the vertices $u$ and $v$ belong to the same connected component of $G \setminus X$
if and only if $(u,v) \in \rel$.

Our algorithm not only resolves \nmwcushort{} instance, but, in the case of a YES answer,
it returns a solution $X$ with minimum possible $|X|$. This property will be used
in the course of the algorithm.

%Before we start, let us note that in the \nmwcushort definition we do not allow $X$ to contain
%any terminals.
%Note that, given an instance $\instance = (G,T,\rel,\cutsize)$
%of \nmwcushort with deletable terminals,
%we can easily reduce it to an instance of the undeletable terminals variant, by attaching
%a degree-1 neighbour $v'$ to each $v \in T$; the new terminal set is $T' = \{v' : v \in T\}$
%and the relation is $\rel' = \{(u',v'): (u,v) \in \rel\}$.

\subsection{Reduction of the number of equivalence classes}

We now show how to reduce the number of equivalence classes of $\rel$
in an \nmwcushort instance $\instance = (G,T,\rel,\cutsize)$.
We use a reduction similar to the one used by Razgon~\cite[Theorem 5]{razgon:mwc-k2-terms}.

\begin{lemma}\label{lem:uncut-manypaths}
Let $\instance = (G,T,\rel,\cutsize)$ be an \nmwcushort instance
and let $v \in V(G) \setminus T$. Assume that there exist $\cutsize+2$ paths
$P_1,P_2,\ldots,P_{\cutsize+2}$ in $G$, such that:
\begin{itemize}
\item for each $1 \leq i \leq \cutsize+2$, the path $P_i$ is a simple path that starts at $v$ and ends at $v_i \in T$;
\item the paths $P_i$ have pairwise disjoint sets of vertices, except for the vertex $v$;
\item for any $i \neq j$, $(v_i,v_j) \notin \rel$.
\end{itemize}
Then for any solution $X$ in $\instance$ we have $v \in X$.
\end{lemma}
\begin{proof}
Let $X \subseteq V(G) \setminus T$ with $v \notin X$ and $|X| \leq \cutsize$.
As the paths $P_i$ are disjoint (except for $v$), there exist two indices $1 \leq i < j \leq \cutsize+2$, such that $P_i$ and $P_j$ does not contain any vertex from $X$. A concatenation
of $P_i$ and $P_j$ is a path from $v_i$ to $v_j$ that avoids $X$. As $(v_i,v_j) \notin \rel$,
   we infer that $X$ is not a solution to $\instance$.
\end{proof}

\begin{lemma}\label{lem:uncut-find-manypaths}
Let $\instance = (G,T,\rel,\cutsize)$ be an \nmwcushort instance.
For any $v \in V(G)$, we can verify if $v$ satisfies the conditions of
Lemma \ref{lem:uncut-manypaths} in $O(\cutsize n^2)$ time.
\end{lemma}
\begin{proof}
Consider the following auxiliary graph $H$. For each equivalence class $A \subseteq T$ of
$\rel$, we attach a new vertex $t_A$ that is a adjacent to all vertices of $A$.
We make $v$ an infinite-capacity source and each vertex $t_A$ a unit-capacity sink;
each other vertex of $H$ has unit capacity. Clearly, $v$ satisfies the conditions
of Lemma \ref{lem:uncut-manypaths} iff there exists a flow of size at least $\cutsize+2$
in $H$. As vertices in $H$ have unit capacities (except for $v$), this can be done
in $O(\cutsize n^2)$ time by the classic Ford-Fulkerson algorithm.
\end{proof}

Lemma \ref{lem:uncut-manypaths} justifies the following step.
\begin{step}\label{step:uncut-manypaths}
For each $v \in V(G)$, if $v$ satisfies the conditions of Lemma \ref{lem:uncut-manypaths},
delete $v$ from the graph and decrease $\cutsize$ by one; if $\cutsize$
becomes negative by this operation, return NO. Afterwards, restart the algorithm.
\end{step}
By Lemma \ref{lem:uncut-find-manypaths}, each application of Step \ref{step:uncut-manypaths}
takes $O(\cutsize n^3)$ time. As we cannot apply Step \ref{step:uncut-manypaths}
more than $\cutsize$ times, all applications of this step take $O(\cutsize^2 n^3)$ time.

Let us now show that Step \ref{step:uncut-manypaths} leads to a bound
on the number of equivalence classes of $\rel$.

\begin{lemma}\label{lem:uncut-after-manypaths}
Let $\instance = (G,T,\rel,\cutsize)$ be a \nmwcushort instance
where Step \ref{step:uncut-manypaths} is not applicable.
If there exists a connected component of $G$ that contains terminals
of more than $\cutsize^2 + \cutsize$ equivalence classes of $\rel$,
then $\instance$ is a NO-instance to \nmwcushort.
\end{lemma}
\begin{proof}
Let $C$ be the vertex set of any connected component of $G$,
and let $X$ be a solution to $\instance$. Fix arbitrary $v \in X \cap C$.
We say that $v$ {\em{sees}} an equivalence class $A$ of $\rel$
if there exists a connected component $C_A$ of $G[C \setminus X]$ that contains
a terminal of $A$ and such that $N_G(v) \cap C_A \neq \emptyset$. Note that if 
$v$ sees $\cutsize+2$ equivalence classes of $\rel$, then in each component $C_A$
for each equivalence class $A$ seen by $v$ we can find a path from $v$
to a terminal of $A$. Thus $v$ satisfies the assumptions of Lemma \ref{lem:uncut-manypaths},
as $C_A \neq C_B$ for $A \neq B$ (recall that $X$ is a solution to $\instance$).
Therefore, each vertex $v \in X$ sees at most $\cutsize+1$ equivalence classes of $\rel$.
From connectivity of $G[C]$ we infer that each component $C_A$ must be seen by some element of $X$, so $C \setminus X$ may contain vertices
of at most $\cutsize(\cutsize+1)$ equivalence classes of $\rel$.
\end{proof}

\begin{step}\label{step:uncut-after-manypaths}
If there exist $u,v \in T$ such that $u$ and $v$ lie in different connected components
of $G$, but $(u,v) \in \rel$, or there exists a connected component of $G$
with terminals of more than $\cutsize^2+\cutsize$ equivalence classes or $\rel$,
return NO.
\end{step}
Clearly, Step \ref{step:uncut-after-manypaths} can be applied in $O(n^2)$ time.

We now ensure connectivity of $G$, by considering separately all connected components.
Recall that we are developing an algorithm that not only resolves the given
\nmwcushort{} instance, but also, in case of the positive answer, returns 
a solution of minimum possible size.
\begin{step}\label{step:uncut-cc}
For each connected component of $G$ with vertex set $C$,
pass the instance $(G,T \cap C, \rel|_{T \cap C},\cutsize)$ to the next step.
If any of the subinstances returns NO, or if the union of the solutions to the subcases
is larger than $\cutsize$, return NO. Otherwise, return
YES and the union of the solutions for the connected components as the solution
to the given instance.
\end{step}
The correctness of Step \ref{step:uncut-cc} is straightforward
(note that Step \ref{step:uncut-after-manypaths} refutes instances
 where one equivalence class of $\rel$ is scattered among more than one
 connected component of $G$) and splitting $G$ into connected components
takes linear time in the size of $G$.
Thus, from this point we may assume that $G$ is connected and that the
number of equivalence classes or $\rel$ is bounded by $\ccrel := \cutsize^2 + \cutsize$.

\subsection{Operations on the input graph}

In this section we show basic operations the algorithm repetitively applies
to the graph.

\begin{definition}\label{def:nmwcu-bypass}
Let $\instance = (G,T,\rel,\cutsize)$ be an \nmwcushort
instance and let $v \in V(G) \setminus T$.
By {\em{bypassing}} a vertex $v$ we mean the following operation:
we delete the vertex $v$ from the graph and, for any $u_1,u_2 \in N_G(v)$,
we add an edge $u_1u_2$ if it is not already present in $G$.
\end{definition}

We now state the properties of the bypassing operation.

\begin{lemma}\label{lem:nmwcu-bypass}
Let $\instance = (G,T,\rel,\cutsize)$ be an \nmwcushort instance,
let $v \in V(G) \setminus T$ and
let $\instance'= (G',T,\rel,\cutsize)$ be the instance
$\instance$ with $v$ bypassed. Then:
\begin{itemize}
\item if $X$ is a solution to $\instance'$, then
$X$ is a solution to $\instance$ as well;
\item if $X$ is a solution to $\instance$ and $v \notin X$
then $X$ is a solution to $\instance'$ as well.
\end{itemize}
\end{lemma}
\begin{proof}
The claim follows from the following correspondence of the paths in $G$ and $G'$:
any path $P'$ in $G'$ has a corresponding walk $P$ in $G$,
where each occurrence of an edge of $E(G') \setminus E(G)$ is replaced with a length-2
subpath via $v$. Moreover, any path $P$ in $G$ that does not start nor end in $v$
has a corresponding path $P'$ in $G'$,
where a possible occurrence of $v$ is circumvented by an edge in $G'$ between two neighbours
of $v$.
\end{proof}

Apart from the bypassing operation, we need to show a way to reduce the number of terminals.
\begin{definition}
Let $\instance = (G,T,\rel,\cutsize)$ be an \nmwcushort instance and let $u,v \in T$ be two terminals
with $u \neq v$, $(u,v) \in \rel$. By {\em{identifying}} $u$ and $v$ we mean the following operation:
we replace vertices $u$ and $v$ with a new vertex $w_{uv}$ that is adjacent to all vertices of $N_G(u) \cup N_G(v)$.
Moreover, we update $\rel$ by substituting $u$ and $v$ with $w$ in the equivalence class they belong to.
\end{definition}

\begin{lemma}\label{lem:nmwcu-ident-terms}
Let $\instance = (G,T,\rel,\cutsize)$ be an \nmwcushort instance and let $u,v \in T$ be two different terminals
with $(u,v) \in \rel$, such that $uv \in E(G)$ or $|N_G(u) \cap N_G(v)| > \cutsize$. Let $\instance'$ be instance $\instance$ with terminals
$u$ and $v$ identified. Then the set of solutions to $\instance'$ and $\instance$ are equal.
\end{lemma}
\begin{proof}
The lemma follows from the fact that, for any $X \subseteq V(G) \setminus T$ of size at most $\cutsize$, in $G \setminus X$ the vertices
$u$ and $v$ lie in the same connected component.
\end{proof}

\begin{lemma}\label{lem:nmwcu-del-terms}
Let $\instance = (G,T,\rel,\cutsize)$ be a \nmwcushort instance and let $u_1,u_2,u_3 \in T$ be three different terminals
of the same equivalence class of $\rel$, pairwise nonadjacent and such that $N_G(u_1) = N_G(u_2) = N_G(u_3) \subseteq V(G) \setminus T$.
Let $\instance'$ be obtained from $\instance$ by deleting the terminal $u_3$ (and all pairs that contain $u_3$ in $\rel$).
Then the set of solutions to $\instance'$ and $\instance$ are equal.
\end{lemma}
\begin{proof}
Let $X \subseteq V(G) \setminus T$. We claim that for any $u,v \in V(G) \setminus \{u_3\}$, $u$ and $v$ are in the same connected component of $G \setminus X$
if and only if they are in the same connected component of $G' \setminus X$. Indeed, the backward implication is trivial, whereas for the forward implication observe that any path
from $u$ to $v$ in $G \setminus X$ that visits $u_3$ can be redirected via $u_1$ or $u_2$.

The proven equivalence already shows that any solution to $\instance$ is a solution to $\instance'$ as well. For the other direction, we need to additionally verify that
for any $v \in T$, we have $(u_3,v) \in \rel$ if and only if $v$ and $u_3$ are in the same connected component of $G \setminus X$, assuming that $X$ is a solution to $\instance'$. As $X$ is a solution to $\instance'$
and $(u_1,u_2) \in \rel$, there exists $w \in N_G(u_1) \setminus X$. Therefore $u_1$, $u_2$ and $u_3$ are in the same connected component of $G \setminus X$.
As $(u_1,v) \in \rel$ iff $(u_3,v) \in \rel$, the claim follows.
\end{proof}

\subsection{Borders and recursive understanding}

For the recursive understanding phase, we need to define the bordered problem.
Let $\instance=(G,T,\rel,\cutsize)$ be an \nmwcushort instance
and let $\bterms \subseteq V(G) \setminus T$ be a set of border terminals; we assume
$|\bterms| \leq 2\cutsize$.
Define $\instance_b = (G,T,\rel,\cutsize,\bterms)$ to be an instance of the
bordered problem.
By $\mathbb{P}(\instance_b)$ we define the set of all triples $\cP=(X_b,E_b,\rel_b)$, such that
$X_b \subseteq \bterms$, $E_b$ is an equivalence relation on $\bterms \setminus X_b$
and $\rel_b$ is an equivalence
relation on $T \cup (\bterms \setminus X_b)$ such that
$E_b \subseteq \rel_b$ and $\rel_b|_T = \rel$.
For a triple $\cP = (X_b,E_b,\rel_b)$, by $G_\cP$ we denote the graph $G \cup E_b$, that is, the graph $G$ with additional edges $E_b$.

We say that a set $X \subseteq V(G) \setminus T$
is a solution to $(\instance_b,\cP)$ if
$|X| \leq \cutsize$, $X \cap \bterms = X_b$ and for any $u,v \in T \cup (\bterms \setminus X_b)$,
the vertices $u$ and $v$ are in the same connected component of the graph
$G_\cP \setminus X$ (i.e., we delete vertices $X$ and add edges $E_b$)
if and only if $(u,v) \in \rel_b$.

We also say that $X$ is a solution to $\instance_b = (G,T,\rel,\cutsize,\bterms)$ whenever $X$ is a solution to $\instance = (G,T,\rel,\cutsize)$
Note that, if $X$ is a solution to $(\instance_b,\cP)$,
the set $X$ is not necessarily a solution to $\instance_b$;
however, $X$ is a solution to the \nmwcushort instance $(G_\cP,T,\rel,\cutsize)$.

One may think of the set of edges $E_b$ as the ``prediction'' which vertices will be connected \emph{outside} the currently considered part of the graph, after the solution edges has been deleted.
Since such a definition may not be very intuitive, we provide detailed proofs of all equivalences in this section; note that a corresponding equivalence claims in the previous sections were nearly straightforward.

We formally define the bordered problem as follows.

\defproblemoutput{\bnmwcu}{An \nmwcushort instance $\instance=(G,T,\rel,\cutsize)$
with $G$ being connected and a set $\bterms \subseteq V(G) \setminus T$ of size at most $2\cutsize$;
denote $\instance_b = (G,T,\rel,\cutsize,\bterms)$}{
For each $\cP = (X_b,E_b,\rel_b) \in \mathbb{P}(\instance_b)$, output
a $\sol_\cP = X_\cP$ being a solution to
$(\instance_b,\cP)$ with minimum possible $|X_\cP|$,
or $\sol_\cP = \bot$ if such a solution does not exist.}

Clearly, \nmwcushort reduces to \bnmwcu, as we may ask for an instance
with $\bterms = \emptyset$. Moreover, in this case the single
answer to \bnmwcu for $\cP = (\emptyset,\emptyset,\rel)$ returns a solution of minimum possible size.

We note that
$$|\mathbb{P}(\instance_b)| \leq (1+|\bterms|(|\bterms| + \ccrel))^{|\bterms|}
\leq (2\cutsize^3+6\cutsize^2+1)^{2\cutsize} = 2^{O(\cutsize \log \cutsize)},$$
as $\rel_b$ has at most $\ccrel + |\bterms|$
equivalence classes, $E_b$ has at most $|\bterms|$ equivalence classes,
and each $v \in \bterms$ can go either to $X_b$ or choose an equivalence
class in $\rel_b$ and $E_b$.
Let $q = k(2\cutsize^3+6\cutsize^2+1)^{2\cutsize} + \cutsize$;
all output solutions to a \bnmwcu instance $\instance_b$ contain
at most $q-\cutsize$ vertices in total.

Armed with the previous definitions, we are now ready to proof a somewhat expected at this point lemma showing
that if a \bnmwcu instance $\instance_b$ contains another, smaller subinstance $\newinst{\instance}_b$,
then it suffices to restrict our attention to only one, fixed output to \bnmwcu on $\newinst{\instance}_b$.
In other words, we show that there exists a valid output to \bnmwcu on $\instance_b$ that, on $\newinst{\instance}_b$,
behaves as prescribed by the aforementioned fixed output.
Due to an involved definition of $E_b$ as a connectivity prediction, the formal proof is quite involved.

\begin{lemma}\label{lem:nmwcu-rekur}
Assume we are given a \bnmwcu instance
$\instance_b = (G,T,\rel,\cutsize,\bterms)$,
a set $Z \subseteq V(G) \setminus T$ with $|Z| \leq \cutsize$, 
and a connected component $\newinst{V}$ of $G-Z$.
Denote $Z_W := N_G(\newinst{V}) \subseteq Z$ and $W :=\newinst{V} \cup Z_W$,
and assume furthermore that $G[W]$ is connected and that $|\newinst{V} \cap \bterms| \leq \cutsize$.

Denote
$\newinst{G} = G[W]$,
$\newinst{\bterms} = (\bterms \cup Z_W) \cap W$,
$\newinst{T} = T \cap W$, $\newinst{\rel} = \rel|_{T \cap W}$ and
$\newinst{\instance}_b = (\newinst{G}, \newinst{T}, \newinst{\rel},\cutsize,\newinst{\bterms})$.
Then $\newinst{\instance}_b$ is a proper \bnmwcu instance.
Moreover, if we denote by 
$(\newinst{\sol}_{\newinst{\cP}})_{\newinst{\cP} \in \mathbb{P}(\newinst{\instance}_b)}$ an arbitrary output
to the \bnmwcu{} instance $\newinst{\instance}_b$ and
$$U(\newinst{\instance}_b) = \newinst{\bterms} \cup \bigcup \{\newinst{X}_{\newinst{\cP}} : \newinst{\cP} \in \mathbb{P}(\newinst{\instance}_b), \newinst{\sol}_{\newinst{\cP}} = \newinst{X}_{\newinst{\cP}} \neq \bot\}$$
to be a set of vertices used by the solutions of
$(\newinst{\sol}_{\newinst{\cP}})_{\newinst{\cP} \in \mathbb{P}(\newinst{\instance}_b)}$,
then there exists a correct output $(\sol_\cP)_{\cP \in \mathbb{P}(\instance_b)}$
to the \bnmwcu instance $\instance_b$ such that whenever $\sol_\cP=X_\cP \neq \bot$
then $X_\cP \cap \newinst{V} \subseteq U(\newinst{\instance}_b)$.
\end{lemma}
\begin{proof}
The claim that $\newinst{\instance}_b$ is a proper \bnmwcu instance follows directly
from the assumptions that $\newinst{G} = G[W]$ is connected, $|Z_W| \leq |Z| \leq\cutsize$
and $|\newinst{V} \cap \bterms| \leq \cutsize$. In the rest of the proof we justify the second claim
of the lemma.

Fix $\cP = (X_b,E_b,\rel_b) \in \mathbb{P}(\instance_b)$ and recall that $G_\cP = G \cup E_b$.
Assume that there exists a solution
to the instance $(\instance_b,\cP)$; let $X_\cP$
be such a solution with minimum possible $|X_\cP|$. To prove the lemma we need to show
a second solution $X_\cP'$ to $(\instance_b,\cP)$,
$|X_\cP'| \leq |X_\cP|$ and $X_\cP' \cap \newinst{V} \subseteq U(\newinst{\instance}_b)$.

Let us first give an intuition of the proof. We are given a solution $X_\cP$; our goal is to modify it on $\newinst{V}$
so that it behaves on $\newinst{\instance}_b$ as predicted by the output
$(\newinst{\sol}_{\newinst{\cP}})_{\newinst{\cP} \in \mathbb{P}(\newinst{\instance}_b)}$.
To this end, we observe how $X_\cP \cap W$ behaves on $\newinst{\instance}_b$
and define a triple $\newinst{\cP} \in \mathbb{P}(\newinst{\instance}_b)$
so that $X_\cP \cap W$ is a solution to $(\newinst{\instance}_b,\newinst{\cP})$.
The information stored in the triple $\newinst{\cP}$ should be enough to argue
that swapping $X_\cP \cap W$ with $\newinst{\sol}_{\newinst{\cP}}$ does not invalidate $X_\cP$
as a solution to $(\instance_b,\cP)$.

Let us now proceed with this strategy in full detail.
We start by defining a triple $\newinst{\cP} = (\newinst{X}_b,\newinst{E}_b,\newinst{\rel}_b)$ that represents how $X_\cP \cap W$ behaves on $\newinst{\instance}_b$.
\begin{itemize}
\item As $\newinst{X}_b$ is meant to keep information on deleted border terminals,
its definition is straightforward. We take $\newinst{X}_b = X_\cP \cap \newinst{\bterms}$; note that $X_b \cap W \subseteq \newinst{X}_b$
since $X_\cP$ is a solution to $(\instance_b,\cP)$.
\item Recall that $\newinst{E}_b$ is meant to predict connectivity between border terminals outside the
instance $\newinst{\instance}_b$. Consequently, in the definition of $\newinst{E}_b$ we need to take into account both the graph $G-\newinst{V}$ after the deletion
of $X_\cP$, as well as the edges $E_b$.
Formally, we define $\newinst{E}_b$ to be the following relation
on $\newinst{\bterms} \setminus \newinst{X}_b$: $(u,v) \in \newinst{E}_b$
iff $u$ and $v$ are in the same connected component of $G_\cP \setminus ((\newinst{V} \setminus \bterms) \cup X_\cP)$
(in particular, if $(u,v) \in E_b$).
\item Recall that $\newinst{\rel}_b$ is meant as a requirement on the final connectivity of the terminals and border terminals in the entire graph;
to this end, we can use connectivity in $G_\cP \setminus X_\cP$. 
That is, we define $\newinst{\rel}_b$ to be the following relation on $\newinst{T} \cup (\newinst{\bterms}\setminus \newinst{X}_b)$: $(u,v) \in \rel_b$ iff $u$ and $v$ are in the same connected
component of $G_\cP \setminus X_\cP$.
As $X_\cP$ is a solution to $(\instance_b,\cP)$,
$\newinst{\rel}_b |_{\newinst{T}} = \newinst{\rel} = \rel|_{\newinst{T}}$.
\end{itemize}
Note that $\newinst{E}_b \subseteq \newinst{\rel}_b$, as both $\newinst{E}_b$ and $\newinst{\rel}_b$
corresponds to the relation of being in the same connected component, but in $\newinst{E}_b$ we consider a smaller graph than in $\newinst{\rel}_b$.
This justifies that $\newinst{\cP} \in \mathbb{P}(\newinst{\instance}_b)$.

The main idea of the definition of $\newinst{\cP}$ is that $X_\cP \cap W$ is a solution to $(\newinst{\instance}_b,\newinst{\cP})$.
Let us now verify it formally.
Clearly, $|X_\cP \cap W| \leq \cutsize$.
By the definition of $\newinst{X}_b$, we have $X_\cP \cap W \cap \newinst{\bterms} = \newinst{X}_b$.
Consider two vertices $u,v \in \newinst{T} \cup (\newinst{\bterms} \setminus \newinst{X}_b)$.
We have $(u,v) \in \newinst{\rel}_b$ iff there exists a path $P$ between $u$ and $v$ in
$G_\cP \setminus X_\cP$. By the definition of $\newinst{E}_b$, such a path $P$ exists iff
there exists a path $\newinst{P}$ connecting $u$ and $v$ in $\newinst{G}_{\newinst{\cP}} \setminus X_\cP$:
each subpath of $P$ with internal vertices in $V(G) \setminus W$ corresponds
to an edge in $\newinst{E}_b$ and vice versa. Thus, $u$ and $v$ are in the same connected
component of $\newinst{G}_{\newinst{\cP}} \setminus X_\cP$ if and only if $(u,v) \in \newinst{\rel}_b$
and the claim is proven.

To wrap up, we have defined an element $\newinst{\cP} \in \mathbb{P}(\newinst{\instance}_b)$ that represents the behavior
of $X_\cP$ on $\newinst{\instance}_b$. Our goal now is to show that if we swap $X_\cP \cap W$ with 
$\newinst{\sol}_{\newinst{\cP}}$, that is, the prescribed solution to $(\newinst{\instance}_b,\newinst{\cP})$, we obtain another solution
to $(\instance_b,\cP)$ that fulfills our requirements.

Since $X_\cP \cap W$ is a solution to $(\newinst{\instance}_b,\newinst{\cP})$,
we infer that $\newinst{\sol}_{\newinst{\cP}} = \newinst{X}_{\newinst{\cP}} \neq \bot$
and $|\newinst{X}_{\newinst{\cP}}| \leq X_\cP \cap W$.
Let us define our new solution to $(\instance_b,\cP)$ as $X_\cP' = (X_\cP \setminus W) \cup \newinst{X}_{\newinst{\cP}}$.
Clearly, $|X_\cP'| \leq |X_\cP|$.
To finish the proof of the lemma, we need to formally show that indeed $X_\cP'$ is a solution to $(\instance_b,\cP)$.

It is straightforward to verify that $X_\cP'$ satisfies the constaint imposed by $X_b$:
As $\newinst{X}_b$ is defined as $X_\cP \cap \newinst{\bterms}$ and $X_\cP$
is a solution to $(\instance_b,\cP)$, we have $X_\cP' \cap \bterms = X_b$.

It remains to check the connectivity requirement imposed by $\rel_b$.
Let $u,v \in T \cup (\bterms \setminus X_b)$.
Our goal is to show that $u$ and $v$ lie in the same connected component of $G_\cP \setminus X_\cP'$ if and only
if they lie in the same connected component of $G_\cP \setminus X_\cP$.
We present the proof only in one direction, as the proofs in both directions are totally symmetric:
we use only the facts that both $X_\cP \cap W$ and $X_\cP' \cap W$ are solutions to $(\newinst{\instance}_b, \newinst{\cP})$ and
that $X_\cP \setminus \newinst{V} = X_\cP' \setminus \newinst{V}$. The last equality holds because $Z_W\subseteq \newinst{\bterms}$, so $Z_W\cap X_\cP=Z_W\cap \newinst{X}_{\newinst{\cP}}$.

Thus, assume that $u,v \in T \cup (\bterms \setminus X_b)$ and $u$ and $v$ lie in the same connected component of $G_\cP \setminus X_\cP$.
Let $P$ be a path that connects $u$ and $v$ in $G_\cP \setminus X_\cP$. 
The remainder of the proof works as follows: we partition $P$ into parts that live entirely in $G_\cP \setminus \newinst{V}$
and in $G_\cP[W]$, and argue that each such part of $P$ has its counterpart in $G_\cP \setminus X_\cP'$.

Formally, let $u=v_0,v_1,v_2,\ldots,v_r=v$ be a sequence of vertices that lie on the path $P$ and belong to $D := T \cup (\bterms \setminus X_b) \cup Z_W$,
in the order they appear on $P$.
First note that, since $X_\cP \setminus \newinst{V} = X_\cP' \setminus \newinst{V}$ and both $X_\cP \cap W$ and $X_\cP' \cap W$ are solutions to $(\newinst{\instance}_b, \newinst{\cP})$,
we have that $X_\cP \cap D = X_\cP' \cap D = X_b \cup \newinst{X}_b$. Thus, for each $0 \leq i \leq r$, we have $v_i \notin X_\cP'$.
To finish the proof of the lemma we need to show that for any $0 \leq i < r$, the vertices $v_i$ and $v_{i+1}$ lie in the same connected component of $G_\cP \setminus X_\cP'$.

Let $P_i$ be the subpath of $P$ between $v_i$ and $v_{i+1}$. As $Z_W \subseteq D$, $P_i$ is either a path in $G_\cP \setminus \newinst{V}$ or a path in $G_\cP[W]$.
In the first case, since $X_\cP \setminus \newinst{V} = X_\cP' \setminus \newinst{V}$, we infer that the path $P_i$ is present in $G_\cP \setminus X_\cP'$ and the claim is proven.
In the second case, note that we have $(v_i,v_{i+1}) \in \newinst{\rel}_b$. As $X_\cP' \cap W$ is a solution to $(\newinst{\instance}_b, \newinst{\cP})$, we infer that
$v_i$ and $v_{i+1}$ are connected via a path $\newinst{P}_i$ in $\newinst{G}_{\newinst{\cP}} \setminus (X_\cP' \cap W)$. However, by the definition of $\newinst{E}_b$,
  for any edge $w_1w_2 \in \newinst{E}_b$ on $\newinst{P}_i$, the vertices $w_1$ and $w_2$ are in the same connected component of $G_\cP \setminus ((\newinst{V} \setminus \bterms) \cup X_\cP)$.
Since $X_\cP \setminus \newinst{V} = X_\cP' \setminus \newinst{V}$ and $X_\cP \cap \bterms = X_\cP' \cap \bterms$, we have that $X_\cP \setminus (\newinst{V} \setminus \bterms) = X_\cP' \setminus (\newinst{V} \setminus \bterms)$
and the claim is proven. This finishes the proof of the lemma.
\end{proof}

Note that in Lemma \ref{lem:nmwcu-rekur} we have $|U(\newinst{\instance}_b) \cap \newinst{V}|  \leq q$.

A recursive call due to an application of Lemma \ref{lem:nmwcu-rekur} allows us to reduce the number of nonterminal vertices in $\newinst{V}$ to at most $q = 2^{O(\cutsize \log \cutsize)}$.
To make the recursion work in FPT time, we need to reduce the number of terminals as well. Fortunately, this is quite easy, due to the identifying operation and Lemma \ref{lem:nmwcu-ident-terms}.

We are now ready to present the recursive step of the algorithm.

\begin{step}\label{step:nmwcu-rekur}
Assume we are given a \bnmwcu{} instance $\instance_b = (G,T,\rel,\cutsize,\bterms)$.
Invoke first the algorithm of Lemma \ref{lem:detect-good-node} in a search for $(q,\cutsize)$-good node separation (with $\undelV = T$).
If it returns a good node separation $(Z,V_1,V_2)$, let $j \in \{1,2\}$ be such that $|V_j \cap \bterms| \leq \cutsize$ and denote $\newinst{Z} = Z$, $\newinst{V} = V_j$.
Otherwise, if it returns that no such good node separation exists in $G$,
invoke the algorithm of Lemma \ref{lem:detect-flower-cut} in a search for $(q,\cutsize)$-flower separation w.r.t. $\bterms$ (with $\undelV = T$ again).
If it returns that no such flower separation exists in $G$,
pass the instance $\instance_b$ to the next step. Otherwise, if it returns a flower separation $(Z,(V_i)_{i=1}^\ell)$, denote $\newinst{Z} = Z$ and $\newinst{V} = \bigcup_{i=1}^\ell V_i$.

In the case we have obtained $\newinst{Z}$ and $\newinst{V}$ (either from Lemma \ref{lem:detect-good-node} or Lemma \ref{lem:detect-flower-cut}), 
invoke the algorithm recursively for the \bnmwcu{} instance $\newinst{\instance}_b$ defined as in the statement
of Lemma \ref{lem:nmwcu-rekur} for separator $\newinst{Z}$ and set $\newinst{V}$, obtaining an output $(\sol_{\newinst{\cP}})_{\cP \in \mathbb{P}(\newinst{\instance}_b)}$.
Compute the set $U(\newinst{\instance}_b)$. Bypass (in an arbitrary order) all vertices of $\newinst{V} \setminus (T \cup U(\newinst{\instance}_b))$. Recall that $\newinst{\bterms}\subseteq U(\newinst{\instance}_b)$, so no border terminal gets bypassed.

After all vertices of $\newinst{V} \setminus U(\newinst{\instance}_b)$ are bypassed, perform the following operations on terminals of $\newinst{V} \cap T$:
\begin{enumerate}
\item As long as there exist two different $u,v \in \newinst{V} \cap T$ that satisfy $uv \in E(G)$ or $|N_G(u) \cap N_G(v)| > \cutsize$ do as follows:
if $(u,v) \in \rel$, identify $u$ and $v$, and otherwise output $\bot$ for all $\cP \in \mathbb{P}(\instance_b)$.
\item If the above is not applicable, then, as long as there exist three pairwise distinct terminals $u_1,u_2,u_3 \in T$
of the same equivalence class of $\rel$ that have the same neighborhood, delete $u_3$ from the graph (and delete all pairs containing $u_3$ from $\rel$).
\end{enumerate}
Let $\instance_b'$ be the outcome instance.

Finally, restart this step on the new instance $\instance_b'$
and obtain a family of solutions $(\sol_\cP')_{\cP \in \mathbb{P}(\instance_b)}$ and return this family as an output to the instance $\instance_b$.
\end{step}

Let us first verify that the application of Lemma \ref{lem:nmwcu-rekur} is justified. Indeed, by the definitions of
the good node separation and the flower separation, as well as the choice of $\newinst{V}$, we have in both cases $|\newinst{V} \cap \bterms| \leq \cutsize$
and that $G[\newinst{V} \cup N_G(\newinst{V})]$ is connected. Moreover, note that the recursive call is applied to a graph with strictly smaller number of vertices
than $G$: in the case of a good node separation, $V_2$ is removed from the graph, and in the case of a flower separation, recall that the definition
of the flower separation requires $Z \cup \bigcup_{i=1}^\ell V_i$ to be a proper subset of $V(G)$.

We have that, after the bypassing operations, $\newinst{V}$ contains at most $q$ vertices that are not terminals (at most $\cutsize$ border terminals and at most $q-\cutsize$ vertices which are neither terminals nor border terminals). Let us now bound the number of terminal vertices
once Step \ref{step:nmwcu-rekur} is applied. Note that, after Step \ref{step:nmwcu-rekur} is applied, for any $v \in T \cap \newinst{V}$, we have $N_G(v) \subseteq (\newinst{V} \setminus T) \cup Z$
and $|(\newinst{V} \setminus T) \cup Z| \leq (q+\cutsize)$. Due to the first rule in Step \ref{step:nmwcu-rekur}, for any set $A \subseteq (\newinst{V} \setminus T) \cup Z$ of
size $\cutsize+1$, at most one terminal of $T \cap \newinst{V}$ is adjacent to all vertices of $A$. Due to the second rule in Step \ref{step:nmwcu-rekur},
for any set $B \subseteq (\newinst{V} \setminus T) \cup Z$ of size at most $\cutsize$ and for each equivalence class of $\rel$, there are at most two terminals
of this equivalence class with neighborhood exactly $B$. We infer that
$$|T \cap \newinst{V}| \leq (q+\cutsize)^{\cutsize+1} + 2\ell \sum_{i=1}^\cutsize (q+\cutsize)^i =: q'.$$
Note that $q' = 2^{O(\cutsize^2 \log \cutsize)}$.

The following lemma verifies the correctness of Step \ref{step:nmwcu-rekur}.

\begin{lemma}\label{lem:nmwcu-rekur-correctness}
Assume we are given a \bnmwcu{} instance $\instance_b = (G,T,\rel,\cutsize,\bterms)$
on which Step \ref{step:nmwcu-rekur} is applied, and let $\instance_b'$ be an instance after Step \ref{step:nmwcu-rekur} is applied.
Then any correct output to the instance $\instance_b'$ is a correct output to the instance $\instance_b$ as well.
Moreover, if Step \ref{step:nmwcu-rekur} outputs $\bot$ for all $\cP \in \mathbb{P}(\instance_b)$, then this is a correct output to $\instance_b$.
\end{lemma}
\begin{proof}
The lemma is a straightforward corollary of Lemma \ref{lem:nmwcu-rekur}, the properties of the bypassing operation described in Lemma \ref{lem:nmwcu-bypass}, and Lemmas \ref{lem:nmwcu-ident-terms} and \ref{lem:nmwcu-del-terms}.
Lemma \ref{lem:nmwcu-rekur} ensures us that each vertex not in $U(\newinst{\instance}_b)$ is omitted by some optimal solution for every $\cP\in \mathbb{P}(\instance_b)$, which enables us to use Lemma \ref{lem:nmwcu-bypass}.
Finally, if for any terminals $u,v \in T$, we have $uv \in E(G)$ or $|N_G(u) \cap N_G(v)| > \cutsize$, then $u$ and $v$ are in the same connected component of $G \setminus X$
for any set $X$ of at most $\cutsize$ nonterminals and, if $(u,v) \notin \rel$, for any $\cP \in \mathbb{P}(\instance_b)$, there is no solution to $(\instance_b,\cP)$.
\end{proof}

We are left with the analysis of the time complexity of Step \ref{step:nmwcu-rekur}.
The applications of Lemmas \ref{lem:detect-good-node} and \ref{lem:detect-flower-cut} use $O(2^{O(\min(q,\cutsize) \log(q+\cutsize))} n^3 \log n) = O(2^{O(\cutsize^2 \log \cutsize)} n^3 \log n)$
time. Let $n' = |\newinst{V}|$; the recursive step is applied to a graph with at most $n' + \cutsize$ vertices and, after bypassing, there are at most $\min(n-1,n-n'+q+q')$ vertices
left. Moreover, each bypassing operation takes $O(n^2)$ time, the computation of $U(\newinst{\instance}_b)$
takes $O(2^{O(\cutsize \log \cutsize)} n)$ time.
Application of Lemma \ref{lem:nmwcu-ident-terms} takes $O(\cutsize n^2)$ time per operation, which can be implemented by having a counter for each pair of terminals
and increasing those counters accordingly by considering every pair of terminals of $N_G(x)$, for each $x \in V$.
Since when a counter reaches value $k+1$ for vertices $u,v$, we know that $|N_G(u) \cap N_G(v)| > k$, the total time consumed is bounded by $O(\cutsize n^2)$.
Application of Lemma~\ref{lem:nmwcu-del-terms} takes $O(n^2 \log n)$ time per one operation, since we can sort terminals from one equivalence class
according to their sets of neighbours.
Thus all applications of Lemmata \ref{lem:nmwcu-ident-terms} and \ref{lem:nmwcu-del-terms} take $O(n^3(\cutsize+\log n))$ time in total.
The value of $\cutsize$ do not change in this step.
Therefore, we have the following recursive formula for time complexity as a function of the number of vertices of $G$:
\begin{equation}
T(n)\leq \max_{q+1\leq n'\leq n-q-1} \Big( O(2^{O(\cutsize^2 \log \cutsize)} n^3\log n) + T(n'+\cutsize) + T(\min(n-1,n-n'+q+q'))   \Big).
\end{equation}
Note that the function $p(t)=t^4\log t$ is convex, so it is easy to see that the maximum is attained either when $n'=q+q'-1$, or when $n'=n-q-1$. A straightforward inductive check of both of the ends proves that we have indeed the claimed bound on the complexity, i.e., $T(n) = O(2^{O(\cutsize^2 \log \cutsize)} n^4 \log n)$.

We conclude this section with a note that Lemma \ref{lem:node-no-separation} asserts that,
   if Step \ref{step:nmwcu-rekur} is not applicable, then for any set
   $Z \subseteq V(G) \setminus T$
  of size at most $\cutsize$, the graph $G \setminus Z$ contains at most
 $t := (2q+1)(2^\cutsize-1) + 2\cutsize + 1$ connected components containing a non-terminal, out of which at most one
has more than $q$ vertices not from $T$.

\subsection{Brute force approach}\label{sec:nmwcu-brute}
If the graph output by Step \ref{step:nmwcu-rekur} has small number of vertices
outside $T$,
   the algorithm may apply a straightforward brute-force approach
to the \bnmwcu{} problem. In this section we describe this method formally.
\begin{lemma}\label{lem:nmwcu-brute}
A correct output to a \bnmwcu{} instance $\instance_b = (G,T,\rel,\cutsize,\bterms)$
can be computed in $O(2^{O(\cutsize \log \cutsize)} n^2n_{\neg T}^\cutsize)$ time,
    where $n_{\neg T} = |V(G) \setminus T|$.
\end{lemma}
\begin{proof}
Simply, for each $\cP \in \mathbb{P}(\instance_b)$ (at most $2^{O(\cutsize \log \cutsize)}$ choices)
for each deletion set $X \subseteq V(G) \setminus T$
with  $|X| \leq \cutsize$ (at most $(\cutsize + 1) n_{\neg T}^\cutsize$ choices)
we verify in $O(n^2)$ time if $X$ is a solution to $(\instance_b,\cP)$.
\end{proof}
\begin{step}\label{step:nmwcu-brute}
If $|V(G) \setminus T| \leq qt + \cutsize$, apply
Lemma \ref{lem:nmwcu-brute} to find a correct output
to a \bnmwcu{} instance $\instance_b = (G,T,\rel,\cutsize,\bterms)$.
\end{step}
Recall that $q,t \leq 2^{O(\cutsize \log \cutsize)}$.
Thus, if Step \ref{step:nmwcu-brute} is applicable,
its running time is $O(2^{O(\cutsize^2 \log \cutsize)} n^2)$.

\subsection{High connectivity phase}\label{sec:nmwcu-high}

Assume we have a \bnmwcu{} instance $\instance_b = (G,T,\rel,\cutsize,\bterms)$
where Steps \ref{step:nmwcu-rekur} and \ref{step:nmwcu-brute} are not applicable.
In this section we show that high connectivity of $G$ makes the problem much easier.
To this end, fix $\cP=(X_b,E_b,\rel_b) \in \mathbb{P}(\instance_b)$.
We focus on finding the solution $\sol_\cP$; iterating through all the possible $\cP$ gives additional $2^{O(\cutsize\log \cutsize)}$ overhead to the running time. Recall that $G_\cP = G \cup E_b$.

Note that, if $|V(G) \setminus T|$ is too large for Step \ref{step:nmwcu-brute} to be applicable,
for any set $Z \subseteq V(G) \setminus T$ of size at most $\cutsize$, 
the bound on the number of connected components from Lemma \ref{lem:node-no-separation}
implies that there exists exactly one connected component of $G \setminus Z$ with
more than $q$ vertices outside $T$; denote its vertex set by $\bigcc(Z)$.

We now use Lemma \ref{lem:random} to get some more structure of the graph $G$.
\begin{definition}
Let $Z \subseteq V(G) \setminus T$ be a set of size at most $\cutsize$ and let $S \subseteq V(G) \setminus T$.
We say that $S$ {\em{interrogates}} $Z$ if the following holds:
\begin{enumerate}
\item $S \cap Z = \emptyset$;
\item for any connected component $C$ of $G \setminus Z$ with at most $q$ vertices outside $T$,
  all vertices of $C$ belong to $S \cup T$.
\end{enumerate}
\end{definition}

\begin{lemma}\label{lem:nmwcu-random-applicable}
Let $\randfamily$ be a family obtained
by the algorithm of Lemma \ref{lem:random} for universe $U=V(G) \setminus T$ and constants $a=qt$ and $b=\cutsize$,
Then, for any $Z \subseteq V(G) \setminus T$ with $|Z| \leq \cutsize$, there exists a set
$S \in \randfamily$ that interrogates $Z$.
\end{lemma}
\begin{proof}
Fix $Z \subseteq V(G) \setminus T$ with $|Z| \leq \cutsize$. Let $A$ be the union of vertex
sets of all connected components of $G \setminus Z$ that have at most $q$ vertices outside $T$;
by Lemma \ref{lem:node-no-separation}, $|A \setminus T| \leq qt$.
By Lemma \ref{lem:random}, as $|A \setminus T| \leq qt$ and $|Z| \leq \cutsize$,
there exists a set $S \in \randfamily$ that contains $A \setminus T$ and 
is disjoint with $Z$. By the construction of the set $A$, $S$ interrogates $Z$
and the lemma is proven.
\end{proof}

Note that, as $q,t = 2^{O(\cutsize \log \cutsize)}$,
     the family $\randfamily$ of Lemma \ref{lem:nmwcu-random-applicable}
is of size $2^{O(\cutsize^2 \log \cutsize)} \log n$ and can be computed
in $O(2^{O(\cutsize^2 \log \cutsize)} n \log n)$ time.
Therefore we may branch, guessing a set $S$ that interrogates
a solution $\sol_\cP=X_\cP$ we are looking for. Formally, we perform computations in each branch and return the minimum size solution from those obtained in the branches.

\begin{step}\label{step:nmwcu-guess-S}
Compute the family $\randfamily$ from Lemma \ref{lem:nmwcu-random-applicable}
and branch into $|\randfamily|$ subcases, indexed by sets $S \in \randfamily$.
In a branch $S$ we seek for a set $X_\cP$ with minimum possible $|X_\cP|$
that not only is a solution to $(\instance_b,\cP)$,
    but also is interrogated by $S$.
\end{step}
Lemma \ref{lem:nmwcu-random-applicable} verifies the correctness of the branching
of Step \ref{step:nmwcu-guess-S}; as discussed, the step is applied in
$O(2^{O(\cutsize^2 \log \cutsize)} n \log n)$ time and leads to
$O(2^{O(\cutsize^2 \log \cutsize)} \log n)$ subcases.

The following observation is crucial to for the final step.
\begin{lemma}\label{lem:nmwcu-finish}
Let $X_\cP$ be a minimum size set that is
a solution to $(\instance_b,\cP)$ interrogated by $S$.
Then there exists a set $T^\bigcc \subseteq T \cup (\bterms \setminus X_b)$
that is empty or contains all vertices of exactly
one equivalence class of $\rel_b$, such that $X_\cP = X_b \cup N_G(S(T^\bigcc))$, where
$S(T^\bigcc)$ is the union of vertex sets of all connected components of
$G[S \cup T \cup (\bterms \setminus X_b)]$ that contain a vertex
of $(T \cup (\bterms \setminus X_b)) \setminus T^\bigcc$.
\end{lemma}
\begin{proof}
Consider the graph $G_\cP \setminus X_\cP$ and let $\bigcc_\cP(X_\cP)$ be the vertex
set of the connected component of $G_\cP \setminus X_\cP$ that contain $\bigcc(X_\cP)$
(recall that $G_\cP$ is the graph $G$ with additional edges $E_p$; thus
 $\bigcc_\cP(X_\cP)$ may be significantly larger than $\bigcc(X_\cP)$).
As $X_\cP$ is a solution to $(\instance_b,\cP)$, we have $X_\cP \cap \bterms = X_b$.
Define $T^\bigcc = (T \cup (\bterms \setminus X_b)) \cap \bigcc_\cP(X_\cP)$;
as $X_\cP$ is a solution to $(\instance_b,\cP)$, $T^\bigcc$ is empty or
contains vertices of exactly one equivalence class or $\rel_b$.

Now let $C$ be the vertex set of a connected component of $G \setminus X_\cP$
that contains a vertex $v \in (T \cup (\bterms \setminus X_b)) \setminus T^\bigcc$.
Clearly, $v \notin \bigcc_\cP(X_\cP)$. As $S$ interrogates $X_\cP$, $\bigcc_\cP(X_\cP)$ contains $\bigcc(X_\cP)$ and
$X_\cP \cap (T \cup \bterms) = X_b \subseteq \bterms$,
we infer that $C$ is the vertex set of
a connected component of $G[S \cup T \cup (\bterms \setminus X_b)]$ as well.
As $v \in C$, $C$ is a connected component of $G[S(T^\bigcc)]$.
Since the choice of $C$ was arbitrary, we infer that $N_G(S(T^\bigcc)) \subseteq X_\cP$.
Denote $X_\cP' = X_b \cup N_G(S(T^\bigcc)) \subseteq X_\cP$.
To finish the proof of the lemma we need to show that $X_\cP'$ is a solution
to $(\instance_b,\cP)$ as well.

Clearly, $X_\cP' \cap (T \cup \bterms) = X_b$, as $N_G(S(T^\bigcc))\cap (T \cup (\bterms \setminus X_b))=\emptyset$ by the definition of $S(T^\bigcc)$. Moreover, as $X_\cP' \subseteq X_\cP$ and $X_\cP$ is a solution to $(\instance_b,\cP)$,
if $(u,v) \in \rel_b$ then $u$ and $v$ are in the same connected component
of $G_\cP \setminus X_\cP'$. We now show that for any $(u,v) \notin \rel_b$ the vertices $u$ and $v$ are in different connected components
of $G_\cP \setminus X_\cP'$.
Assume the contrary, and let $u,v \in T \cup (\bterms \setminus X_b)$ be such 
that $(u,v) \notin \rel_b$, $u$ and $v$ are in the same connected component of $G_\cP \setminus X_\cP'$
and that the distance between $u$ and $v$ in $G_\cP \setminus X_\cP'$ is minimum possible.
Let $P$ be a shortest path between $u$ and $v$ in $G_\cP \setminus X_\cP'$.

As $X_\cP$ is a solution to $(\instance_b,X_b)$, $u$ and $v$ are in different connected components
of $G_\cP \setminus X_\cP$; without loss of generality assume $v \notin \bigcc_\cP(X_\cP)$
and let $C$ be the vertex set of the connected component of $G \setminus X_\cP$ that contains
$v$. Clearly, since $(u,v) \notin \rel_b$, we have $u \notin C$.
Moreover, $v \in (T \cup (\bterms \setminus X_b)) \setminus T^\bigcc$ and $C$
is a connected component of $G[S(T^\bigcc)]$. Therefore $N_G(C) \subseteq X_\cP'$.
Since $u \notin C$, the path $P$ needs to go via an edge $v_1u_1 \in E_b$, where $v_1 \in C$
but $u_1 \notin C$. Note that then $u_1,v_1\in \bterms$. As $v_1 \in C$ and $X_\cP$ is a solution to $(\instance_b,\cP)$,
we have $(v,v_1) \in \rel_b$. As $E_b \subseteq \rel_b$, we have that
$(v,u_1) \in \rel_b$. As $(u,v) \notin \rel_b$, we infer than $(u_1,u) \notin \rel_b$,
but $u_1$ and $u$ are connected via a proper subpath of $P$ in $G_\cP \setminus X_\cP'$,
a contradiction to the choice of $u$, $v$ and $P$. This finishes the proof of the lemma.
\end{proof}

Lemma \ref{lem:nmwcu-finish} justifies the final step.
\begin{step}\label{step:nmwcu-finish}
In each branch, let $S$ be the corresponding guess, for each set $T^\bigcc$ that is empty or contains all vertices
of one equivalence class of $\rel_b$, check if $X_b \cup N_G(S(T^\bigcc))$
is a solution to $(\instance_b,\cP)$ that is interrogated by $S$.
For given $\cP$, output the smallest solution to $(\instance_b,\cP)$ found,
    or $\bot$ if no solution is found for any choice of $S$ and $T^\bigcc$.
\end{step}
Note that $\rel$ has at most $\ccrel = \cutsize^2 + \cutsize$ equivalence classes.
As $|\bterms| \leq 2\cutsize$, there are at most $1 + 3\cutsize + \cutsize^2$ choices
of the set $T^\bigcc$. For each $T^\bigcc$,
computing $X_b \cup N_G(S(T^\bigcc))$ and verifying if it is a solution
to $(\instance_b,\cP)$ interrogated by $S$ takes $O(n^2)$ time.
Therefore Step \ref{step:nmwcu-finish} takes $O(2^{O(\cutsize^2 \log \cutsize)} n^2 \log n)$
time for all subcases.

This finishes the description of fixed-parameter algorithm for \nmwcu.

\newcommand{\ap}{\gamma}
\newcommand{\nxt}{\mathrm{next}}

\section{Lower bound for big alphabet size}\label{sec:lower}

In this section we prove that the dependence on $s$ in the algorithm from Theorem~\ref{thm:ulc-main} is probably essential, even for the edge deletion case and in the classical setting, when every vertex has a full list of possible labels and the partial permutations on edges are required to be permutations. We define formally the problem as follows:

\defparproblemu{\eulc($\cutsize$)}{An undirected graph $G$, a finite alphabet $\Sigma$ of size $s$,
 an integer $\cutsize$, and: for each vertex $v \in V(G)$ a set $\phi_v \subseteq \Sigma$ and
 for each edge $e \in E(G)$ and each its endpoint $v$ a partial permutation $\psi_{e,v}$ of $\Sigma$, such that
if $e = uv$ then $\psi_{e,u} = \psi_{e,v}^{-1}$.}{$\cutsize$}{Does there exist
a set $F \subseteq E(G)$ of size at most $\cutsize$ and a function
$\Psi: V(G) \to \Sigma$ such that
for any $v \in V(G)$ we have $\Psi(v) \in \phi_v$ and for any $uv \in E(G) \setminus F$ we have $(\Psi(u),\Psi(v)) \in \psi_{uv,u}$?}

\begin{theorem}\label{thm:hardness}
\eulc($\cutsize$) is $W[1]$-hard, even in the restricted case, when $\phi_v=\Sigma$ for all $v\in V(G)$ and $\psi_{uv,u},\psi_{uv,v}$ are permutations for all $uv\in E(G)$.
\end{theorem}

Before we proceed to the proof, we state that this restricted case is not easier than the general one.

\begin{lemma}\label{lem:restricted}
There exists a polynomial time algorithm that, given an instance 
$$I=(G,\Sigma,\cutsize,(\phi_v)_{v \in V(G)},(\psi_{e,v})_{e \in E(G),v \in e})$$ of \eulc, outputs an equivalent instance
$$I'=(G',\Sigma',\cutsize',(\phi'_v)_{v \in V(G)},(\psi'_{e,v})_{e \in E(G),v \in e})$$ where $\cutsize'=\cutsize(\cutsize+2)$, $|\Sigma'|=|\Sigma|+\cutsize+2$, $\phi'_v=\Sigma'$ for all $v\in V(G)$ and $\psi'_{e,v}$ is a permutation for all $e \in E(G),v \in e$.
\end{lemma}
\begin{proof}
The graph $G'$ we are going to construct will be a multigraph, possibly with loops. Note that we can easily get rid of multiple edges and loops by subdividing every edge and loop, and for each subdivision preserving the constraint on one of the obtained edges while setting the constraint on the other edge to be identity.

We start with setting $\cutsize'=\cutsize(\cutsize+2)$ and $\Sigma'=\Sigma\cup \Gamma$, where $\Gamma=\{\ap_1,\ap_2,\ldots,\ap_{\cutsize+2}\}$ is the set of $\cutsize+2$ new symbols that do not belong to $\Sigma$. Now we construct the multigraph $G'$ as follows. Firstly, $V(G')=V(G)$. For every vertex $v\in V(G)$ we take an arbitrary permutation $\pi_v$ of $\Sigma'$ such that $\phi_v$ is exactly the set of labels that $\pi_v$ stabilizes; note that this is possible due to $k+2\geq 2$. We create $\cutsize'+1$ loops in $v$ with $\pi_v$ as the constraint. Then, for every edge $uv\in E(G)$ denote by $X_{uv,u}$ the set of labels from $\Sigma$ that do not have an image in $\psi_{uv,u}$, and similarly denote by $X_{uv,v}$ the set of labels from $\Sigma$ that do not have an image in $\psi_{uv,v}$. Let $\{\psi^i_{uv,u}\}_{i=1,\ldots,\cutsize+2}$ be an arbitrary family of permutations of $\Sigma'$, such that:
\begin{itemize}
\item each $\psi^i_{uv,u}$ extends $\psi_{uv,u}$;
\item each label $\alpha\in X_{uv,u}\cup \Gamma$ is mapped to pairwise different labels in $\psi^i_{uv,u}$ for $i=1,2,\ldots,\cutsize+2$;
\item each label $\beta\in X_{uv,v}\cup \Gamma$ is mapped to pairwise different labels in $\psi^i_{uv,v}=\left(\psi^i_{uv,u}\right)^{-1}$ for $i=1,2,\ldots,\cutsize+2$.
\end{itemize}
Observe that as $|\Gamma|=\cutsize+2$, one can find such family $\{\psi^i_{uv,u}\}_{i=1,\ldots,\cutsize+2}$ by enumerating $X_{uv,u}\cup \Gamma$ and $X_{uv,v}\cup \Gamma$ in arbitrary orders, fixing one bijection between them and shifting it cyclicly $\cutsize+1$ times. Between $u$ and $v$ we insert the set of $\cutsize+2$ edges $P_{uv}=\{uv^i\}_{i=1,2,\ldots,\cutsize+2}$, imposing the constraints $(\psi^i_{uv,u},\psi^i_{uv,v})$ on $uv^i$. Finally, we set $\phi'_v=\Sigma'$ for all $v\in V(G)$. This concludes the construction. We are left with a formal proof of the equivalence.

Assume first that there exists a set of edges $F\subseteq E(G)$, $|F|\leq \cutsize$, such that $G\setminus F$ admits a labeling $\Psi$ respecting constraints in the input instance $I$. Let $F'=\{e^i: e\in F\}$; note that $|F'|=(\cutsize+2)|F|\leq \cutsize'$. A direct check shows that $\Psi$ is also a correct labeling in $G'\setminus F'$, which proves that $F'$ is a solution to the instance $I'$.

Now assume that there exists a set of edges $F'\subseteq E(G')$, $|F'|\leq \cutsize'$, such that $G'\setminus F'$ admits a labeling $\Psi'$ respecting constraints in the output instance $I'$. Note that for each $v\in V(G)$ we have that $\Psi'(v)\in \phi_v$, as otherwise the set $F'$ would need to contain $\cutsize'+1$ loops at $v$. Let $F\subseteq E(G)$ be the set of edges $uv$ of $G$ such that $(\Psi'(u),\Psi'(v))\notin \psi_{uv,u}$. Clearly, $F$ is a solution in the instance $I$ as $\Psi'$ is a correct labeling of $G\setminus F$. It remains to prove that $|F|\leq \cutsize$. Assume otherwise, i.e., $|F|\geq\cutsize+1$.

Consider an edge $uv\in E(G)$ such that $(\Psi'(u),\Psi'(v))\notin \psi_{uv,u}$. We claim that $|F'\cap P_{uv}|\geq \cutsize+1$. If $\Psi'(u)$ belongs to the domain of $\psi_{uv,u}$, then all the constraints $\psi^i_{uv,u}$ map $\Psi'(u)$ to a label different that $\Psi'(v)$. Hence $P_{uv}\subseteq F'$ and the claim holds. Otherwise, $\Psi'(v)$ is mapped to $\cutsize+2$ different images in constraints $\psi^i_{uv,u}$, which means that at least $\cutsize+1$ of them must be different than $\Psi'(v)$. The corresponding edges have to be contained in $F'$ and the claim holds in this case as well. As $|F|\geq\cutsize+1$, we have that $|F'|\geq (\cutsize+1)^2=\cutsize'+1$, which is a contradiction.
\end{proof}

We are now ready to prove Theorem~\ref{thm:hardness}.

\renewcommand{\cutsize}{k'}

\begin{proof}[Proof of Theorem~\ref{thm:hardness}]
By Lemma~\ref{lem:restricted}, we may consider the general problem definition, where we allow lists in vertices and partial permutations as constraints imposed on edges.

We provide a parameterized reduction from the \mclique problem, which is known to be W[1]-hard~\cite{fellows-hermelin-rosamond-vialette-multicolored-hardness}.

\defparproblemu{\mclique}{An undirected graph $H$ with vertices partitioned into $k$ parts $V_0,V_1,\ldots,V_{k-1}$, such that $H$ does not contain edges connecting vertices from the same part $V_i$, for $i=0,1,\ldots,k-1$.}{$k$}{Is there a clique $C$ in $G$ of size $k$?}

Observe that by the assumption on the structure of $H$, the clique $C$ has to contain exactly one vertex from each part $V_i$. Moreover, by adding independent vertices we can assume that each part $V_i$ is of the same size $n$. In each part $V_i$ fix an arbitrary ordering of vertices $v^i_0,v^i_1,\ldots,v^i_{n-1}$.

Now, we are going to construct an instance $(G,\Sigma,\cutsize,(\phi_v)_{v \in V(G)},(\psi_{e,v})_{e \in E(G),v \in e})$ that is a YES instance of \eulc if and only if $H$ contains a clique of size $k$. As the construction will be performed in polynomial time and $\cutsize=k^2$, this gives the promised parameterized reduction.

We take $\Sigma=\{0,1,2,\ldots,n\}\times \{0,1,2,\ldots,n\}$, and let $\Lambda=\{0,1,\ldots,n-1\}\times \{0,1,\ldots,n-1\}\subseteq \Sigma$. For every part $V_i$ we create a cycle $C_i$ of length $kn$. Denote the vertices of $C_i$ by $u^i_0,u^i_1,\ldots,u^i_{kn-2},u^i_{kn-1}$ in the order of their appearance on the cycle. For every vertex $u^i_p$ let $\nxt(u^i_p)$ be the next vertex on the cycle $C_i$, i.e., $u^i_{p+1}$ if $p<kn-1$ and $u^i_{0}$ if $p=kn-1$. Let $e(u^i_p)$ be the edge connecting $u^i_p$ with $\nxt(u^i_p)$.

On every edge of the cycle $C_i$ we impose a constraint given by the permutation $\pi_0((a,b))=(a-1,b)$, where the numbers behave cyclicly modulo $n+1$. More precisely, the constraint on the edge $e(u^i_p)$ states that the label of $\nxt(u^i_p)$ has the first coordinate decremented by $1$ modulo $n+1$ comparing to the label of $u^i_p$. Now, for every $i\neq j$, $0\leq i,j<k$, we create an edge $u^i_{j\cdot n}u^j_{i\cdot n}$ with constraint given by the partial permutation $\sigma_{i,j}=\{((p,q),(q,p))\ |\ v^i_pv^j_q\in E(H)\}$. In other words, from the domain of the permutation $\sigma((a,b))=(b,a)$ we remove out all the pairs that contain $n+1$ and all the pairs that correspond to nonedges between $V_i$ and $V_j$. Finally, we set $\phi_v=\Lambda$ for every $v\in V(G)$ and $\cutsize=k^2$. This concludes the construction.

Let us firstly assume that $C$ is a clique of size $k$ in $H$ and let $\{v^i_{c_i}\}=V(C)\cap V_i$. We construct 
\begin{itemize}
\item a set of edges $F=\{e(u^i_{jn+c_i})\ |\ 0\leq i,j<k \}$;
\item a labeling $\Psi(u^i_p)=((c_i-p)\mod n,c_{q/n})$, where $u^i_q$ is the closest next vertex on the cycle that has lower index being a multiplicity of $n$, i.e., $q/n=\lceil p/n\rceil \mod n$.
\end{itemize}
Obviously, $|F|=\cutsize$. Let us check that $\Psi$ is a correct labeling of $G\setminus F$. Clearly, $\Psi(v)\in \Lambda=\phi_v$ for any $v\in V(G)$. Consider any edge $e(u^i_p)\notin F$. As $p\mod n\neq c_i$, we have that $\Psi(u^i_p)=(x,y)$ for some $x>0$ and $\Psi(\nxt(u^i_p))=(x-1,y)$; hence, these constraints are satisfied. Now consider any edge of the form $u^i_{j\cdot n}u^j_{i\cdot n}$ for $i\neq j$, $0\leq i,j<k$. By the construction of $\Psi$ we have that $\Psi(u^i_{j\cdot n})=
(c_i,c_j)$ and $\Psi(u^j_{i\cdot n})=(c_j,c_i)$. Recall that $C$ is a clique, so $v^i_{c_i}v^j_{c_j}\in E(H)$. Hence, $(c_i,c_j)$ lies in the domain of $\sigma_{ij}$ and the constraint imposed on this edge is satisfied as well.

Let us now assume that there is a set of edges $F\subseteq E(G)$, $|F|\leq \cutsize$, such that there exists a correct labeling $\Psi$ of $G\setminus F$. Firstly, we claim that for every $n$ consecutive edges of every cycle $C_i$, $F$ has to contain at least one of these edge. Otherwise there would be $n+1$ consecutive vertices $u^i_p,u^i_{p+1},\ldots,u^i_{p+n}$ such that edges $u^i_{p+i}u^i_{p+i+1}$ do not belong to $F$ for $i=0,1,\ldots,n-1$ (indices behave cyclicly). It follows that if $\Psi(u^i_p)=(\ell,d)$ for some $\ell<n$, then we would have $\Psi(u^i_{p+\ell+1})=(n,d)$, but $n$ is a forbidden value in a label for every vertex. As every cycle $C_i$ has length $kn$, it has to contain at least $k$ edges from $F$. As $\cutsize=k^2$, it has to contain exactly $k$ edges from $F$. We can use again the claim to infer that between every two subsequent edges from $F$ there must be exactly $n-1$ edges not from $F$, as otherwise there would be $n$ consecutive edges not belonging to $F$. Moreover, the same argumentation yields that the vertices of each interval on the cycle between the two subsequent edges from $F$ have to be labeled with $(n-1,d),(n-2,d),\ldots,(0,d)$, in this order, for some $d$ depending on the interval, but constant within. Hence, for every cycle $C_i$ we can find an integer $c_i\in \{0,1,\ldots,n-1\}$, such that $F=\{e(u^i_{jn+c_i})\ |\ 1\leq i,j<k \}$ and $\Psi(u^i_{jn})=(c_i,d^i_j)$ for all $j=0,1,\ldots,k-1$ and some numbers $d^i_j$. 

We are going to prove that vertices $v^i_{c_i}$ for $i=0,1\ldots,k-1$ induce a clique in $H$. Take parts $V_i,V_j$ for $i\neq j$ and examine the edge $u^i_{j\cdot n}u^j_{i\cdot n}$ with constraint $\sigma_{ij}$. As $\Psi(u^i_{j\cdot n})=(c_i,d^i_j)$, $\Psi(u^j_{i\cdot n})=(c_j,d^j_i)$, and $\sigma_{ij}$ swaps the elements of the pair, we find that $d^i_j=c_j$ and $d^j_i=c_i$. Moreover, $(c_i,c_j)$ is in the domain of $\sigma_{ij}$ if and only if $v^i_{c_i}v^j_{c_j}\in E(H)$. Therefore, $v^i_{c_i}$ and $v^j_{c_j}$ are adjacent for all $i\neq j$ and we are done.
\end{proof}

\section{Weights}\label{sec:weights}

We would like to note that using our technique we can solve a more general problem,
where the graph is edge-weighted (or vertex-weighted, in the vertex-deletion setting),
and the goal is, instead of minimizing the cardinality of the cutset, to find
a cutset of size at most $k$, having minimum sum of weights of the edges (or vertices) it contains.
For example for the problem considered in Section~\ref{sec:illustration}, the formal definition is as follows.

\defparproblemu{\textsc{Weighted} \eulc (\wulc)}{An undirected (multi)graph $G$ together with a weight function $\weight:E(G) \rightarrow \mathbb{R}_+$, a finite alphabet $\Sigma$ of size $s$,
 an integer $k$, and for each edge $e \in E(G)$ and each of its endpoints $v$ a permutation $\psi_{e,v}$ of $\Sigma$, such that
if $e = uv$ then $\psi_{e,u} = \psi_{e,v}^{-1}$.}{$k+s$}{What is the minimum weight of
a set $X \subseteq E(G)$ of size at most $k$ such that there exists a function $\Psi: V(G) \to \Sigma$ satisfying that
for any $uv \in E(G) \setminus X$ we have $(\Psi(u),\Psi(v)) \in \psi_{uv,u}$?}

Note, that now we have to reformulate the bordered problem definition as well,
because solutions to the bordered problem need to have a prescribed cardinality
in order to make them comparable.
Let us see it on the example of \wulc{}.

By $\mathbb{P}(\instance_b)$ we define the set of all pairs $\cP=(\Psi_b,k_b)$,
   such that $\Psi_b$ is a function from $\bterms$ to $\Sigma$ and $0 \le k_b \le k$.
We say that a set $X \subseteq E(G)$
is a solution to $(\instance_b,\cP)$ if
$|X| \leq k_b$, there exists a function $\Psi:V(G) \to \Sigma$ extending $\Psi_b$
such that for any $uv \in E(G) \setminus X$ we have $(\Psi(u),\Psi(v)) \in \psi_{uv,u}$.
and the sum of weights of edges in $X$ is minimum possible (comparing to all
other sets $X'$ satisfying the remaining constraints).
The border problem is defined as follows.

\defproblemoutput{\bwulc}{An \wulc instance $\instance = (G,\weight,\Sigma,k,(\psi_{e,v})_{e \in E(G), v \in e})$ with $G$ being connected,
and a set $\bterms \subseteq V(G)$ of size at most $4k$; denote $\instance_b = (\instance, \bterms)$.}{For each $\cP \in \mathbb{P}(\instance_b)$
output a solution $\sol_{\cP} = X_\cP$ to $(\instance_b,\cP)$ or
output $\sol_{\cP} = \bot$ if such a solution does not exist.}

Since, while finding a good separation, our algorithm does not perform any greedy choices, we almost leave the algorithm unchanged.
Similarly, the recursive understanding step in the node-deletion problems is not affected significantly by this change.
However, when solving a weighted problem, we need to be more careful in the final, high connectivity phase,
as the existence of weights limits our possibilities of being greedy.
In the following paragraphs we shortly argument that the high connectivity phases of the
algorithms presented in this paper can be adjusted to the weighted variants without greater effort.

\paragraph{\ulc.}
In the case of the \ulc problem, the high connectivity phase
remains almost unchanged; however, we need to argue that
all greedy steps used in this part of the algorithm are justified
also in the weighted case.
As in the case of the other problems, we start with
guessing an interrogating set that
is {\emph{(i)}} disjoint with the solution $Z$ we are looking for,
{\emph{(ii)}} contains all vertices of all small connected components of $G \setminus Z$,
{\emph{(iii)}} contains a large connected set adjacent to each
vertex of $Z$ that is adjacent to the large connected component of $G \setminus Z$.
The algorithm performs now two simple greedy steps: it checks whether $Z = \emptyset$
is a solution, and looks for not forsaken vertices without neighbours in $S$.
Both steps can be easily justified in the weighted case,
as we assume nonnegative weights and we require only $|X| \leq k_b$
in the border problem definition.
The crucial observation --- that there are only at most $s$ reasonable labelings
of the big stains (big connected components) of $S$ --- does not interfere with weights.
In the final bounded search tree algorithm we argue that there is a limited
number of vertices, out of which we need to delete at least one
(the Neighbourhood Branching Rule) or that there are only limited number
of ways a small stain can be handled (the Small Stains Rule).
Both argumentations are oblivious to weights;
note that this is also true in the second part of
Lemma \ref{lem:ulc-branch-smallcc}, where we argue about a greedy choice
of a labeling in case when the chosen labeling of the big stains
can be consistently extended to a connected component of $G \setminus N[\Psi]$.

\paragraph{\steinercut.}
In the case of the \steinercut problem, we need to slightly change the
final dynamic programming routine.
Recall that in the high connectivity phase for this problem we first guess a set of edges
$S$ that is {\emph{(i)}} disjoint with the solution $Z$ we are looking for, {\emph{(ii)}}
contains a spanning tree of each small connected component of $G \setminus Z$,
{\emph{(iii)}} contains a large spanning tree with an endpoint of an edge of the solution $Z$,
for each such endpoint contained in the large connected component of $G \setminus Z$.
Then we obtain a graph $H$ by contracting the edges of $S$ and identifying the images
of the large trees of $S$ (assumed in point {\emph{(iii)}}) into the core vertex $b$.
For each connected component $B_i'$ of $H \setminus b$ we have two choices: either
we delete all edges, or no edges from $B_i' \cup \{b\}$.
The choices between different components $B_i'$ are independent,
and we find the optimal solution via a simple dynamic programming routine.
In the weighted case we need to add to the dynamic programming table one more dimension
responsible for storing the cardinality
of the constructed cutset, and the value in the table $T$
is the minimum weight of a cutset of the prescribed cardinality.

\paragraph{\nmwcu.}
The simplicity of the high connectivity phase of the \nmwcushort{} algorithm
allows us to solve the weighted variant with almost no changes.
Recall that in this phase we first guess a set $S$ that is {\emph{(i)}} disjoint with
the solution $Z$ we are looking for, {\emph{(ii)}} covers all nonterminal vertices
of small connected components of $G \setminus Z$. Then we argue
that any inclusion-wise minimal solution chooses at most one equivalence
relation of $\rel_b$ to be the set of terminals contained in the big connected
component of $G \setminus Z$, and takes as the solution the neighbourhood of all connected
components of $G[S \cup T]$ that contain a terminal not contained in the selected
equivalence class.
The same argumentation holds in the case of nonnegative weights; note that
in the border problem we require $|X| \leq k_b$ (instead of maybe more natural $|X| = k_b$),
thus we may consider only inclusion-wise minimal solutions.

\bibliographystyle{abbrv}
\bibliography{rand-contractions}

\end{document}